\documentclass{article}
\usepackage{xcolor}
\usepackage{graphicx}
\usepackage{hyperref}
\usepackage{physics}
\usepackage{amsmath} 
\usepackage{amssymb}
\usepackage{commath}
\usepackage{enumitem}
\usepackage{amsthm}
\usepackage{mathtools}
\usepackage{relsize}
\usepackage{comment}
\usepackage{float} 
\usepackage{dsfont}
\usepackage{thm-restate}
\usepackage{cleveref}
\usepackage{tikz}
\usepackage{tikz-cd}
\usetikzlibrary{braids}
\usepackage{ytableau}
\usepackage{stackengine}
\usepackage{array}
\usepackage[utf8]{inputenc}
\usepackage{mathtools,xparse}
\usepackage{relsize}
\usepackage{tikz,tkz-tab}
\usepackage{authblk}
\usepackage{pgfplots}
\pgfplotsset{width=10cm,compat=1.9}
\DeclarePairedDelimiter{\oldnormaux}{\bracevert}{\bracevert}

\NewDocumentCommand{\oldnorm}{som}{%
  \IfBooleanTF{#1}
    {\oldnormaux*{#3}}
    {\IfNoValueTF{#2}
       {\oldnormaux*{\vphantom{dq}#3}}
       {\oldnormaux[#2]{#3}}%
    }%
}

\usetikzlibrary{decorations.markings}
\tikzset{
     pt/.style={insert path={node[scale=2]{.}}},
     dnup/.style={insert path={ [pt] .. controls +(0,1) and +(0,-1) .. +(#1,2) [pt]}},
     dndn/.style={insert path={ [pt] .. controls +(0,1) and +(0,1) .. +(#1,0) [pt]}},
     upup/.style={insert path={ [pt] .. controls +(0,-1) and +(0,-1) .. +(#1,0) [pt]}},
}
\tikzset{->-/.style={decoration={
  markings,
  mark=at position .5 with {\arrow{>}}},postaction={decorate}}}
\tikzset{-<-/.style={decoration={
  markings,
  mark=at position .5 with {\arrow{<}}},postaction={decorate}}}
  
\newcommand{\CC}{\mathbb{C}}
\newcommand{\RR}{\mathbb{R}}
\newcommand{\NN}{\mathbb{N}}
\newcommand{\QQ}{\mathbb{Q}}
\newcommand{\ZZ}{\mathbb{Z}}
\newcommand{\qa}{\mathfrak{q}}
\newcommand{\qam}{\mathfrak{q}^{-1}}
\newcommand{\Uq}{U_{\mathfrak{q}}\mathfrak{sl}_2}
\newcommand{\Uqu}{\overline{U}^\mathsf{H}_{\mathfrak{q}}\mathfrak{sl}_2}
\newcommand{\Hs}{\mathcal{H}}
\newcommand{\Hb}{\mathcal{H}_b}
\newcommand{\Hbb}{\mathcal{H}_{2b}}
\newcommand{\KK}{\mathsf{K}}
\newcommand{\KKm}{\mathsf{K}^{-1}}
\newcommand{\EE}{\mathsf{E}}
\newcommand{\FF}{\mathsf{F}}
\newcommand{\HH}{\mathsf{H}}
\newcommand{\CCC}{\mathsf{C}}
\newcommand{\BB}{\mathsf{B}}
\newcommand{\BBB}{\mathcal{B}}
\newcommand{\RRR}{\mathsf{R}}
\newcommand{\VV}{\mathcal{V}}
\newcommand{\TT}{\mathcal{T}}
\newcommand{\WW}{\mathcal{W}}
\newcommand{\XX}{\mathcal{X}}
\newcommand{\YY}{\mathcal{Y}}

\newcommand{\TL}{\mathsf{TL}_{\delta,N}}
\newcommand{\Blob}{\mathsf{B}_{\delta,y,N}}
\newcommand{\Blobb}{\mathsf{B}_{\delta,y,N-1}}
\newcommand{\twoBlob}{2\BB_{\delta,y_l,y_r,Y,N}}
\newcommand{\twoBlobm}{2\BB_{\delta,y_l,y_r,Y_m,N}}
\newcommand{\utwoBlob}{2\BB_{\delta,y_l,y_r,N}^{\mathrm{uni}}}
\newcommand{\Id}{\mathrm{Id}}
\newcommand{\End}{\mathrm{End}}
\newcommand{\Hom}{\mathrm{Hom}}
\newcommand{\Vir}{\mathrm{Vir}}

\newcommand{\ddd}{\mathrm{d}}
\newcommand{\qtr}{\mathrm{qtr}}

\newcommand{\Ker}{\mathrm{Ker}}
\newcommand{\Ress}{\mathrm{Res}}

\newcommand{\yyy}{y_{-}(\mu,N)}

\newcommand{\sqq}{\mathsmaller{\mathsmaller{\mathsmaller{\mathsmaller{\mathsmaller{\blacksquare}}}}}}
\newcommand{\usqq}{\mathsmaller{\mathsmaller{\mathsmaller{\mathsmaller{\mathsmaller{\square}}}}}}

\newcommand{\bbullet}{\mathlarger{\mathlarger{\mathlarger{\bullet}}}}
\newcommand{\ccirc}{\mathlarger{\mathlarger{\mathlarger{\circ}}}}

\declaretheorem[name=Proposition]{prop}

\newtheorem{lemma}{Lemma}
\newtheorem{conj}{Conjecture}

\theoremstyle{remark}

\title{$\Uq$-invariant non-compact boundary conditions for the XXZ spin chain}

\author[1,3]{Dmitry Chernyak}
\author[2]{Azat M. Gainutdinov}
\author[3,4]{Hubert Saleur}

\affil[1]{Laboratoire de Physique de l'{\'E}cole Normale Sup{\'e}rieure, ENS, Universit{\'e} PSL, CNRS, Sorbonne Universit{\'e}, Universit{\'e} de Paris, F-75005 Paris, France}
\affil[2]{Institut Denis Poisson, CNRS, Universit\'e de Tours, Parc de Grandmont, 37200 Tours, France}
\affil[3]{Institut de Physique Th{\'e}orique, Paris Saclay, CEA, CNRS, 91191 Gif-sur-Yvette, France}
\affil[4]{USC Physics and Astronomy Department, Los Angeles Ca 90089, USA}

\date{~}
\begin{document}
\maketitle

\begin{abstract}
We introduce new $\Uq$-invariant boundary conditions for the open XXZ spin chain. For generic values of $\qa$ we couple the bulk Hamiltonian to an infinite-dimensional Verma module on one or both boundaries of the spin chain, and for $\qa=e^{\frac{i\pi}{p}}$ a $2p$-th root of unity -- to its $p$-dimensional analogue. Both cases are parametrised by a continuous ``spin'' $\alpha\in\CC$. 

To motivate our construction, we first specialise to  $\qa=i$, where  we obtain a modified XX Hamiltonian with unrolled quantum group symmetry, whose spectrum and scaling limit is computed explicitly using free fermions. In the continuum, this model is identified with the $(\eta,\xi)$ ghost CFT on the upper-half plane with a continuum of conformally invariant boundary conditions on the real axis. The different sectors of the Hamiltonian are identified with irreducible Virasoro representations.

Going back to generic $\qa$ we investigate the algebraic properties of the underlying lattice algebras. We show that if $\qa^\alpha\notin\pm\qa^\ZZ$, the new boundary coupling provides a faithful representation of the blob algebra which is Schur-Weyl dual to $\Uq$. Then, modifying the boundary conditions on both the left and the right, we obtain a representation of the universal two-boundary Temperley-Lieb algebra. The generators and parameters of these representations are computed explicitly in terms of $\qa$ and $\alpha$. Finally, we conjecture the general form of the Schur-Weyl duality in this case.

This paper is the first in a series where we will study, at all values of the parameters, the spectrum and its continuum limit, the representation content of the relevant lattice algebras and the fusion properties of these new spin chains.
\end{abstract}

\newpage

\tableofcontents

\newpage

\section{Introduction}

The XXZ spin chain of length $N$ with open boundary conditions is governed by the Hamiltonian
\begin{equation}
\label{xxz}
H_{\mathrm{XXZ}}^{\text{open}}:=\frac{1}{2}\sum_{i=1}^{N-1}\left(\sigma^x_{i}\sigma^x_{i+1}+\sigma^y_{i}\sigma^y_{i+1}+\frac{\qa+\qa^{-1}}{2}\sigma^z_{i}\sigma^z_{i+1}\right)
\end{equation}
acting on the Hilbert space $\Hs:=(\CC^2)^{\otimes N}$, where $\qa$ is a complex parameter. If $\qa\neq 1$, the global $\mathsf{SU}(2)$ symmetry of the XXX spin chain ($\qa=1$) breaks down to $\mathsf{U}(1)$. However, we can recover the larger symmetry by deforming $\mathsf{SU}(2)$ into the $\Uq$ quantum group and changing the boundary conditions as
\begin{equation}
\label{XXZ}
\begin{aligned}
H_{\mathrm{XXZ}} & :=\frac{1}{2}\sum_{i=1}^{N-1} \left(\sigma^x_{i}\sigma^x_{i+1}+\sigma^y_{i}\sigma^y_{i+1}+\frac{\qa+\qa^{-1}}{2}\sigma^z_{i}\sigma^z_{i+1}\right)+\frac{\qa-\qa^{-1}}{4}(\sigma_N^z-\sigma_1^z)\\
& := H_{\mathrm{XXZ}}^{\text{open}}+\frac{\qa-\qa^{-1}}{4}(\sigma_N^z-\sigma_1^z)\,.
\end{aligned}
\end{equation}
The Hamiltonian $H_{\mathrm{XXZ}}$ is then $\Uq$-invariant \cite{PASQUIER1990523}. If moreover $|\qa|=1$ it is known to be critical \cite[Ch. 12]{Baxter:1982zz} (see also \cite{Alcaraz_1987}).

The aim of this paper is to introduce more general boundary conditions for $H_{\mathrm{XXZ}}$. Let us add a vector space $\VV$ on the leftmost boundary and take an arbitrary nearest-neighbour coupling $J\in\End(\VV\otimes\CC^2)$ on two leftmost sites to define a new local Hamiltonian
\begin{equation}
\label{hbform}
H_b=J+H_{\mathrm{XXZ}}
\end{equation}
acting on $\VV\otimes (\CC^2)^{\otimes N}$. To choose $\VV$ and $J$ our guiding principle is to preserve the $\Uq$ symmetry of the XXZ model. This means that $\VV$ must be a representation of $\Uq$ and $J$ some operator commuting with its action.

We will be most interested in ``non-compact'' boundary conditions, that is, infinite-dimensional $\VV$. There are many possible choices, but, for generic $\qa$, the most natural for us is to take Verma modules of $\Uq$. Denoted $\VV_\alpha$, these are infinite-dimensional highest-weight modules depending on a complex parameter $\alpha\in\CC$ (see Section~\ref{uqsec}). If $\qa=e^{\frac{i\pi}{p}}$ is a $2p$-th root of unity, we consider instead $p$-dimensional analogues of these representations, also parametrised by $\alpha\in\CC$ (see Section \ref{sec:Hb-root-q}). For simplicity they are equally denoted by $\VV_\alpha$.

To define a $\Uq$-invariant boundary condition on
\begin{equation}\label{intro:Hb-space}
 \Hb:=\VV_\alpha\otimes (\CC^2)^{\otimes N}
\end{equation}
we need to find an operator $J$ acting on $\VV_\alpha\otimes\CC^2$ and commuting with $\Uq$. It is known that generically\footnote{Exact conditions for this result to hold will be specified in Section \ref{construction}.} we have an isomorphism
\begin{equation}
\label{fusrule}
\VV_\alpha\otimes \CC^2 \cong \VV_{\alpha+1}\oplus \VV_{\alpha-1} \,,
\end{equation}
so the only operators acting on this space and commuting with $\Uq$ are the projectors on $\VV_{\alpha\pm1}$ which we denote $b_\pm$. Since $b_++b_-=1$, it is sufficient to take one of them, say $b:=b_+$. The most general Hamiltonian of the form \eqref{hbform} acting on $\Hb$ and compatible with $\Uq$ symmetry is then
\begin{equation}
\label{Hb}
H_b=-\mu b+H_{\mathrm{XXZ}}
\end{equation}
for some $\mu\in\CC$. We call $H_b$ the {\em one-boundary Hamiltonian}.

To change the boundary conditions on both the left and the right boundaries, we can proceed in the same fashion, taking the Hilbert space to be
\begin{equation}\label{intro:Hbb-space}
\Hbb:=\VV_{\alpha_l}\otimes\left(\CC^2\right)^{\otimes N}\otimes\VV_{\alpha_r}
\end{equation}
and defining projectors $b_l$ (resp.\ $b_r$) on the $\VV_{\alpha_l+1}$ (resp.\ $\VV_{\alpha_r+1}$) factors of $\VV_{\alpha_l}\otimes\CC^2$ (resp. $\CC^2\otimes\VV_{\alpha_r}$). The most general extension of $H_{\mathrm{XXZ}}$ to $\Hbb$ by $\Uq$-invariant two-site boundary terms is then
\begin{equation}
\label{Hbb}
H_{2b}=-\mu_l b_l+H_{\mathrm{XXZ}}-\mu_r b_r
\end{equation}
for some $\mu_l,\mu_r\in\CC$. We call $H_{2b}$ the \emph{two-boundary Hamiltonian}.

\medskip

The general goals of the series of works we start here are:
\begin{itemize}
\item[$i)$] To compute the spectra of $H_b$ and $H_{2b}$ and their conformal scaling limits.
\item[$ii)$] To understand the underlying \emph{lattice algebras} generated by the operators of nearest neighbour couplings.
\item[$iii)$] To define a consistent fusion procedure for these models.
\end{itemize}

\medskip

In this first paper we initiate our study by applying this program to the special case $\qa=i$, where the Hamiltonian \eqref{XXZ} becomes that of the XX model and admits a free-fermion representation, making it rather simple to obtain exact results for the spectrum. After constructing  in Section \ref{construction} the one-boundary and two-boundary Hamiltonians, $H_b$ and $H_{2b}$ for all $\qa$,  we  specialise to the value $\qa=i$ and diagonalise $H_b$ in Section \ref{specqi}. This is achieved by performing a Jordan-Wigner transformation, reducing the spectral problem to an $N\times N$ linear system which can be solved using plane waves (see Section \ref{planewaves}). The boundary conditions constrain the energies $\lambda$ of the corresponding modes to be the roots of the polynomial of degree $N$, i.e.\ they satisfy
\begin{equation*}
U_N(\lambda/2)+\mu U_{N-1}(\lambda/2)+(1-\mu y)U_{N-2}(\lambda/2)=0 \,,
\end{equation*}
where $U_n$ are Chebyshev polynomials of the second kind and  $y=\cot\frac{\pi\alpha}{2}$. We then show in Section~\ref{compspec} that these roots, as well as the momenta of the associated plane waves, are all real in a certain domain of the parameter space and that, together with $\Uq$ symmetry, they generate a complete basis of eigenstates of~$H_b$. Finally, in Section~\ref{sclim} we compute the large-$N$ limit of the spectrum and express the scaling limit of the partition function in terms of Virasoro characters of generic conformal weights depending only on $\alpha$. The continuum limit of the spin chain is then identified in Section \ref{sfsec} with the $(\eta,\xi)$ ghost system with the action
\begin{equation}
\label{exaction}
S[\eta,\xi,\bar{\eta},\bar{\xi}]=\frac{1}{2\pi}\int\ddd^2z\left(\eta\bar{\partial}\xi+\bar{\eta}\partial\bar{\xi}\right)
\end{equation}
defined on the upper-half plane with some specific $\alpha$-dependent conformal boundary conditions on the real axis \cite{Kausch:1995py, Kausch_2000}.

In Section \ref{spech2b} the above procedure for the spectrum analysis and the scaling limit is extended to the two-boundary Hamiltonian $H_{2b}$, also for $\qa=i$, by following the same steps.

In the second half of the paper, we present a general algebraic formalism, valid for any $\qa$, relating the Hamiltonians $H_b$ and $H_{2b}$ to some well-known lattice algebras. It is known \cite{Jimbo1986AQO} that the local Hamiltonian densities $(e_i)_{1\leq i\leq N-1}$ of $H_{\mathrm{XXZ}}$ satisfy the defining relations of the Temperley-Lieb (TL) algebra
\begin{equation}
\label{TLrel}
e_i^2=(\qa+\qam)e_i\,,\qquad e_ie_{i\pm1}e_i=e_i\,,\qquad [e_i,e_j]=0\ ,\quad |i-j|\geq 2\,,
\end{equation}
faithfully (that is, without any additional relations between them) and that their action on $\Hs=(\CC^2)^{\otimes N}$ and that of $\Uq$ centralise each other (see Section \ref{XXZalg}). This result is sometimes called quantum Schur-Weyl duality \cite{goodman, martin1992commutants, Martin1992ONSD}. We extend it to $\Hb$ defined in~\eqref{intro:Hb-space} by showing that the additional boundary operator $b$ we construct satisfies
\begin{equation}
\label{blobrule}
b^2=b\,,\qquad e_1 b e_1 = y e_1\,, \qquad [b,e_i]=0 \quad \text{for} \quad 2\leq i\leq N-1
\end{equation}
with
\begin{equation*}
y=\frac{\qa^{\alpha+1}-\qa^{-\alpha-1}}{\qa^{\alpha}-\qa^{-\alpha}}
\end{equation*}
and thus gives a representation of the so-called blob algebra \cite{Martin:1993jka}. We then prove that it is faithful and that $\Uq$ and the blob algebra are mutual centralisers on $\Hb$ (see Proposition~\ref{blobrepqgen} in Section~\ref{blobsec}). We equally consider the root of unity case $\qa=e^{\frac{i\pi}{p}}$ where we place at the boundary $p$-dimensional representations parametrised by $\alpha\in\CC$. If the weight $\alpha$ is generic we prove in Proposition~\ref{blobreproot} that $\Hb$ still provides a faithful representation of the blob algebra in this case, however now the role of the quantum group $\Uq$ is played by its unrolled version $\Uqu$ with an additional generator $\HH=\log(\KK)$ where $\KK$ is the standard Cartan generator of $\Uq$. The unrolled quantum group was widely used in knot theory~\cite{Oh,GPT}\footnote{The important aspect used in the mathematics literature is that the unrolled quantum group admits a universal $\RRR$-matrix at roots of unity and therefore we have a non-trivial braiding for the generic $p$-dimensional representations $\VV_\alpha$. This braiding will be important in our analysis of the two-boundary spin-chains.} but its generic $\qa$ version has also been known in the spin chain literature since long ago (see for example \cite{Kulre, KSk} where the spin projection operator $S^z$ plays the role of $\HH$ and $\qa^{S^z}$ that of $\KK$). We also briefly discuss the situation with $p$-dimensional cyclic representations at the boundary (see Remark at the end of Section~\ref{blobsec}).

A similar study is done for the two-boundary system $\Hbb$ defined in~\eqref{intro:Hbb-space}, for both generic $\qa$ and the root of unity cases. The relevant lattice algebra is now the two-boundary Temperley-Lieb algebra \cite{twob2004, de_Gier_2009}. In our case however, we need a slightly generalised version of it (a central extension) called the \textit{universal} two-boundary Temperley-Lieb algebra, which we introduce in Section \ref{twoboundarysec}. The result is that $\Hbb$ carries a representation of the universal two-boundary TL algebra, and we find explicit expressions for the generators and parameters. The main technical difficulty is to compute the central element $Y$ corresponding, in the lattice algebra language, to the (universal) weight of a closed loop decorated by both the left and right boundary operators. By using diagrammatic calculus tools, we manage to find its explicit expression in terms of the Casimir element of $\Uq$. For generic $\qa$ this is established in Proposition~\ref{repthtwo} in Section~\ref{Ycomp}, and for the root of unity cases in Proposition~\ref{repthtworoot}. Decomposing $\Hbb$ into $Y$-eigenspaces, we recover representations of the ``standard" two-boundary TL algebra with a fixed (scalar) value of $Y$ in each eigenspace (see Propositions~\ref{2bpropqgen} and~\ref{2bpropqroot} in Section \ref{twoblobdec} for $\qa$ generic and a root of unity respectively). However, due to the non-generic values $Y$ takes in these sectors that make the representation theory non-semisimple, so far we can only conjecture the Schur-Weyl duality in this case (see Conjectures~\ref{conj1} and~\ref{conj2} in Section \ref{conjsec} for $\qa$ generic and a root of unity respectively).

The irreducible representations of these lattice algebras -- also known as standard modules -- can be identified with sectors of $H_b$ and $H_{2b}$. For the one-boundary system, it was conjectured in \cite{Nichols2006, Jacobsen_2008} from various symmetry considerations and numerical evidence \cite{PASQUIER1990523, Grimm:1990gg, Ritt1990} that at criticality ($|\qa|=1$) the scaling limit of the spectrum in each sector should correspond to some specific representation of the Virasoro algebra. The results of this paper are sufficient to prove this claim for $\qa=i$ (Section~\ref{sclim}). As the computation of the spectrum and of its scaling limit for arbitrary $\qa$ requires the introduction of the heavy machinery of boundary Bethe ansatz, the general proof for all $|\qa|=1$ will be given in a separate paper. For the two-boundary system, a similar conjecture (also for $|\qa|=1$) relating the scaling limit of the spectrum of $H_{2b}$ in each of its irreducible sectors to Virasoro modules also exists \cite{Dubail_2009}. For $\qa=i$, we were able to compute the scaling limit of $H_{2b}$ in those sectors (Section~\ref{spech2b}) but it still remains to identify them with known representations of the two-boundary TL algebra, which is made difficult by the non-generic values of the $Y$ parameter we have to consider. Nevertheless, we have a strong conjecture on what these representations should be (Conjecture \ref{conj2}) and it is consistent with the prediction from \cite{Dubail_2009}.

Even though our results on the spectra of the one-boundary and two-boundary Hamiltonians at $\qa=i$ match those already stated in previous works \cite{Nichols2006, magic} and the underlying algebraic structures (for generic and non-generic values of the parameters) are mostly known \cite{deG2005, Nichols22006, de_Gier_2009}, we want to emphasise that our approach is very different conceptually. The above-cited papers use another type of spin-chain representation for the blob and two-boundary algebras : they have no extra degrees of freedom at the boundary and so the boundary operators $b_{l/r}$ act only on the leftmost and rightmost $\CC^2$ sites, i.e. they are locally represented by $2\times 2$ matrices with scalar entries. This contrasts with our spin chains where $b_{l/r}$ are locally represented by $2\times 2$ matrices whose entries are \emph{infinite-dimensional matrices} acting on the additional spaces $\VV_{\alpha_{l/r}}$. Though the resulting Hamiltonians -- called non-diagonal XXZ models -- used in these works are also of the form~\eqref{Hbb} and are known to be integrable \cite{bk1993, bk1994} (see also the discussion in~\cite[Sec.\,3.4]{Dubail_2009}), the boundary operators break $\Uq$ symmetry and even the standard $\mathsf{U}(1)$ symmetry in this case. This symmetry breaking makes it impossible to carry out neither the Algebraic Bethe Ansatz procedure directly due to absence of the standard reference state, nor the free fermion method at $\qa=i$ because the Jordan-Wigner transformed Hamiltonian is not quadratic. Instead, a coordinate Bethe ansatz in the reduced state basis \cite{Gier_2004} or, for the one-boundary system, a non-trivial mapping to an equivalent $\mathsf{U}(1)$-invariant spin chain \cite{Nichols2006, deG2005} had to be used in these works.

In the context of the above-mentioned non-diagonal XXZ models, it is also known that one can perform an intricate gauge transformation to construct a suitable reference state \cite{cao2002exact}, or derive some functional relations \cite{nepomechie2002functional, nep2003}, or even use the Modified Algebraic Bethe Ansatz~\cite{PBel} to obtain the Bethe ansatz equations. Moreover, a general algebraic framework to clarify the integrability structure of the non-diagonal XXZ models based on the so-called $\qa$-Onsager algebra was developed in \cite{Baseilhac_2007}. Nevertheless, in all these works, the lack of a sufficiently strong spin chain symmetry greatly complicates the computations and the analysis of the spectrum.

Considering more ``canonical'' spin-chain representations~\eqref{intro:Hb-space} and~\eqref{intro:Hbb-space} of the relevant lattice algebras in the sense that they preserve $\Uq$ symmetry -- even by adding additional degrees of freedom and sacrificing finite-dimensionality -- enables us to circumvent the above problems and to provide a more rigid and arguably simple formalism to diagonalise these one-boundary and two-boundary systems as well as to organise their spectra  into sectors which are standard modules over the underlying lattice algebras. For example, in the two-boundary case and for generic $\qa$, Proposition~\ref{2bpropqgen} together with Conjecture~\ref{conj1} suggest that a \emph{single} spin chain~\eqref{intro:Hbb-space} contains an infinite discrete series of non-diagonal XXZ models, where the integrable boundary conditions are certain functions on $\alpha_{l/r}$, $\mu_{l/r}$ and $Y$ (see the discussion at the end of Section~\ref{conclusion}).

It is worth mentioning that the idea to use extra degrees of freedom to simplify the diagonalisation problem was previously applied to compute the spectrum of the open XX spin chain with arbitrary boundary fields \cite{xx1999} and the large-$N$ expansion of the ground state for some choices of parameters \cite{xx2000}. In these papers, auxiliary spectator $1/2$-spins were added at the boundary, so as to obtain a free-fermion system on $N+2$ sites. Even if our approach may resemble what was done in these works, the Hamiltonian we consider has completely different (and stronger) algebraic properties and it involves `continuous spins', yielding simpler expressions and enabling us to find the complete scaling limit for all values of the continuous parameters.

Finally, let us note that this philosophy was also used in various other related contexts, most notably the boundary sine-Gordon model \cite{blz1999, BASEILHAC2003491, bas2003}. The idea is always to try to replace complicated boundary interactions by some equivalent coupling to an additional degree of freedom which preserves the relevant symmetries of the model and makes it easier to derive exact results (see for example~\cite{blz1999}). Here, we apply the same technique to a lattice model and with a rather modest symmetry group : $\Uq$. We will see however -- in this paper and the following ones -- that this algebra as well as its various Schur-Weyl duals are sufficient to understand the structure of the continuum limit and the integrability properties of our spin chains. More generally, if we think of spin chains as discretisations of QFTs, symmetry-preserving boundary couplings to additional degrees of freedom should have natural continuum counterparts. This could provide new insights into the study of boundary CFTs and related physical applications, such as quantum impurity problems \cite{saleur1998lectures} and out-of-equilibrium quantum systems \cite{eq2016}.

\section*{Notations}

\begin{itemize}
\item $N$~: Length of bulk of spin chains
\item $\sigma^x=\begin{pmatrix} 0 & 1\\ 1 & 0\end{pmatrix},~\sigma^y=\begin{pmatrix} 0 & -i\\ i & 0\end{pmatrix},~\sigma^z=\begin{pmatrix} 1 & 0\\ 0 & -1\end{pmatrix}$~: Pauli matrices
\item $\sigma^+:=\frac{1}{2}\left(\sigma^x+i\sigma^y\right)=\begin{pmatrix} 0 & 1\\ 0 & 0\end{pmatrix},~\sigma^-:=\frac{1}{2}\left(\sigma^x-i\sigma^y\right)=\begin{pmatrix} 0 & 0\\ 1 & 0\end{pmatrix}$
\item $\qa=e^\frac{i\pi}{p}$~: Deformation parameter of the XXZ spin chain with $p\in\NN\backslash\{0,1\}$ if $\qa$ is a $2p$-th root of unity
\item $[x]_\qa:=\frac{\qa^x-\qa^{-x}}{\qa-\qam}$~: $\qa$-deformed numbers
\item $\{x\}:=\qa^x-\qa^{-x}$
\item $\Uq$~: Quantum group, a $\qa$-deformation of $\mathsf{SU}(2)$
\item $\Uqu$~: For $\qa$ a $2p$-th root of unity, the restricted unrolled quantum group constructed from $\Uq$
\item $\EE,\FF,\KK,\KK^{-1}$~: Generators of $\Uq$. $\Uqu$ has an additional generator $\HH$
\item $\CCC$~: Casimir operator of $\Uq$ and $\Uqu$
\item $\VV_\alpha$~: Infinite-dimensional Verma module of $\Uq$ of highest-weight ${\alpha-1\in\CC}$ or $p$-dimensional representation of $\Uqu$ of highest-weight ${\alpha+p-1\in\CC}$ if $\qa$ is a $2p$-th root of unity
\item $b$~: $\Uq$-invariant boundary coupling operator, or the blob generator
\item $b_{l/r}$~: Left/right $\Uq$-invariant boundary coupling operators, or left/right blobs ($b_l:=b$)
\item $\Hs:=\left(\CC^2\right)^{\otimes N}$~: Hilbert space of the XXZ spin chain on $N$ sites
\item $H_{\mathrm{XXZ}}$~: $\Uq$-invariant open XXZ Hamiltonian
\item $\Hb:=\VV_\alpha\otimes\left(\CC^2\right)^{\otimes N}$~: Hilbert space of the one-boundary spin chain on $N$ sites
\item $H_b:=-\mu b+H_{\mathrm{XXZ}}$~: $\Uq$-invariant one-boundary Hamiltonian with coupling constant $\mu\in\RR$
\item $\Hb:=\VV_{\alpha_l}\otimes\left(\CC^2\right)^{\otimes N}\otimes\VV_{\alpha_r}$~: Hilbert space of the two-boundary spin chain on $N$ sites
\item $H_b:=-\mu_l b_l+H_{\mathrm{XXZ}}-\mu_rb_r$~: $\Uq$-invariant two-boundary Hamiltonian with couplings $\mu_l,\mu_r\in\RR$
\item $H_{\mathrm{XX}}$~: $\Uqu$-invariant open XX Hamiltonian ($H_{\mathrm{XXZ}}$ at $\qa=i$)
\item $\theta_k^\dagger/\theta_k$~: Fermionic creation/annihilation operators of mode $1\leq k\leq N$
\item $\lambda_k$ : Energy of mode $1\leq k\leq N$
\item $(U_n)_{n\geq 0}$~: Chebyshev polynomials of the second kind
\item $]\cdot,\cdot[$~: Interval of $\RR$ with endpoints excluded
\item $c=-2$~: Central charge
\item $h_{r,s}:=\frac{(2r-s)^2-1}{8}$~: Conformal weights corresponding to Kac labels $(r,s)$
\item $\delta:=\qa+\qam$~: Loop weight
\item $\TL$~: Temperley-Lieb (TL) algebra on $N$ sites with loop weight $\delta$
\item $(e_i)_{1\leq i\leq N-1}$~: Generators of the Temperley-Lieb algebra
\item $\left(\mathcal{T}_j\right)_{0\leq j\leq N/2}$~: Standard $\TL$-modules ($j$ half-integer if $N$ is odd)
\item $y:=\frac{[\alpha+1]_\qa}{[\alpha]_\qa}$~: Blob weight
\item $\Blob$~: Blob algebra on $N$ sites with loop weight $\delta$ and blob weight $y$
\item $\left(\WW_j\right)_{-N/2\leq j\leq N/2}$~: Standard $\Blob$-modules ($j$ half-integer if $N$ is odd)
\item $y_{l/r}:=\frac{[\alpha_{l/r}+1]_\qa}{[\alpha_{l/r}]_\qa}$~: Left/right blob weights
\item $\twoBlob$~: Two-boundary Temperley-Lieb algebra on $N$ sites with loop weight $\delta$, left/right blob weights $y_l/y_r$ and two-blob weight $Y$
\item $\utwoBlob$~: Universal two-boundary Temperley-Lieb algebra on $N$ sites with loop weight $\delta$ and left/right blob weights $y_l/y_r$
\end{itemize}

\section{Construction of Hamiltonians $H_b$ and $H_{2b}$}
\label{construction}

Our first task is to construct the one- and two- boundary Hamiltonians $H_b$ and $H_{2b}$ explicitly from $\Uq$ symmetry. The procedure is slightly different for $\qa$ a root of unity or $\qa$ generic, so these two cases have to be dealt with separately.

\subsection{$\Uq$ symmetry and its representations}
\label{uqsec}

The algebra $\Uq$ \cite{Drinfeld:1985rx, Jimbo:1985zk} (see also \cite[Ch.\,6.4]{qgroups} and \cite[Ch.\,VI-VII]{kassel}) is defined by generators $\EE$, $\FF$, $\KK$ and $\KK^{-1}$ and relations
\begin{equation*}
\KK\EE\KK^{-1}=\qa^2\EE\,,\quad \KK\FF\KK^{-1}=\qa^{-2}\FF\,,\quad [\EE,\FF]=\frac{\KK-\KK^{-1}}{\qa-\qam}\,,\quad \KK\KK^{-1}=\KK^{-1}\KK=\mathsf{1}\,.
\end{equation*}
It is a $\qa$-deformation of the universal enveloping algebra of the Lie algebra $\mathfrak{sl}_2$, in the sense that we recover the commutation relations of the $\mathfrak{sl}_2$ triple $(\EE,\FF,\HH)$ in the limit $\qa\to 1$ with $\KK=\qa^\HH$. It is important for defining the action on tensor products of representations that this algebra admits the coproduct
\begin{equation}
\label{coproduct}
\Delta(\EE)=\mathsf{1}\otimes\EE+\EE\otimes\KK\,,\quad \Delta(\FF)=\KK^{-1}\otimes\FF+\FF\otimes\mathsf{1}\,,\quad \Delta(\KK^{\pm 1})=\KK^{\pm 1}\otimes\KK^{\pm 1}\ .
\end{equation}

As $\mathfrak{sl}_2$, $\Uq$ admits ($2j+1$)-dimensional spin-$j$ representations for all $j\in\frac{1}{2}\NN$. For our purposes we will  need the fundamental spin-$\frac{1}{2}$ representation $\CC^2$ where the action of the generators is given by
\begin{equation}
\label{funduq}
\EE_{\CC^2}=\sigma^+\,,\qquad\FF_{\CC^2}=\sigma^-\,,\qquad \KK_{\CC^2}^{\pm 1}=\qa^{\pm\sigma^z}\,.
\end{equation}
To extend it to the $N$ sites of the Hilbert space $\Hs:=(\CC^2)^{\otimes N}$ we apply the coproduct~\eqref{coproduct}
$N-1$ times (recall that the coproduct is coassociative, and so the result does not depend on the order of its application). The Hamiltonian $H_{\mathrm{XXZ}}$ then commutes with the action of all the $\Uq$-generators \cite{PASQUIER1990523}.

Finally, let us introduce the Verma modules $\VV_\alpha$ \cite[Ch.\,VI.3]{kassel} that we shall need to define our modified boundary conditions. For all $\alpha\in\CC$ they are given in a basis $\VV_\alpha:=\bigoplus_{0\leq n}\CC\ket{n}$ by
\begin{equation}
\label{alpharepgen}
\begin{aligned}
& \EE_{\VV_\alpha}\ket{n}=[n]_\qa[\alpha-n]_\qa\ket{n-1} \,, \\
& \FF_{\VV_\alpha}\ket{n}=\ket{n+1} \,, \\
& \KK_{\VV_\alpha}^{\pm 1}\ket{n}=\qa^{\pm(\alpha-1-2n)}\ket{n} \\
\end{aligned}
\end{equation}
for all $n\geq 0$, with $\ket{-1}=0$, and where
\begin{equation*}
[x]_\qa:=\frac{\qa^x-\qa^{-x}}{\qa-\qam}=\frac{\{x\}}{\{1\}}\,,\qquad \{x\}:=\qa^x-\qa^{-x}\,.
\end{equation*}
The basis vectors $\ket{n}$ diagonalise $\KK$ and their $\KK$-eigenvalue $\qa^{\alpha-1-2n}$ is called the weight. The vector $\ket{0}$  is annihilated by the raising operator $\EE$ and is thus called the highest-weight vector.

Note that the modules $\VV_\alpha$ and $\VV_{\alpha+2p}$, where we parametrise $\qa=e^{\frac{i\pi}{p}}$ for some complex $p\in\CC$, are the same. Therefore, technically we should define it only for $\alpha\in\CC/2p\ZZ$ or work with the exponentiated parameter $\qa^\alpha$. However, in all practical applications, we will only encounter modules of the form $\VV_{\alpha+j}$ with $j\in\ZZ$, so as long as $\qa$ is not a root of unity (that is $p\notin\QQ$), $\alpha+j\neq\alpha+j'$ mod $2p$ for all integers $j\neq j'$ and this will not be an issue. Therefore, we will often slightly abuse notation and lift $\alpha$ to $\CC$.

The Verma modules $\VV_\alpha$ are not always irreducible. For example, for certain values of $\alpha$, one may have $\EE_{\VV_\alpha}\ket{m}=0$ for some $m\geq 1$ giving rise to a non-trivial stable subspace $\bigoplus_{m\geq n}\CC\ket{n}$. One can actually show \cite[Ch.\,VI]{kassel} that this is the only way $\VV_\alpha$ can become non-irreducible. Thus, $\VV_\alpha$ is irreducible if and only if $[\alpha-n]_\qa\neq 0$ for all $n\in\NN^{*}$ or, in other words, if and only if $\qa^\alpha\neq\pm\qa^n$ for all $n\in\NN^{*}$ ($[n]_\qa$ can never vanish if $\qa$ is not a root of unity). If that is the case, $\VV_\alpha$ is also unique, meaning any $\Uq$-module generated from a highest-weight vector of weight $\qa^{\alpha-1}$ is isomorphic to $\VV_\alpha$.

\subsection{Generic $\qa$}\label{sec:Hb-gen-q}

Let us first assume that $\qa$ is generic, that is, not a root of unity. 

As was explained in the introduction, the boundary Hamiltonian $H_b$ is obtained by tensoring the standard spin chain $\Hs=(\CC^2)^{\otimes N}$ with the Verma module $\VV_\alpha$ and adding a new $\Uq$-invariant boundary term $-\mu b$ acting on the two leftmost sites $\VV_\alpha\otimes\CC^2$ of the new Hilbert space $\Hb$. 

To construct the most general such operator, one has to understand the decomposition of $\VV_\alpha\otimes\CC^2$ into irreducible $\Uq$-modules. First, let us assume that  $\qa^\alpha\neq\pm\qa^n$ for all $n\in\NN^{*}$ so that $\VV_\alpha$ is irreducible. Next, if $\qa^\alpha\notin \pm \qa^\NN$, we have $\VV_\alpha\otimes\CC^2\cong\VV_{\alpha+1}\oplus \VV_{\alpha-1}$. Indeed, using the coproduct \eqref{coproduct} it is easy to check that 
\begin{equation}
\label{eq-hwv-Verma}
\ket{v_+}:=\ket{0}\otimes\ket{\uparrow} \qquad \text{and} 
\qquad \ket{v_-}:=\ket{1}\otimes\ket{\uparrow}-\qa[\alpha-1]_\qa\ket{0}\otimes\ket{\downarrow}
\end{equation}
are highest-weight vectors. By the uniqueness result for the Verma modules discussed above, these two vectors  generate  submodules isomorphic to $\VV_{\alpha+1}$ and $\VV_{\alpha-1}$ respectively, and these submodules are irreducible due to the assumption on $\qa^{\alpha}$.  Then, comparing the dimensions of the weight spaces we conclude that the tensor product $\VV_\alpha\otimes\CC^2$ is spanned by these two irreducible submodules and therefore consists of their direct sum
\begin{equation*}
\VV_\alpha\otimes\CC^2\cong\VV_{\alpha+1}\oplus \VV_{\alpha-1}\,.
\end{equation*}
With some more work, one can actually show that this tensor product $\Uq$ decomposition remains true even if $\VV_{\alpha\pm 1}$ are non-irreducible, as long as $\qa^\alpha\neq\pm1$.\footnote{Concretely, if $\qa^\alpha=\pm\qa^n$ for some $n\geq 1$, the two vectors~\eqref{eq-hwv-Verma} are still highest-weight and generate the submodules $\VV_{\alpha+1}$ and $\VV_{\alpha-1}$. The only non-trivial stable subspaces of $\VV_{\alpha\pm 1}$ are isomorphic to $\VV_{\alpha-2n\mp 1}$, and so we must have $\VV_{\alpha+1}\cap\VV_{\alpha-1}=\{0\}$. Comparing the dimensions of the weight spaces we obtain $\VV_\alpha\otimes\CC^2\cong\VV_{\alpha+1}\oplus \VV_{\alpha-1}$. Note that for $\qa^\alpha=\pm 1$ this is no longer the case, in particular $\FF\ket{v_+}=\ket{v_-}$ so $\VV_{\alpha-1}\subset\VV_{\alpha+1}$. We give more details about this case in Remark at the end of this section.}

\medskip

We now turn to the construction of the boundary Hamiltonians. Let us begin with the assumption $\qa^\alpha\notin\pm\qa^{\NN}$ (we give comments on the case $\qa^\alpha\in \pm \qa^\NN$ at the end of this section). In this case, the only $\Uq$-invariant operators acting on $\VV_\alpha\otimes\CC^2$ are linear combinations of the projectors $b_\pm$ on $\VV_{\alpha\pm 1}$. Obviously, $b_++b_-=1$ so without loss of generality we can choose $b:=b_+$ as our boundary operator (the other choice would just shift the Hamiltonian by a constant term and change the sign of $\mu$).

To find the explicit expression of the projector $b_+$, we study the action of the $\Uq$ Casimir element~$\CCC$ defined by
\begin{equation}
\label{casimir}
\CCC:=\{1\}^2\FF\EE+\qa\KK+\qam\KKm=\{1\}^2\EE\FF+\qam\KK+\qa\KKm\,.
\end{equation}
It is central in $\Uq$ and, moreover, all possible operators of $\End(\VV_\alpha\otimes\CC^2)$ commuting with the $\Uq$-action are of the form
\begin{equation*}
A=\mu_1\mathrm{Id}+\mu_2\CCC\,.
\end{equation*}
Using the coproduct \eqref{coproduct} and the fact that, as operators on $\VV_\alpha$,
\begin{equation}
\label{fexp}
\{1\}^2\FF\EE=\qa^\alpha+\qa^{-\alpha}-\qa\KK-\qam\KK^{-1}\,,
\end{equation} 
the Casimir element value on $\VV_{\alpha}\otimes \CC^2$ is
\begin{equation}
\label{casimirmat}
\CCC_{\VV_\alpha\otimes\CC^2}=
\begin{pmatrix}
-\qam\{1\}\KK^{-1}+\qa(\qa^\alpha+\qa^{-\alpha}) & \FF\{1\}^2\\
\qa\{1\}^2\KK^{-1}\EE & \qa\{1\}\KK^{-1}+\qam(\qa^\alpha+\qa^{-\alpha}) 
\end{pmatrix}
\end{equation}
where we interpreted $\CCC_{\VV_\alpha\otimes\CC^2}\in\End(\VV_\alpha\otimes\CC^2)$ as a $2\times 2$ matrix with coefficients in $\End(\VV_\alpha)$ and dropped the $\VV_\alpha$ subscripts to lighten notations.

By Schur's lemma, $\CCC$ is constant on any irreducible representation of $\Uq$. From \eqref{alpharepgen} and \eqref{casimir} we easily compute\footnote{In what follows, we often omit the identity operator $\mathrm{Id}$ and show only the corresponding coefficient.}
\begin{equation}
\label{casval}
\CCC_{\VV_{\alpha}}=\qa^{\alpha}+\qa^{-\alpha}\,.
\end{equation} 
Thus the projectors on $\mathcal{V}_{\alpha\pm 1}$ are given by the solutions of the linear systems
\begin{equation}
\label{cassys}
\left\{
\begin{array}{ll}
\mu_1+\mu_2(\qa^{\alpha\pm 1}+\qa^{-\alpha\mp 1})=1 \,, \\
\mu_1+\mu_2(\qa^{\alpha\mp 1}+\qa^{-\alpha\pm 1})=0 \,.
\end{array}
\right.
\end{equation}
Solving these we obtain
\begin{equation}
\label{blobres}
b_{\pm}=\pm\dfrac{\CCC_{\VV_\alpha\otimes\CC^2}-\qa^{\alpha\mp 1}-\qa^{-\alpha\pm 1}}{\{1\}\{\alpha\}} \,,
\end{equation}
that is,
\begin{equation}
\label{blobqgen}
\begin{aligned}
b_\pm & =\frac{\pm 1}{\{\alpha\}}
\begin{pmatrix}
 -\qam \KK^{-1}+\qa^{\pm \alpha} & \{1\}\FF\\
 \qa\{1\}\KK^{-1}\EE & \qa\KK^{-1}-\qa^{\mp \alpha}
\end{pmatrix}\\
& =\frac{1}{2}\pm\frac{1}{[\alpha]_\qa}\left(\frac{1}{2}\KK^{-1}\otimes 1 +\frac{1}{2\{1\}}\left(\qa^\alpha+\qa^{-\alpha}-(\qa+\qam)\KK^{-1}\right)\otimes\sigma^z+\FF\otimes\sigma^++\qa\KK^{-1}\EE\otimes\sigma^-\right) \,,
\end{aligned}
\end{equation}
where again  we interpreted $b_\pm\in\End(\VV_\alpha\otimes\CC^2)$ as  a $2\times 2$ matrix with coefficients in $\End(\VV_\alpha)$.
By construction, we have $b_\pm^2=b_\pm$, $b_++b_-=1$, and $b_\pm b_\mp=0$. Setting $b:=b_+$ we obtain the explicit expression of the one-boundary Hamiltonian 
\begin{equation}\label{eq:Hb}
H_b:=-\mu b+H_{\mathrm{XXZ}}\ ,
\end{equation}
where $H_{\mathrm{XXZ}}$ is defined in~\eqref{XXZ}, which is an operator acting on $\Hb=\VV_\alpha\otimes (\CC^2)^{\otimes N}$.

To construct the two-boundary Hamiltonian $H_{2b}$ we proceed in the same fashion. We already know that the projector $b_l$ on $\VV_{\alpha_l+1}$ is given by $b_+$ from \eqref{blobqgen} (with $\alpha$ replaced by $\alpha_l$). To compute $b_r$ -- the projector on $\VV_{\alpha_r+1}$ -- we now take the linear combination
\begin{equation*}
b_r=\mu_1+\mu_2\CCC_{\CC^2\otimes\VV_{\alpha_r}}\,.
\end{equation*}
Note that the Casimir element action on $\CC^2\otimes\VV_{\alpha_r}$ is different from the action on $\VV_{\alpha_r}\otimes\CC^2$, as the tensor product of $\Uq$ representations is not symmetric but braided (see Section \ref{twoboundarysec}). For $b_r$ to be the projector on $\VV_{\alpha_r+1}$, $\mu_1$ and $\mu_2$ must satisfy the same constraints \eqref{cassys} (with $\alpha$ replaced by $\alpha_r$). We obtain
\begin{equation}
\label{blobrqgen}
\begin{aligned}
b_r & =\dfrac{\CCC_{\CC^2\otimes\VV_{\alpha_r}}-\qa^{\alpha_r-1}-\qa^{-\alpha_r+1}}{\{1\}\{\alpha_r\}}\\
& =\frac{1}{\{\alpha_r\}}
\begin{pmatrix}
 \qa\KK - \qa^{-\alpha_r} & \qa\{1\}\KK\FF\\
 \{1\}\EE & -\qam\KK+\qa^{\alpha_r}
\end{pmatrix}\\
& =\frac{1}{2}+\frac{1}{[\alpha_r]_\qa}\left(\frac{1}{2} 1 \otimes \KK+\frac{\sigma^z}{2\{1\}}\otimes\left((\qa+\qam)\KK-\qa^\alpha-\qa^{-\alpha}\right)+\qa\sigma^+\otimes\KK\FF+\sigma^-\otimes\EE\right) 
\end{aligned}
\end{equation}
and thus the explicit expression of the two-boundary Hamiltonian 
\begin{equation}\label{eq:H2b}
H_{2b}:=-\mu_l b_l+H_{\mathrm{XXZ}}-\mu_r b_r\ .
\end{equation}
Again, the projector on $\VV_{\alpha_r-1}$ is simply $1-b_r$, so we lose no generality by writing $H_{2b}$ in this form.

\subsection{Root of unity $\qa$}
\label{sec:Hb-root-q}

Let us now assume that $\qa=e^{\frac{i\pi}{p}}$, $p\in\NN\backslash\{0,1\}$, a $2p$-th root of unity.\footnote{The definitions and conventions for even and odd roots of unity differ slightly. We decided to work with only the even roots to lighten the exposition. The odd root cases are conceptually the same, and analogous results can be obtained with only minor modifications. Also, different choices of $2p$-th roots of unity of the form $\qa=e^{\frac{i\pi p'}{p}}$ with $p$ and $p'$ coprime will yield exactly the same results as $\qa=e^{\frac{i\pi}{p}}$.} This case is different, because we now have $\EE_{\VV_\alpha}\ket{p}=0$ \emph{for any} $\alpha$, and so the subspace $\bigoplus_{n\geq p}\ket{n}\subset\VV_\alpha$ is stable under the action of $\Uq$, making $\VV_\alpha$ reducible but indecomposable. Working with such  $\VV_\alpha$ then complicates the algebraic analysis. The proper way to fix this is to slightly modify the algebra $\Uq$ by introducing the restricted unrolled quantum group $\Uqu$ (see for example \cite{geermodtr0}). It has an additional generator $\HH$ satisfying
\begin{equation*}
[\HH,\KK^{\pm 1}]=0, \qquad [\HH,\EE]=2\EE,\qquad [\HH,\FF]=-2\FF
\end{equation*}
with the coproduct $\Delta(\HH)=\HH\otimes\mathsf{1}+\mathsf{1}\otimes\HH$, as well as the relations
\begin{equation*}
\EE^p=\FF^p=0\,.
\end{equation*}
The action of this algebra on the spin-chain sites $\CC^2$ is the same as $\Uq$ with $\HH_{\CC^2}=\sigma^z$. However, the infinite-dimensional Verma modules are now truncated to $p$-dimensional modules $\VV_\alpha:=\bigoplus_{n=0}^{p-1}\CC\ket{n}$, which we still denote $\VV_\alpha$. They are given by
\begin{equation}
\label{alpharep}
\begin{aligned}
& \EE_{\VV_\alpha}\ket{n}=[n]_\qa [n-\alpha]_\qa\ket{n-1} \,, \\
& \FF_{\VV_\alpha}\ket{n}=\ket{n+1} \,, \\
& \HH_{\VV_\alpha}\ket{n}=(\alpha+p-1-2n)\ket{n} \,, \\
& \KK_{\VV_\alpha}^{\pm 1}=\qa^{\pm \HH_{\VV_\alpha}}
\end{aligned}
\end{equation}
for all $0\leq n\leq p-1$, with $\ket{-1}=\ket{p}=0$. Note that the weight of $\ket{n}$ (meaning its $\KK$-eigenvalue) is shifted by $\qa^p=-1$ with respect to the definition \eqref{alpharepgen}. For generic values of $\alpha\in \CC$, i.e.\ $\alpha\notin\ZZ$, these $p$-dimensional modules are irreducible. They are also irreducible if $\alpha  = 0\, \mathrm{mod}\, p$ but not for $\alpha\in \{1, \dots, p-1\} \, \mathrm{mod}\, p$ where they contain an irreducible submodule which is not a direct summand  \cite{geermodtr0}.

The advantage of using $\Uqu$ instead of $\Uq$ at roots of unity is that its representation theory has better properties making it more suitable for our purposes. In particular, at roots of unity, conservation of $\HH$ -- which holds for all the Hamiltonians we consider -- is stronger than the conservation of $\KK:=\qa^\HH$ (which is equivalent to the conservation of $\HH$ only modulo $2p$) so we need to add $\HH$ to our symmetry algebra to correctly describe the centralizers of our systems. A consequence of this is that the $\Uqu$-modules $\VV_\alpha$ and $\VV_{\alpha+2p}$ are no longer isomorphic and so one has to treat $\alpha$ as an element of $\CC$ (and not $\CC/2p\ZZ$ as before). This modification is also needed for the braiding of representations to be well-defined, a feature which will be useful for us later on in Section \ref{twoboundarysec}. Finally, if $\qa^{\alpha}\ne \pm1$ (that is $\alpha\notin p\ZZ$), the fusion rule $\VV_\alpha\otimes\CC^2\cong\VV_{\alpha+1}\oplus \VV_{\alpha-1}$ remains true for the $\Uqu$ representations \eqref{alpharep}. Therefore, if $\alpha\notin\ZZ$ (recall that $\VV_\alpha$ is irreducible if and only if $\alpha\notin\ZZ\backslash p\ZZ$) the only $\Uqu$-invariant operators acting on $\VV_\alpha\otimes\CC^2$ are linear combinations of the projectors $b_\pm$ on $\VV_{\alpha\pm 1}$.

The Casimir $\CCC$ defined in~\eqref{casimir} is a central element of $\Uqu$ and so we can apply the same method as in the generic $\qa$ case. We obtain
\begin{equation}
\label{blobqroot}
b_\pm=\frac{\mp 1}{\{\alpha\}}
\begin{pmatrix}
-\qa^{-1}\KK^{-1}-\qa^{\pm\alpha} & \{1\}\FF\\
\qa\{1\}\KK^{-1}\EE & \qa\KK^{-1}+\qa^{\mp\alpha}
\end{pmatrix} \,,
\end{equation}
where again we interpreted $b_\pm\in\End(\CC^p\otimes\CC^2)$ as  a $2\times 2$ matrix with coefficients in $\End(\VV_\alpha=\CC^p)$.
Note that \eqref{blobqgen} and \eqref{blobqroot} are related by the change of variables $\alpha\to\alpha+p$, which is clear from the definitions. Again, by construction, $b_\pm^2=b_\pm$, $b_++b_-=1$, and $b_\pm b_\mp=0$. Setting $b:=b_+$ we obtain the explicit expression of $H_b = -\mu b+H_{\mathrm{XXZ}}$. We notice that  the only case where $H_b$ is not well-defined is when  $\alpha = 0\, \mathrm{mod}\, p$ (see more explanations from the representation theory point of view below).

The two-boundary Hamiltonian $H_{2b}$ is introduced again very similarly to the generic $\qa$ case. The left coupling is given by the projector $b_l:=b_+$ on $\VV_{\alpha_l+1}$  from \eqref{blobqroot} (with $\alpha$ replaced by $\alpha_l$). 
As for the projector $b_r$, it is now given by
\begin{equation}
\label{blobrqroot}
b_r=\frac{1}{\{\alpha_r\}}
\begin{pmatrix}
 -\qa\KK - \qa^{-\alpha_r} & -\qa\{1\}\KK\FF\\
 -\{1\}\EE & \qam\KK+\qa^{\alpha_r}
\end{pmatrix}\,.
\end{equation}
We thus obtain $H_{2b}:=-\mu_l b_l+H_{\mathrm{XXZ}}-\mu_r b_r$, for $\alpha_l, \alpha_r \ne 0\, \mathrm{mod}\, p$.

One may wonder if we ``lose" something by restricting to $p$-dimensional representations at roots of unity instead of still working with Verma modules as for generic $\qa$. In a nutshell, with this simplification, we might only miss the non-trivial Jordan block structure we could have obtained with Verma modules (which are reducible but indecomposable at roots of unity) but the spectrum remains the same. This will be discussed in more detail in a future paper, where we will study the spectrum of our systems for general $\qa$ and show that their scaling limit partition functions behave smoothly as $\qa$ approaches a root of unity.

\paragraph{Remark for non-generic $\alpha$.} 
If $\qa^\alpha = \pm 1$, the fusion rule $\VV_\alpha\otimes\CC^2\cong\VV_{\alpha+1}\oplus \VV_{\alpha-1}$ does not hold any more for generic as well as for root of unity $\qa$. In both cases, $\VV_\alpha\otimes\CC^2$ is then a reducible but indecomposable module. For example for $\alpha=0$, this is a projective module in the category $\mathcal{O}$ of $\Uq$ (the category of all finitely-generated $\Uq$-modules with finite-dimensional weight spaces and such that every vector is annihilated by a sufficiently large power of $\EE$) containing the Verma modules  $\VV_{1}$ and $\VV_{-1}$ as non-trivial submodules where the Casimir operator is non-diagonalisable \cite{humphreys}, and similarly for $\Uqu$ at roots of unity~\cite{geermodtr0}. Therefore, we have no non-trivial projectors in these cases. This can already be seen in the expressions \eqref{blobqgen} and \eqref{blobqroot} for the projectors $b_\pm$ : they are degenerate if and only if $\qa^\alpha=\pm 1$. This explains why the one-boundary Hamiltonians $H_b$ constructed above are well-defined only for $\qa^\alpha\neq\pm 1$. To cover the missing cases, we note that by rescaling $\mu b\to\tilde{\mu}[\alpha]_\qa b$ and taking the limit $\qa^{\alpha}\to \pm1$ one can extend the definition of $H_b$ even to $\qa^\alpha = \pm 1$. This trick will later be used in the study of the spectrum of the $p=2$ model in Section~\ref{sclim}.

Now if $\qa^\alpha=\pm \qa^n$ for some $n\in\ZZ^*$, or some $n\in\{1,\ldots, p-1\}$ if $\qa$ is a root of unity, even if $H_b$ is well-defined, we will have a similar problem of indecomposable and yet reducible modules but in bigger chains. Indeed, for sufficiently large $N$ (actually for $N=|n|$) a module $\VV_{\beta}$ with $\qa^{\beta}=\pm 1$ appears in the decomposition of the Hilbert space which brings us back to the previous situation with $\VV_0\otimes \CC^2$ for larger values of $N$.
Generally speaking, the presence of such  representations with non-diagonal action of the Casimir element makes the spectrum analysis more complicated, with for example, the appearance of non-trivial Jordan blocks for the Hamiltonian action. We leave its careful treatment for a further work. This is why, in this paper, we will mostly consider the case $\qa^\alpha\notin \pm \qa^{\ZZ}$ (thus simply $\alpha\notin\ZZ$ for root of unity~$\qa$), and only occasionally give comments for $\qa^\alpha\in\pm \qa^{\ZZ}$.

\section{Spectrum of $H_b$ for $\qa=i$ and its scaling limit}
\label{specqi}

In order to motivate our construction -- in particular on physical grounds -- we discuss as an aside, in this section and the next, the case $\qa=i$.  We find the spectrum of $H_b$ for $\qa=i$ using free fermions. We then study its scaling limit, and show that it is given by a conformal field theory whose partition function is calculated explicitly. We identify this CFT with the $(\eta,\xi)$ ghost system on the upper half-plane, with special boundary conditions on the real axis that are indexed by the lattice boundary parameter in the domain $\alpha \in \RR/2\ZZ$, that is $y=\cot\frac{\pi\alpha}{2}\in\RR$.

\subsection{The one-boundary XX model}

When taking $\qa=i$, (\textit{i.e.}, $p=2$) in \eqref{XXZ} we obtain the Hamiltonian of the open $\Uqu$-invariant XX model,
\begin{equation*}
H_{\mathrm{XX}}=\frac{1}{2}\sum_{j=1}^{N-1}(\sigma_j^x\sigma_{j+1}^x+\sigma_j^y\sigma_{j+1}^y)+\frac{i}{2}(\sigma_N^z-\sigma_1^z)\,.
\end{equation*}
Let us now modify the boundary conditions on the leftmost boundary by adding a representation $\VV_\alpha$, as explained in the previous section. For $p=2$, $\VV_\alpha$ defined in~\eqref{alpharep} is two-dimensional and we can express $b$ in terms of Pauli matrices,
\begin{equation}
\label{leftblob}
b=\frac{1}{2}(1+\sigma_0^z)-\cot{\frac{\pi\alpha}{2}}\left(e^{-\frac{i\pi\alpha}{2}}\sigma_0^+\sigma_1^-+\frac{1}{\cos{\frac{\pi\alpha}{2}}}\sigma_0^-\sigma_1^++\frac{i}{2}(\sigma_1^z-\sigma_0^z)\right) \,,
\end{equation}
where the site $0$ corresponds to $\VV_\alpha$. The one-boundary XX Hamiltonian is then given by \eqref{Hb} at $\qa=i$ or $p=2$, i.e.\ $H_b=-\mu b+H_{\mathrm{XX}}$ and it acts on the chain $\Hb=\VV_\alpha\otimes (\CC^2)^{\otimes N}$ where $\VV_\alpha$ is the two-dimensional representation introduced in~\eqref{alpharep}.

\subsection{Construction of plane waves}
\label{planewaves}

Let us introduce the Jordan-Wigner transform. For $0\leq j\leq N$, consider the operators
\begin{equation}
\label{eq:def-c}
c_j^\dagger:=(-1)^j\prod_{k=0}^{j-1}\sigma_k^z\sigma_j^+\,,\qquad
c_j:=(-1)^j\prod_{k=0}^{j-1}\sigma_k^z\sigma_j^-\,,
\end{equation}
satisfying the anti-commutation relations
\begin{equation}
\label{anticomc}
\{c^\dagger_j, c^\dagger_{j'}\}=0\,,\qquad
\{c_j, c_{j'}\}=0\,,\qquad
 \{c^\dagger_j, c_{j'}\}=\delta_{j,j'}\,.
\end{equation}
Then
\begin{equation*}
\frac{1}{2}(\sigma_j^x\sigma_{j+1}^x+\sigma_j^y\sigma_{j+1}^y)+\frac{i}{2}(\sigma_{j+1}^z-\sigma_j^z)=c_jc^\dagger_{j+1}+c_{j+1}c^\dagger_{j}+i(c^\dagger_jc_j-c^\dagger_{j+1}c_{j+1})
\end{equation*}
for $1\leq j\leq N-1$, and
\begin{equation*}
b=c^\dagger_0c_0+\cot{\frac{\pi\alpha}{2}}\left(e^{-\frac{i\pi\alpha}{2}}c_1c^\dagger_0+\frac{1}{\cos{\frac{\pi\alpha}{2}}}c_0c^\dagger_1+i(c^\dagger_0c_0-c^\dagger_1c_1)\right)\,.
\end{equation*}
Therefore $H_b$ is quadratic in $c$'s.

We now want to build operators $\theta^\dagger$ diagonalising the adjoint action of the Hamiltonian, that is, satisfying
\begin{equation}
\label{adjdag}
[H_b,\theta^\dagger]=\lambda \theta^\dagger
\end{equation}
for some $\lambda\in\mathbb{C}$. We will look for $\theta^\dagger$ of the form
\begin{equation*}
\theta^\dagger=\sum_{j=0}^NA_jc^\dagger_j \,.
\end{equation*}
Using the commutation relations
\begin{equation*}
[c_j^\dagger c_j, c^\dagger_k]=\delta_{j,k} c_j^\dagger\,, \qquad
[c_j c_{j+1}^\dagger, c^\dagger_k]=-\delta_{j,k} c_{j+1}^\dagger\,, \qquad
[c_j c_{j-1}^\dagger, c^\dagger_k]=-\delta_{j,k} c_{j-1}^\dagger\,,
\end{equation*}
the equation~\eqref{adjdag} then gives the following  linear system of equations on $A_j$'s:
\begin{equation}
\label{speceq}
\left\{
\begin{array}{llll}
-A_0\mu \left(1+i\cot{\frac{\pi\alpha}{2}}\right)+A_1\mu\cot{\frac{\pi\alpha}{2}}e^{-\frac{i\pi\alpha}{2}}=\lambda A_0 \,, \\
A_0\frac{\mu}{\sin{\frac{\pi\alpha}{2}}}+A_1 i\mu\cot{\frac{\pi\alpha}{2}}+(A_2-iA_1)=\lambda A_1 \,, \\
A_{j-1}+A_{j+1}=\lambda A_j\,, \qquad\qquad\quad  ~2\leq j\leq N-1 \,, \\
A_{N-1} +iA_N=\lambda A_N \,.
\end{array}
\right.
\end{equation}
Eliminating $A_0$ from the first two equations we obtain
\begin{equation}
\label{elimA}
\left\{
\begin{array}{ll}
A_0=A_1\frac{\cos{\frac{\pi\alpha}{2}}e^{-\frac{i\pi\alpha}{2}}}{ie^{-\frac{i\pi\alpha}{2}}+\frac{\lambda}{\mu}\sin{\frac{\pi\alpha}{2}}} \,, \\
A_1\frac{i\lambda\cos{\frac{\pi\alpha}{2}}}{ie^{-\frac{i\pi\alpha}{2}}+\frac{\lambda}{\mu}\sin{\frac{\pi\alpha}{2}}}+(A_2-iA_1)=\lambda A_1 \,,
\end{array}
\right.
\end{equation}
and introducing
\begin{equation*}
\chi(\lambda)=\frac{i\lambda\cos{\frac{\pi\alpha}{2}}}{ie^{-\frac{i\pi\alpha}{2}}+\frac{\lambda}{\mu}\sin{\frac{\pi\alpha}{2}}}
\end{equation*}
equations (\ref{speceq}) become
\begin{equation}
\label{speceqA}
\left\{
\begin{array}{lll}
A_2-(i-\chi(\lambda))A_1=\lambda A_1 \,, \\
A_{j-1}+A_{j+1}=\lambda A_j \,,\qquad\quad 2\leq j\leq N-1 \,, \\
A_{N-1} +iA_N=\lambda A_N \,.
\end{array}
\right.
\end{equation}

Similarly, solving
\begin{equation*}
[H_b,\theta]=\lambda \theta
\end{equation*}
with
\begin{equation*}
\theta=\sum_{j=0}^NB_jc_j
\end{equation*}
gives
\begin{equation*}
B_0=\frac{B_1}{ie^{-\frac{i\pi\alpha}{2}}-\frac{\lambda}{\mu}\sin{\frac{\pi\alpha}{2}}}
\end{equation*}
and
\begin{equation}
\label{speceqB}
\left\{
\begin{array}{lll}
B_2-(i-\chi(-\lambda))B_1=-\lambda B_1 \,, \\
B_{j-1}+B_{j+1}=-\lambda B_j \qquad \forall~2\leq j\leq N-1 \,, \\
B_{N-1} +iB_N=-\lambda B_N \,.
\end{array}
\right.
\end{equation}
We see that the systems (\ref{speceqA}) and (\ref{speceqB}) are related by the transformation $\lambda\to -\lambda$, and so it is sufficient to consider only one of them.

We now look for plane wave solutions, that is, amplitudes $A_j$ of the form
\begin{equation*}
A_j=a_+x^j+a_-x^{-j} \,,
\end{equation*}
where $x\in\mathbb{C}\backslash\{-1,1\}$.\footnote{Here $x=\pm 1$ are not allowed, since the amplitude $A_j$ will then only depend on the combination $a_++a_-$, and so $a_+$ and $a_-$ cannot be treated as independent variables. One can check that the corresponding constant (or alternating) amplitude is not a solution of (\ref{speceqA}).} This ansatz solves the system (\ref{speceqA}) in the bulk (for all $2\leq j\leq N-1$) with $\lambda=x+x^{-1}$. We now have to look for $x$ such that the two boundary equations admit a non-zero solution. Rewriting (\ref{speceqA}) we obtain
\begin{equation*}
\left\{
\begin{array}{ll}
(1+(i-\chi(\lambda))x)a_++(1+(i-\chi(\lambda))x^{-1})a_-=0 \,, \\
x^N(i-x)a_++x^{-N}(i-x^{-1})a_-=0 \,.
\end{array}
\right.
\end{equation*}
The equation on $x$ is obtained by imposing that the determinant of this system vanishes. This computation gives
\begin{equation}
\label{waveq}
\lambda\left(x^N\left(\sin{\frac{\pi\alpha}{2}}\left(1+\frac{\lambda}{\mu}\right)-x^{-1}\cos{\frac{\pi\alpha}{2}}\right)-x^{-N}\left(\sin{\frac{\pi\alpha}{2}}\left(1+\frac{\lambda}{\mu}\right)-x\cos{\frac{\pi\alpha}{2}}\right)\right)=0 \,.
\end{equation}

We note that the presence of the overall $\lambda$ factor, and the corresponding solution $\lambda=0$, is due to the $\Uqu$ symmetry. In what follows, we only analyse solutions with non-zero $\lambda$, and as said above $x\ne\pm1$. Then
dividing by ${\lambda(x-x^{-1})}$ and introducing the Chebyshev polynomials of the second kind
\begin{equation*}
U_n\left(\frac{x+x^{-1}}{2}\right)=\frac{x^{n+1}-x^{-n-1}}{x-x^{-1}} \,,
\end{equation*}
we can rewrite~\eqref{waveq} in a more compact form (recall that $\lambda:= x +x^{-1}$)
\begin{equation}
\label{waveqU}
P(\lambda):=U_N(\lambda/2)+\mu U_{N-1}(\lambda/2)+(1-\mu y)U_{N-2}(\lambda/2)=0 \,,
\end{equation}
where we introduced 
\begin{equation}
\label{blobwqi}
y:=\cot{\frac{\pi\alpha}{2}}\,.
\end{equation}
Denoting by $\lambda_k=x_k+x_k^{-1}$, for $1\leq k\leq N$, the solutions of this polynomial equation we obtain (with some choice of normalisation) $N$ operators $\theta_k^\dagger$,
\begin{equation*}
\theta_k^\dagger:=\sum_{j=0}^NA_j(\lambda_k)c_j^\dagger \,, \qquad 1\leq k\leq N \,,
\end{equation*}
where\footnote{By convention $U_{-1}=0$ and $U_0=1$.}
\begin{equation}
\label{expA}
\begin{aligned}
& A_0(\lambda)=-\frac{e^{-\frac{i\pi\alpha}{2}}\chi(\lambda)}{\lambda}\left(U_{N-2}(\lambda/2)+iU_{N-1}(\lambda/2)\right)\,,\\
& A_j(\lambda)=U_{N-j}(\lambda/2)-iU_{N-j-1}(\lambda/2)\qquad\forall~1\leq j\leq N\,.
\end{aligned}
\end{equation}
The explicit expressions \eqref{expA} are found by noticing that the standard property of the Chebyshev polynomials
\begin{equation*}
\lambda U_n(\lambda/2)=U_{n+1}(\lambda/2)+U_{n-1}(\lambda/2)
\end{equation*}
implies that, for any $\lambda\in\CC$, $A_j(\lambda)$ satisfy the system \eqref{speceqA} for $2\leq j\leq N$. If additionally $\lambda$ is a solution of \eqref{waveqU} the first equation of \eqref{speceqA} is also fulfilled. Computing $A_0(\lambda)$ from \eqref{elimA}, we obtain the coefficients $A_j(\lambda_k)$ corresponding to a solution $\lambda_k$. By construction, the operators $\theta_k^\dagger$ then satisfy
\begin{equation}
\label{adjk}
[H_b,\theta_k^\dagger]=\lambda_k\theta^\dagger_k\,.
\end{equation}

Using the $\lambda\to-\lambda$ transformation we also have $N$ operators $\theta_k$,
\begin{equation*}
\theta_k=\sum_{j=0}^NB_j(\lambda_k)c_j \,, \qquad 1\leq k\leq N \,,
\end{equation*}
where
\begin{equation}
\label{expB}
\begin{aligned}
& B_0(\lambda)=-\frac{\chi(\lambda)}{\lambda\cos{\frac{\pi\alpha}{2}}}(U_{N-2}(\lambda/2)+iU_{N-1}(\lambda/2))=\frac{e^{\frac{i\pi\alpha}{2}}}{\cos\frac{\pi\alpha}{2}}A_0(\lambda)\,,\\
& B_j(\lambda)=A_j(\lambda)\qquad\forall~1\leq j\leq N\,,
\end{aligned}
\end{equation}
satisfying
\begin{equation*}
[H_b,\theta_k]=-\lambda_k\theta_k\,.
\end{equation*}

Finally, the zero modes corresponding to the eigenvalue $\lambda_0\equiv0$ are given by
\begin{equation}
\label{zeromodes}
\theta_0^\dagger=\cos{\frac{\pi\alpha}{2}}c_0^\dagger+\sum_{j=1}^Ni^jc_j^\dagger\,,\qquad
 \theta_0=c_0+e^{-\frac{i\pi\alpha}{2}}\sum_{j=1}^Ni^jc_j\,.
\end{equation}
They are related to $\Uqu$ generators as
\begin{equation}
\label{uqtheta}
\EE_{\Hb}=i\theta_0^\dagger \KK_{\Hb}\,,\qquad \FF_{\Hb}=\theta_0\,,
\end{equation}
and they satisfy the anti-commutation relation
\begin{equation}
\label{zerocomrel}
\{\theta_0^\dagger,\theta_0\}=\frac{e^{\frac{i\pi\alpha}{2}}+(-1)^N e^{-\frac{i\pi\alpha}{2}}}{2}\,.
\end{equation}

All the constructed operators satisfy the anti-commutation relations
\begin{equation}
\label{anticom}
\{\theta^\dagger_k,\theta^\dagger_{k'}\}=0 \,, \qquad
\{\theta_k,\theta_{k'}\}=0 \,.
\end{equation}
Note however that $\{\theta_k^\dagger,\theta_{k'}\}\neq \delta_{k,k'}$ and that $\theta_k^\dagger$ and $\theta_k$ are not actually the adjoint of one another. In general $\{\theta_k^\dagger,\theta_{k'}\}=S_{kk'}\Id_{\Hb}$ for some $(N+1)\times (N+1)$ matrix $(S_{kk'})$. Since $S_{kk'}\to\delta_{kk'}$ as $\mu\to 0$ (for some normalisation of $\theta_k$), this matrix is generically invertible, and so we can in principle find a linear transformation $\theta_k\mapsto\tilde{\theta}_{k}$ such that $\{\theta_k^\dagger,\tilde{\theta}_{k'}\}= \delta_{k,k'}$, but the new modes $\tilde{\theta}_{k}$ will no longer diagonalise the adjoint action of the Hamiltonian $H_b$. As the matrix $(S_{kk'})$ depends explicitly on the solutions of \eqref{waveqU} it does not admit a simple closed form in full generality. Later on, we will still manage to calculate it up to order $1/N$ (see equation \eqref{eq:anti-com-N}).

\subsection{Computation of the spectrum}
\label{compspec}

For $p=2$, $\VV_\alpha:=\CC\ket{0}\oplus\CC\ket{1}$ is of dimension $2$, so for convenience let us denote the basis vectors $\ket{0}$ and $\ket{1}$ defined in \eqref{alpharep} by $\ket{\uparrow}$ and $\ket{\downarrow}$ respectively.

We take a subset $S\subseteq\{1,\ldots,N\}$ and consider the vector
\begin{equation*}
\ket{S}=\left(\prod_{k\in S}\theta^\dagger_k\right)\ket{\downarrow}^{\otimes N+1}\,.
\end{equation*}
It is non-zero, and by \eqref{anticom} the ordering of $\theta^\dagger_k$'s will only affect its sign. These vectors are eigenstates for the Hamiltonian $H_b$. Indeed, by \eqref{adjk}, and since $\ket{\downarrow}^{\otimes N+1}\in\mathrm{Ker}~H_b$ (from \eqref{XXZ} and \eqref{leftblob} one easily checks that $H_{\mathrm{XXZ}}\ket{\downarrow}^{\otimes N+1}=b\ket{\downarrow}^{\otimes N+1}=0$), we have
\begin{equation*}
H_b\ket{S}=\sum_{k\in S}\lambda_k\ket{S}:=\lambda_S\ket{S}\,,
\end{equation*}
where $\lambda_k$ are corresponding solutions of the equation~\eqref{waveqU}.
There are $2^N$ such eigenvectors, but $\mathrm{dim}~\Hb=2^{N+1}$.  However, when $\alpha\notin\mathbb{Z}$, using repeatedly the fusion rule $\VV_\alpha\otimes\CC^2\cong\VV_{\alpha+1}\oplus \VV_{\alpha-1}$, we know that the Hilbert space decomposes as
\begin{equation}
\label{uqdecomp}
\Hb=\bigoplus_{k=0}^N {N\choose k} \VV_{\alpha-N+2k}
\end{equation}
into two-dimensional irreducible representations of $\Uqu$. Therefore if the $\lambda_S$'s are pairwise distinct (which is true for generic $\alpha$ and $\mu$) each eigenvalue of $H_b$ is only twice degenerate with each $\lambda_S$-eigenspace spanned by the two linearly independent sates $\ket{S}$ and $\theta_0^\dagger\ket{S}\propto\EE_{\Hb}\ket{S}$. We thus get
\begin{equation}
\mathrm{Spec}(H_b)=\{\lambda_S~|~S\subseteq\{1,\ldots,N\}\}\,.
\end{equation}
For special non-generic values of $\alpha$ and $\mu$, some $\lambda_S$ corresponding to different $S$ may ``accidentally" be equal but with vectors $\ket{S}$ and $\theta_0^\dagger\ket{S}$ still providing an eigenbasis. A more complicated situation is when these vectors become linearly dependent. In that case we still expect that \emph{all} the eigenvectors of $H_b$ are (possibly, linear combinations of vectors) of the form $\ket{S}$ or $\theta_0^\dagger\ket{S}$ but that they cannot span the whole chain $\Hb$ because of the appearance of non-trivial Jordan blocks in the action of $H_b$. Such a phenomenon can be observed for $\alpha\in\ZZ$ (see the corresponding Remark at the end of Section~\ref{construction}), but rigorously proving that it happens only at those values requires further study. For example, the case of $\alpha=1$ produces the standard open XX chain on an even number of sites where the Hamiltonian is known to be non-diagonalisable. This case was studied in detail in~\cite{Gainutdinov_2014} (see the discussion at the end of Section~\ref{sclim}). 

For $\alpha\notin\ZZ$, the basis $\{\ket{S}, \theta_0^\dagger\ket{S}\}_{S\subseteq\{1,\ldots,N\}}$ respects the $\Uqu$-decomposition \eqref{uqdecomp}. Indeed, if we introduce the $k$-fermion subspaces, for $0\leq k \leq N$,
\begin{equation}\label{eq:Fock-space-W}
W_k:=\bigoplus_{|S|=k}\CC\ket{S}\quad\text{and}\quad \overline{W}_k:=\theta_0^\dagger W_k\ ,
\end{equation}
each of them is of dimension $\binom{N}{k}$,
then, using the commutation relations of the zero modes $\theta_0=\FF$ and $\theta_0^\dagger=\EE\KK^{-1}$, we see that
\begin{equation*}
\theta_0\overline{W}_k=W_k \qquad 
\text{and}\qquad
\theta_0 W_k=\theta_0^\dagger \overline{W}_k= \{0\} \,.
\end{equation*}
In other words, $W_k\oplus \overline{W}_k$ is stable under the action of $\Uqu$. Moreover, since 
\begin{equation}
\label{fermnb}
\HH |_{W_k}=(\alpha+2k-N-1)\Id_{W_k}\,,\qquad \HH |_{\overline{W}_k}=(\alpha+2k-N+1)\Id_{\overline{W}_k} \,,
\end{equation}
and $W_k$ (resp.\ $\overline{W}_k$) is annihilated by $\FF$ (resp.\ $\EE$), we have 
\begin{equation*}
W_k\oplus \overline{W}_k={N\choose k} \VV_{\alpha-N+2k}\subset\bigoplus_{k=0}^N {N\choose k} \VV_{\alpha-N+2k}=\Hb
\end{equation*}
for all $0\leq k\leq N$. Let us also stress again that $W_k$ and $\overline{W}_k$ are both stable under the action of $H_b$, and that $\mathrm{Spec}(H_b|_{W_k})=\mathrm{Spec}(H_b|_{\overline{W}_k})$. We will show in Section~\ref{blobsec} that these spaces are even irreducible representations of the lattice algebra generated by the nearest neighbour couplings $(e_i)_{1\leq i\leq N-1}$ and $b$, that is the blob algebra.

We would now like to show that the spectrum of $H_b$ is real in some domain in $(\mu,y)$-space. For this, we have to show that equation \eqref{waveqU} has exactly $N$ real roots for some values of $\mu$ and $y$. We would also like to parametrise all the solutions as $\lambda_k=2\cos{p_k}$, $p_k\in]0,\pi[$, in order to interpret $p_k$ as the momentum of the spin wave associated to $\theta_k^\dagger$. Therefore we will be looking for values of $\mu$ and $y$ such that all the solutions of \eqref{waveqU} are in the interval $]-2,2[$.

The first requirement is obviously $\mu,y\in\RR$, so that the coefficients of $P$ are real. Moreover, since the spectral equation \eqref{waveqU} is invariant under $(\lambda,\mu,y)\to -(\lambda,\mu,y)$, we can restrict ourselves to $\mu>0$, $y\in\RR$ (if $\mu=0$, $H_b=H_{\mathrm{XX}}$ whose spectrum is known). Now for $y=0$ \eqref{waveqU} it takes the form
\begin{equation*}
(\lambda+\mu)U_{N-1}(\lambda/2)=0 \,,
\end{equation*}
so as long as $|\mu|<2$ and $U_{N-1}(\mu/2)\neq 0$, all the solutions are distinct and in $]-2,2[$. Now, if we fix such a $\mu$, this will remain true for $y$ in some neighbourhood of $0$. Let $]y_-(\mu,N), y_+(\mu,N)[$ be the maximal such neighbourhood. Its endpoints correspond to the smallest values of $y$, such that either one of the roots of $P$ leaves $]-2,2[$, or two of them collide and are ejected into the complex plane. To determine them, set $\lambda=2\cos{\xi}$, $\xi\in]0,\pi[$ and rewrite \eqref{waveqU} as
\begin{equation}
\label{compeq}
\begin{aligned}
& \sin{(N+1)\xi}+\mu\sin{N\xi}+(1-\mu y)\sin{(N-1)\xi}=0 \,, \\
& \Leftrightarrow \left((2-\mu y)\cos{\xi}+\mu\right)\sin{N\xi}+\mu y \sin{\xi}\cos{N\xi}=0 \,, \\
& \Leftrightarrow \Delta\sin{\left(N\xi+\varphi(\xi)\right)}=0 \,,
\end{aligned}
\end{equation}
where
\begin{equation*}
\Delta=\sqrt{\left((2-\mu y)\cos{\xi}+\mu\right)^2+\left(\mu y \sin{\xi}\right)^2}
\end{equation*}
and
\begin{equation*}
\varphi(\xi)=\arccot{\frac{(2-\mu y)\cos{\xi}+\mu}{\mu y \sin{\xi}}}\,.
\end{equation*}
We thus have to determine whether we can find $N$ distinct $\xi\in]0,\pi[$, such that
\begin{equation}
\label{simpeq}
N\xi+\varphi(\xi)=k\pi
\end{equation}
for some $k\in\ZZ\,$. If $0<\mu<2$ and $0<y<\frac{2-\mu}{\mu}:=y_{\text{max}}$, then $\varphi(0^+)=0^+$ and $\varphi(\pi^-)=\pi^-$, and so by the intermediate value theorem, \eqref{simpeq} will have exactly $N$ solutions in $]0,\pi[$ corresponding to $1\leq k\leq N$. Therefore $y_{\text{max}}\leq y_+(\mu,N)$ for all $N$. Actually, by plotting $P$ for different values of $0<\mu<2$, $y>0$ and $N$, one sees that $y_+(\mu,N)$ is reached because of a solution hitting $-2$, so one has
\begin{equation}
 U_N(-1)+\mu U_{N-1}(-1)+(1-\mu y_+)U_{N-2}(-1)=0~\Rightarrow~y_+(\mu,N)=\frac{2-\mu}{\mu}\frac{N}{N-1}\,.
\end{equation}
Note that $y_+(\mu,N)\to y_{\text{max}}^+$ as $N\to\infty$.

Now, if $0<\mu<2$ and $y<0$, $\varphi(0^+)=\pi^-$ and $\varphi(\pi^-)=0^+$, so \eqref{simpeq} will have at least $N-2$ solutions in $]-2,2[$. The two remaining roots of $P$ can also belong to $]-2,2[$, as the function $\xi\to N\xi+\varphi(\xi)$ may pass by some $\pi k$ more than once as $\xi$ goes from $0$ to $\pi$, if it is not strictly increasing. Numerically, one sees that as $y$ decreases below $0$, two of the $N$ roots of $P$ at $y=0$ collide and become complex. The value $y_-(\mu,N)$ at which this happens is given by the solution of the system $P_{y_-}(\lambda)=P'_{y_-}(\lambda)=0$ in $\lambda$ and $y_-$ and cannot be expressed analytically. However, it is clear that for sufficiently large $N$, the function $\xi\to N\xi+\varphi(\xi)$ is strictly increasing on $]0,\pi[$, and so \eqref{simpeq} will have exactly $N-2$ solutions, corresponding to $2\leq k\leq N-1$. This implies that $y_-(\mu,N)\to 0^-$ as $N\to\infty$. Therefore we see that the spectral equations of $y>0$ and $y<0$ have quite different properties. Additional properties of the roots of $P$ for $\mu>0$ and $y\in\RR$ are given in Appendix \ref{generalpropspec}.

\paragraph{Example.}

Take $\mu=y=y_{\text{max}}=1$. With the parametrisation $\lambda=2\cos{\xi}$, \eqref{waveqU} becomes
\begin{equation*}
U_N(\cos{\xi})+U_{N-1}(\cos{\xi})=\frac{\sin{(N+1)\xi}+\sin{N\xi}}{\sin{\xi}}=2\cot{\frac{\xi}{2}}\sin{\frac{(2N+1)\xi}{2}}=0 \,,
\end{equation*}
and so we see that it has $N$ solutions $\lambda_k\in]-2,2[$, 
\begin{equation}
\lambda_k=2\cos{\frac{2k\pi}{2N+1}} \,, \qquad 1\leq k\leq N \,,
\end{equation}
with real associated momenta $p_k=\frac{2k\pi}{2N+1}$.

\subsection{Scaling limit}
\label{sclim}

We now want to study the spectrum of $H_b$ in the scaling limit $N\to+\infty$, to extract information about the CFT which this lattice model will give in the continuum. In particular we would like to find the surface energy, the central charge and the partition function in terms of $y$ and $\mu$. To do so, we have to compute the ground-state energy and the low-lying excitations at order $o(1/N)$. 

Assume $0<\mu<2$, $0\leq y\leq y_{\text{max}}$, and $N$ even (the odd $N$ case requires only slight modifications and will be briefly discussed later on). Since the ground-state energy is the sum of all the negative $\lambda_k$, corresponding to $N/2+1\leq k\leq N$, let us change variables $\xi\to\pi-\xi$, $k\to N+1-k$ in \eqref{simpeq} to rewrite it as
\begin{equation}
\label{neweq}
N\xi+\phi(\xi)=k\pi \,,
\end{equation}
with $\phi(\xi):=\pi-\varphi(\pi-\xi)$. Denote $(\xi_k)_{1\leq k\leq N}$ its $N$ solutions in $]0,\pi[$ corresponding to each $1\leq k\leq N$. One of the two ground states\footnote{Recall that each eigenvalue has multiplicity $2$ because of the $\Uqu$ symmetry.} is then given by 
$$
\ket{{\rm vac}_0}:=\left(\prod_{k=1}^{N/2}\theta_k^\dagger\right)\ket{\downarrow}^{\otimes N+1}\in W_{N/2}\ ,
$$
 and its energy is
\begin{equation*}
E_0=-\sum_{k=1}^{N/2} 2\cos{\xi_k}\,.
\end{equation*}

Let us expand $\xi_k$ at order $o(1/N^2)$,
\begin{equation*}
\xi_k=\frac{k\pi}{N}+\frac{\xi_k^{(1)}}{N}+\frac{\xi_k^{(2)}}{N^2}+o(1/N^2)\,.
\end{equation*}
Plugging this expression into (\ref{neweq}) we obtain
\begin{equation}
\label{xiexp}
\begin{cases}
\xi_k^{(1)}=-\phi\left(\frac{k\pi}{N}\right) \,, \\
\xi_k^{(2)}=\phi'\phi\left(\frac{k\pi}{N}\right) \,.
\end{cases}
\end{equation}
Therefore
\begin{equation}
\label{dl}
\begin{aligned}
2\cos{\xi_k}= & ~  2\cos{\frac{k\pi}{N}}+\left(2\phi\sin\right){\left(\frac{k\pi}{N}\right)}\frac{1}{N}\\
& -\left(2\phi'\phi\sin+\phi^2\cos\right)\left(\frac{k\pi}{N}\right)\frac{1}{N^2}+o(1/N^2)\,.
\end{aligned}
\end{equation}
Let us now compute the sum of each term in \eqref{dl} at order $o(1/N)$. First,
\begin{equation*}
\sum_{k=1}^{N/2} 2\cos{\frac{k\pi}{N}}=-1+\cot{\frac{\pi}{2N}}=\frac{2N}{\pi}-1-\frac{\pi}{6N}+o(1/N)\,.
\end{equation*}
To compute the sum of the second term at order $o(1/N)$ we have to use the Euler-Maclaurin formula
\begin{equation*}
\begin{aligned}
\frac{1}{N}\sum_{k=1}^{N/2} \left(2\phi\sin\right)\left(\frac{k\pi}{N}\right) & =\int_{0}^{1/2} \left(2\phi\sin\right)(\pi u)\mathrm{d}u+\frac{\left(\phi \sin\right){\left(\frac{\pi}{2}\right)}+\left(\phi \sin\right)(0)}{N}+o(1/N)\\
& =\frac{2}{\pi}\int_{0}^{\pi/2}\phi(u)\sin(u)\mathrm{d}u+\frac{\pi-\arccot{(1/y)}}{N}+o(1/N)\,.
\end{aligned}
\end{equation*}
Finally,
\begin{equation*}
\begin{aligned}
-\frac{1}{N^2}\sum_{k=1}^{N/2} & \left(2\phi'\phi\sin+\phi^2\cos\right)\left(\frac{k\pi}{N}\right)=\\
& =-\frac{1}{N}\int_{0}^{1/2}  \left(2\sin\phi'\phi+\cos\phi^2\right)(\pi u)\mathrm{d}u+o(1/N)\\
& = -\frac{1}{\pi N}\left[\sin(u)\phi(u)^2\right]_{0}^{\pi/2}+o(1/N)\\
& = -\frac{1}{\pi N}(\pi-\arccot{(1/y)})^2+o(1/N)\,.
\end{aligned}
\end{equation*}
Putting everything together we obtain
\begin{equation}
\label{energyexp}
E_0=Ne_{\rm b}+E_{\rm s}-\frac{\pi v_{\rm F}}{N}\frac{c_{\mathrm{eff}}}{24}+o(1/N^2) \,,
\end{equation}
where
\begin{equation}
e_{\rm b}=-\frac{2}{\pi}
\end{equation}
is the bulk energy per site,
\begin{equation}
\label{surface}
E_{\rm s}=1-\frac{2}{\pi}\int_{0}^{\pi/2}\phi(u)\sin(u)\mathrm{d}u
\end{equation}
is the surface energy,
\begin{equation}
\label{fermiv}
v_{\rm F}=p\sin{\frac{\pi}{p}}=2
\end{equation}
is the Fermi velocity and
\begin{equation}
\label{charge}
\begin{aligned}
c_{\mathrm{eff}} & =-2+\frac{12}{\pi}(\pi-\arccot(1/y))-\frac{12}{\pi^2}(\pi-\arccot(1/y))^2\\
& = -2+\frac{12}{\pi}\arccot(1/y)-\frac{12}{\pi^2}\arccot(1/y)^2\\
& = 1-3\alpha^2
\end{aligned}
\end{equation}
is the effective central charge. The value of $e_{\rm b}$ matches that of the usual XX model, which is unsurprising, as boundary conditions are not expected to influence the bulk properties of the system. The value of $E_{\rm s}$ is different, however. In the limit  $\mu\to 0$, the integral in \eqref{surface} vanishes and we recover the surface energy of the XX model (which is equal to $1$). As far as we know, the exact expression \eqref{surface} is new. Finally, the central charge of the XX is known to be $c:=-2$, (see for example \cite{PASQUIER1990523}). With our new boundary conditions, we obtain an effective central charge $c_{\mathrm{eff}}$ which only depends on $\alpha$ (and not $\mu$). Introducing the conformal weights
\begin{equation}
\label{confw}
h_{r,s}:=\frac{(2r-s)^2-1}{8} \,,
\end{equation}
we have
\begin{equation*}
c_{\mathrm{eff}}=c-24h_{\alpha,\alpha}\,.
\end{equation*}

Let us compute the scaling limit of the low-lying excitations. First, let us find the ground states of the other $W_j$ sectors. They are obtained by acting on $\ket{{\rm vac}_0}$ with the $\theta_k^\dagger$, $N/2+1\leq k\leq N$, corresponding to the smallest (positive) $\lambda_k$'s, or by removing $\theta_k^\dagger$, $1\leq k\leq N/2$ corresponding to the biggest (negative) $\lambda_k$. In both cases we have to take $k$ close to $N/2$. Set $\tilde{k}=N/2-k$. For $\tilde{k}$ close to $0$ and at large $N$,
\begin{equation}
\label{lambdaexp}
\lambda_{\tilde{k}}=2\cos{\xi_{\tilde{k}}}=\frac{2\tilde{k}\pi}{N}+\frac{2\phi\left(\frac{\pi}{2}\right)}{N}+o(1/N)=\frac{(2\tilde{k}+\alpha+1)\pi}{N}+o(1/N)\,.
\end{equation}
Thus, the energy $E_j$ of the ground state $\ket{{\rm vac}_j}:=\left(\prod_{k=1}^{N/2+j}\theta_k^\dagger\right)\ket{\downarrow}^{\otimes N+1}$ of the $W_{N/2+j}$ sector, $1\leq j$, is
\begin{equation*}
E_j=E_0-\frac{\pi}{N}\sum_{\tilde{k}=-1}^{-j}(2\tilde{k}+\alpha+1)+o(1/N)=E_0+\frac{\pi v_F}{N}\frac{j(j-\alpha)}{2}+o(1/N) \,,
\end{equation*}
and the energy $E_{-j}$ of the ground state $\ket{{\rm vac}_{-j}}:=\left(\prod_{k=1}^{N/2-j}\theta_k^\dagger\right)\ket{\downarrow}^{\otimes N+1}$ of the $W_{N/2-j}$ sector, $1\leq j$, is
\begin{equation*}
E_{-j}=E_0+\frac{\pi}{N}\sum_{\tilde{k}=0}^{j-1}(2\tilde{k}+\alpha+1)+o(1/N)=E_0+\frac{\pi v_F}{N}\frac{j(j+\alpha)}{2}+o(1/N)\,.
\end{equation*}

Now, let us fix a sector $W_{N/2+j}$, $j\in\ZZ$, and compute the low-lying excitations above $\ket{{\rm vac}_j}$ with the energy $E_j$. These are obtained by replacing some of the $\theta_k^\dagger$ with $k$ close to $N/2$ appearing in $\ket{{\rm vac}_j}$ by some others with bigger $\lambda_k$ (but still close to $N/2$), keeping in mind that no $\theta_k^\dagger$ must appear more than once (otherwise the resulting vector is zero). By \eqref{lambdaexp}, an elementary substitution of some $\theta_k^\dagger$ by some $\theta_{k'}^\dagger$ with $k<k'$ increases the energy by $\frac{2\pi(k'-k)}{N}$ at order $o(1/N)$. We claim that for $N\to\infty$ the number of ways to increase the ground-state energy $E_j$ by $\frac{2\pi m}{N}$ for $m\in\NN$ is equal to the number of integer partitions of $m$. To see this, take such a partition $K=(\kappa_0\geq\ldots\geq\kappa_\ell)$ of $m$ and associate to it the vector $\ket{K,j}$ by performing the series of substitutions
\begin{equation*}
(\theta^{\dagger}_{N/2+j},\ldots,\theta^{\dagger}_{N/2+j-\ell}) \to (\theta^{\dagger}_{N/2+j+\kappa_0},\ldots,\theta^{\dagger}_{N/2+j-\ell+\kappa_\ell})
\end{equation*}
inside $\ket{{\rm vac}_j}$. Conversely, it is clear that to any eigenvector of $H_b$ in $W_{N/2+j}$ one can associate such a partition $K$ in a unique way. Of course, this bijection only holds for $m$ not too big, since we have a finite number of $\theta^\dagger_k$ at our disposal. A sufficient condition is, for example, that $m< N/2-|j|$. As we are only interested in the low-lying excitations, that is $m\ll N$, this will be enough.

To summarise, in the scaling limit, the generating function of the spectrum of $H_b$, restricted to $W_{N/2+j}$ for each $j\in\ZZ$, is given by
\begin{equation}
\label{wpartfun}
Z_j(q,\alpha) :=\lim_{M\to\infty}\lim_{N\to\infty}\tr_{W_{N/2+j}}^{<M} q^{\frac{N}{\pi v_F}(H_b-Ne_b-E_s)}=\frac{q^{-\frac{1-3\alpha^2}{24}}}{P(q)}q^{\frac{j(j-\alpha)}{2}} \,,
\end{equation}
where $1/P(q):=1/\prod_{n=1}^{+ \infty}(1-q^n)$ is the generating function of integer partitions, and $\tr_{W_{N/2+j}}^{<M}$ means that we only sum over the $M$ first excitations above the ground state in the $W_{N/2+j}$ sector. This double-limit construction is needed to ensure that the final expression is convergent and consistent with our computation.

The full generating function of the spectrum of $H_b$ in the scaling limit is given by
\begin{equation}
\label{partfun}
\begin{aligned}
Z(q,\alpha) & :=\lim_{M\to\infty}\lim_{N\to\infty}\tr_{\Hb}^{<M} q^{\frac{N}{\pi v_F}(H_b-Ne_b-E_s)}\\
& \,=2\sum_{j\in\ZZ}Z_j(q,\alpha)=\frac{2q^{-\frac{1-3\alpha^2}{24}}}{P(q)}\sum_{j\in\ZZ}q^{\frac{j(j-\alpha)}{2}} \,,
\end{aligned}
\end{equation}
where the factor of $2$ comes from the $\Uqu$-symmetry which makes each eigenvalue twice degenerate (or, equivalently, from the contribution of the $\overline{W}_j$ sectors). Rewriting this result in terms of the conformal weights \eqref{confw} and of the XX central charge $c:=-2$, as
\begin{equation*}
-\frac{1-3\alpha^2}{24}+\frac{j(j\pm\alpha)}{2}=-\frac{c}{24}+h_{\alpha,\alpha\mp 2j} \,,
\end{equation*}
we obtain
\begin{equation}
\label{wpartfunh}
Z_j(q,\alpha):=\frac{2q^{-\frac{c}{24}}}{P(q)}q^{h_{\alpha,\alpha+2j}}\,.
\end{equation}
Thus, in the scaling limit, the spectrum of $H_b|_{W_{N/2+j}}$ is exactly that of a CFT with central charge $-2$ in the generic Virasoro Verma representation of conformal weight $h_{\alpha,\alpha+2j}$. This proves a special case ($\qa=i$) of a prediction made in \cite{Nichols2006, Jacobsen_2008} for the abstract one-boundary Hamiltonian evaluated on standard representations of the blob algebra. We discuss this more from lattice algebra considerations in Sections~\ref{blobsec} and~\ref{conclusion}. The proof of the general result for any $\qa$ with $|\qa|=1$ will be given in a forthcoming paper.

All these results are manifestly independent of $\mu$, which is consistent with \cite{Doikou_2003, Jacobsen_2008} and renormalisation arguments \cite{Gainutdinov_2013}. Nevertheless, one should not forget that in our derivation we had to assume $0<\mu<2$ and $0\leq y\leq y_{\text{max}}$, so technically speaking, equations \eqref{partfun} and \eqref{wpartfun} are valid only for $\alpha\in]0,1]$. To extend them to $y<0$ (and still $\mu>0$), that is to $\alpha\in[1,2[$,\footnote{We chose this interval (and not $]-1,0]$ for example), because $b$ becomes singular at $\alpha=0$ and computations \eqref{charge}-\eqref{lambdaexp} are valid only for $0<\alpha<2$.} we have to modify the expansion \eqref{xiexp} to take into account the fact that at large $N$ we will have only $N-2$ solutions in $]-2,2[$. A faster way is to invoke the unicity of the analytic continuation in $\alpha$, and claim that \eqref{wpartfun} and \eqref{partfun} remain valid for $\alpha\in [1,2[$. It is however important to remember that even if the final result does not explicitly depend on $\mu$, its sign is still important, as the spectral equation of $H_b$ is only invariant under the simultaneous transformation $(\mu,y)\to -(\mu,y)$. Concretely, from \eqref{blobwqi}, and since $\alpha\in]0,2[$, the change of sign $y\to -y$ is implemented by $\alpha\to 2-\alpha$. This means that the scaling limit for the two choices of sign for $\mu$ are related by
\begin{equation}
\label{musign}
Z_j^{(\mu\leq 0)}(q,\alpha)=Z_j^{(\mu\geq 0)}(q,2-\alpha)\,.
\end{equation}
Note also that ${Z(q,\alpha)=Z(q,2-\alpha)}$, so the \emph{total} spectrum of $H_b$ does not depend on the sign of $\mu$ in the scaling limit, but the spectrum \emph{in each sector} does. 

To complete the analysis, let us explain what happens at the endpoints of our fundamental $\alpha$-domain $]0,2[$, corresponding to $y\to\pm\infty$, and also at the special point $\alpha=1$ for which $y=0$ and $\VV_\alpha$ is reducible but indecomposable. 

Let us start with $\alpha=1$. At this value we have
\begin{equation*}
b=\frac{1}{2}(1+\sigma_0^z)-\sigma_0^-\sigma_1^+ \,,
\end{equation*}
so $\ket{\downarrow}\otimes\Hs$ is a proper subspace of $H_b$. Moreover, for any $\ket{v}\in\ket{\uparrow}\otimes\Hs$,
\begin{equation*}
H_b\ket{v}=(-\mu+\Id_{\CC^2}\otimes H_{\mathrm{XX}})\ket{v}+\ket{w} \,,
\end{equation*}
where $\ket{w}\in\ket{\downarrow}\otimes\Hs$. In other words, with respect to the direct-sum decomposition $\Hb=\ket{\downarrow}\otimes\Hs\oplus\ket{\uparrow}\otimes\Hs$, $H_b$ is of the form
\begin{equation}
\label{Hb-block}
H_b=
\begin{pmatrix}
H_{\mathrm{XX}} & *\\
0 & H_{\mathrm{XX}}-\mu
\end{pmatrix}\,.
\end{equation}
Therefore, if $\mu\neq 0$, that is, if there is a finite energy gap between the two diagonal blocks, the scaling limit of $H_b$ will be exactly that of $H_{\mathrm{XX}}$ on an even number of sites (with surface energy decreased by $\mu$ if $\mu>0$), and if $\mu=0$, that of two copies of $H_{\mathrm{XX}}$ (which is obvious from the definitions). It is known \cite{PASQUIER1990523, Gainutdinov_2014} that in the continuum the XX spin chain is described by the partition function (recall that we assumed that $N$ is even)
\begin{equation*}
Z_{\mathrm{XX}}(q)=q^{-\frac{c}{24}}\sum_{j\geq 0}(2j+1)\frac{q^{h_{1,1+2j}}-q^{h_{1,-1-2j}}}{P(q)}=\frac{4q^{-\frac{c}{24}}}{P(q)}\sum_{j\geq 0}q^{\frac{j(j+1)}{2}}\,.
\end{equation*}
On the other hand,
\begin{equation*}
Z(q,\alpha=1) =\frac{2q^{-\frac{c}{24}}}{P(q)}\left(1+\sum_{j\geq 1}\left(q^{\frac{j(j-1)}{2}}+q^{\frac{j(j+1)}{2}}\right)\right)=\frac{4q^{-\frac{c}{24}}}{P(q)}\sum_{j\geq 0}q^{\frac{j(j+1)}{2}} \,,
\end{equation*}
so we have, indeed, $Z_{\mathrm{XX}}(q)=Z(q,\alpha=1)$.

Now let us consider the limits $\alpha\to 0^+,2^-$. By \eqref{musign}, it is sufficient to study only one of these endpoints, say $\alpha=0$, as the behaviour for the other one can be recovered by changing the sign of $\mu$. 

In the limit $\alpha\to 0^+$, $b$ diverges, but we can still obtain a well-defined Hamiltonian by rescaling $\mu=\frac{2}{\pi}\tilde{\mu}\alpha$ so that the boundary term $-\mu b$ converges to a finite limit. It is then easy to see from \eqref{leftblob} that
\begin{equation*}
\lim_{\alpha\to 0}-\mu b=\tilde{\mu}\left[\frac{1}{2}(\sigma_0^x\sigma_{1}^x+\sigma_0^y\sigma_{1}^y)+\frac{i}{2}(\sigma_{1}^z-\sigma_0^z)\right]:=-\tilde{\mu}e_0 \,,
\end{equation*}
so $H_b$ just becomes the XX Hamiltonian on $N+1$ sites with an inhomogeneous coupling constant $\tilde{\mu}$ on the first two sites. This is quite natural since at $\alpha=0$, $\VV_\alpha$ becomes isomorphic to the fundamental $\CC^2$ representation of $\Uqu$, so the boundary coupling $b$ has to be proportional to the bulk coupling, as it is the only $\Uqu$ intertwiner of $\CC^2\otimes\CC^2$, up to a constant shift.

Because the scaling limit only depends on the sign of $\mu$, it is sufficient to take $\tilde{\mu}=\pm 1$. If $\tilde{\mu}=1$, $H_b$ just becomes the XX Hamiltonian on $N+1$ sites. The scaling limit of $H_\mathrm{XX}$ on an odd number of sites (recall that we assumed $N$ even) was carefully treated in \cite{Gainutdinov_2014} and is given by
\begin{equation*}
Z_\mathrm{XX}^\mathrm{odd}(q)=\frac{2q^{-\frac{1}{24}}}{P(q)}\left(1+2\sum_{j\geq 1}q^{\frac{j^2}{2}}\right) \,.
\end{equation*}
We see that $Z_\mathrm{XX}^\mathrm{odd}(q)=Z(q,0)$, which is consistent with our argument. If $\tilde{\mu}=-1$ or, equivalently, $\alpha\to 2^-$, $\mu\sim\frac{2}{\pi}(2-\alpha)$, the scaling limit of $H_b$ is a priori different and has not been computed independently in the literature to our knowledge. However, carefully following the computations of this section, it is not too hard to see that the relation
\begin{equation*}
Z^{(\mu\leq 0)}(q,\alpha)=Z^{(\mu\geq 0)}(q,2-\alpha)
\end{equation*}
still holds for $\alpha=0^+,2^-$, provided of course that we rescale $\mu$ as discussed above. Since $Z(q,0)=Z(q,2)$, this means that the scaling limit of the XX model on an odd number of sites (which is $N+1$) with an inhomogeneous coupling on the first two sites does not depend on the coupling constant $\tilde{\mu}$ and is therefore the same as the scaling limit of the homogeneous XX model on odd number of sites. For an even number of sites this property is also true but somewhat more obvious as $Z(q,1^+)=Z(q,1^-)$.

As the spectrum of $H_b$ and $Z(q,\alpha)$ only depends on $\alpha$ modulo $2$, this completes the study for all $\alpha\in\RR$ and even $N$.

For odd $N$, we need in principle to modify the expansion \eqref{dl} and also be careful as to whether the ``middle" mode $\theta^\dagger_{(N+1)/2}$ decreases or increases the energy depending on the value of $\alpha$. This computation can be performed straightforwardly, but let us present a faster argument. From the discussion above, we have seen that taking the limit $\alpha\to 0^+$ (while keeping $\mu$ positive and of order $\alpha$) amounts to considering an XX spin chain of size $N+1$ while $\alpha=1$ corresponds to an XX spin chain of size $N$. Therefore, we must have $Z^\mathrm{odd}(q,1)=Z(q,0)$ and $Z^\mathrm{odd}(q,0)=Z(q,1)$. The obvious guess for $Z^\mathrm{odd}(q,\alpha)$ is then
\begin{equation}
\label{zodd}
Z^\mathrm{odd}(q,\alpha)=Z(q,\alpha-1)=\frac{2q^{\frac{1}{12}}}{P(q)}\sum_{j\in\ZZ}q^{\frac{(\alpha-1-2j)^2-1}{8}}=\frac{2q^{-\frac{c}{24}}}{P(q)}\sum_{j\in\frac{1}{2}+\ZZ}q^{h_{\alpha,\alpha+2j}}\,.
\end{equation}
This is indeed the correct expression and can be checked by direct computation. Note also that $Z^\mathrm{odd}(q,\alpha)$ only differs from $Z(q,\alpha)$ by the fact that the rightmost sum in \eqref{zodd} now runs over half-integer and not integer~$j$. This is not a coincidence and has a natural explanation in terms of lattice algebras (see Section~\ref{blobsec}).

\subsection{The $(\eta,\xi)$ ghost system}
\label{sfsec}

As the generating functions $Z_j(q,\alpha)$ exactly match the characters of Virasoro Verma modules of weight $h_{\alpha,\alpha+2j}$, we would like to identify the scaling limit of our spin chain with a known conformal field theory. Consider the action \eqref{exaction} of the $(\eta,\xi)$ ghost system
\begin{equation*}
S[\eta,\xi,\bar{\eta},\bar{\xi}]=\frac{1}{2\pi}\int\ddd^2z\left(\eta\bar{\partial}\xi+\bar{\eta}\partial\bar{\xi}\right)\,.
\end{equation*}
This model was originally introduced in the context of string theory \cite{Friedan:1985ge}, and it is an important example of a logarithmic CFT \cite{LCFT2013}. It is also related to the $\mathsf{GL}(1|1)$ WZNW model \cite{gl112006, gl112009}. It was shown in \cite{Gainutdinov_2014} that the $\Uqu$-invariant open XX spin chain with Hamiltonian $H_{\mathrm{XX}}$ gives this model in the continuum limit.\footnote{Actually, a more conceptually accurate description would be in terms of symplectic fermions \cite{Kausch:1995py, Kausch_2000}, but for our case of twisted boundary conditions this will not make a difference.} We would like to generalise this result to the twisted XX model considered here.

The $(\eta,\xi)$ system has the energy momentum tensor
\begin{equation}
\label{emten}
T(z)=-:\eta(z)\partial\xi(z):\,,
\end{equation}
where we used the standard normal ordering, and a similar expression for the anti-chiral part $\bar{T}(\bar{z})$.
Additionally, it has a $\mathsf{U}(1)$ symmetry with (chiral) Noether current
\begin{equation*}
J(z)=:\xi(z)\eta(z):\,.
\end{equation*}
The field $\xi(z)$ increases the associated charge $\mathsf{Q}$ by $1$, whereas $\eta(z)$ decreases it by $1$. The canonical anti-commutation relation reads
\begin{equation}
\label{anticomxe}
\{\xi(z),\eta(w)\}=\delta(z-w)\,,
\end{equation}
and similarly for the anti-chiral fermionic fields.

Since our spin-chain system has the geometry of a strip, we need to define this theory on the upper-half of the complex plane and prescribe some conformally invariant boundary conditions on the real axis, for which there are many possible choices \cite{Kausch:1995py, Kausch_2000, bc2006, Creutzig_2008}. Let us simply take\footnote{We use $\tau$ to denote the twist, since its usual notation $\mu$ is already used for the boundary coupling. This $\tau$ should not be confused with the modular parameter, which we do not use in this paper.}
\begin{equation}
\label{bc}
\begin{gathered}
\eta(x)=\bar{\eta}(x)\,, \quad \xi(x)=\bar{\xi}(x)\quad\text{for}~x\in\RR_{>0}\,,\\
\eta(x)=e^{-2i\pi\tau}\bar{\eta}(x)\,,\quad \xi(x)=e^{2i\pi\tau}\bar{\xi}(x)\quad \text{for}~x\in\RR_{<0}\,,
\end{gathered}
\end{equation}
for $\tau\in\RR/\ZZ$. Note that they preserve the energy-momentum tensor \eqref{emten}. This is equivalent to defining the theory on $\CC^*$, but prescribing a fixed monodromy around $0$ to the fields $(\eta,\xi,\bar{\eta},\bar{\xi})$. It is easy to see that
\begin{equation}
\label{mono}
\begin{gathered}
\eta(e^{2i\pi}z)=e^{-2i\pi\tau}\eta(z),\quad \xi(e^{2i\pi}z)=e^{2i\pi\tau}\xi(z)\,,\\
\bar{\eta}(e^{2i\pi}z)=e^{2i\pi\tau}\bar{\eta}(z),\quad\bar{\xi}(e^{2i\pi}z)=e^{-2i\pi\tau}\bar{\xi}(z)\,.
\end{gathered}
\end{equation}
Therefore, to compute the partition function of the $(\eta,\xi)$ system on the upper-half plane with boundary conditions \eqref{bc}, we can use the results of~\cite{Kausch:1995py,Kausch_2000}. They state that the chiral partition function (which is the one we need to consider in the context of a boundary CFT \cite{Cardy2004BoundaryCF}) of the twisted theory \eqref{mono} takes the form
\begin{equation*}
\begin{aligned}
Z_{(\eta,\xi)}(q,\tau,g) & =q^{-\frac{c}{24}-\frac{\tau(1-\tau)}{2}}\prod_{n=0}^{+\infty}(1+gq^{n+\tau})(1+g^{-1}q^{n+1-\tau})\\
& =\frac{q^{-\frac{c}{24}}}{P(q)}\sum_{j\in\ZZ} g^j q^{h_{1-2\tau,1-2\tau+2j}} \,,
\end{aligned}
\end{equation*}
where $g$ is a formal parameter keeping track of the value of the $\mathsf{U}(1)$ charge in each sector. More abstractly, one can think of $Z_{(\eta,\xi)}(q,\tau,g)$ as a $\Vir\times \mathsf{U}(1)$-character. Note also that
\begin{equation*}
Z_{(\eta,\xi)}(q,\tau+1,g)=g^{-1}Z_{(\eta,\xi)}(q,\tau,g)\,.
\end{equation*}
One might absorb the $g^{-1}$ factor by rescaling the partition function as $Z_{(\eta,\xi)}\to g^\tau Z_{(\eta,\xi)}$, so that it only depends $\tau\in\RR/\ZZ$, but we will not do this here.

To make an explicit connection between this theory and our spin chain, introduce the lattice fields
\begin{equation*}
\begin{aligned}
\chi^+(z)^{<M} & :=-\sum_{-M\leq k\leq M}\chi^+_{k+\tau}z^{-k-1+\tau}\,,\\
\chi^-(z)^{<M} & :=\sum_{-M\leq k\leq M}\chi^-_{k-\tau}z^{-k-1-\tau}\,,
\end{aligned}
\end{equation*}
for some $\tau\notin\ZZ$ with modes
\begin{equation}
\label{chimodes}
\chi^+_{k+\tau}:=\frac{\sin(\xi_{N/2+k})}{\sqrt{\pi}}\theta^\dagger_{N/2+k}\,,\qquad \chi^-_{k-\tau}:=\frac{-i\sin(\xi_{N/2-k})}{\sqrt{\pi}}\theta_{N/2-k}
\end{equation}
and a cut-off $M\in\NN$. Here, $\xi_k$ are solutions of \eqref{simpeq}, and $\theta^\dagger_k$, $\theta_k$ the associated modes shifting the energy by $\pm 2\cos(\xi_k)$. It is clear from the definitions that $\chi^\pm(z)^{<M}$ has monodromy $e^{\pm 2i\pi\tau}$. Set
\begin{equation*}
\begin{aligned}
\xi(z)^{<M} & :=\sum_{-M\leq k\leq M}\frac{\chi^+_{k+\tau}}{k-\tau}z^{-k+\tau}\,,\\
\eta(z)^{<M} & :=\sum_{-M\leq k\leq M}\chi^-_{k-\tau}z^{-k-1-\tau}\,.
\end{aligned}
\end{equation*}
Note that $\partial\xi(z)^{<M}=\chi^+(z)^{<M}$. We claim that in the scaling limit we have, for all $\alpha\notin 1+2\ZZ$,
\begin{equation*}
\begin{aligned}
\lim_{M\to\infty}\lim_{N\to\infty}\xi(z)^{<M}=\xi(z)\,,\\
\lim_{M\to\infty}\lim_{N\to\infty}\eta(z)^{<M}=\eta(z)\,,
\end{aligned}
\end{equation*}
with
\begin{equation*}
\tau=\frac{1-\alpha}{2}\,.
\end{equation*}
This is motivated by the fact that
\begin{equation}
\label{anticomsc}
\lim_{M\to\infty}\lim_{N\to\infty}\{\xi(z)^{<M},\eta(w)^{<M}\} =\frac{1}{w}\sum_{k\in\ZZ}(z/w)^{k+\tau}:=\delta_\tau(z,w) \,,
\end{equation}
where $\delta_\tau(z,w)$ is the $\tau$-twisted delta function. This is the natural generalisation of the canonical anti-commutation relation \eqref{anticomxe} (which we recover for $\tau=0\in\RR/\ZZ$) to fields with non-trivial monodromy. The proof of \eqref{anticomsc} is given in Appendix \ref{anticomp} together with some elementary properties of $\delta_\tau(z,w)$.

We can thus identify the scaling limits of the lattice fields $\chi^+(z)^{<M}$, $\xi(z)^{<M}$, and $\chi^-(z)^{<M}=\eta(z)^{<M}$ with respectively $\partial\xi(z)$, $\xi(z)$, and $\eta(z)$. Moreover, $\chi^\pm(z)^{<M}$ shifts the eigenvalue of $\HH$ by $\pm 2$. Recalling that $\xi(z)$ and $\eta(z)$ shift the $\mathsf{U}(1)$ charge $\mathsf{Q}$ by $+1$ and $-1$ respectively, let us set
\begin{equation}
\mathsf{Q}=\frac{1}{2}(\HH-\alpha)\,.
\end{equation}
Then 
\begin{equation*}
\mathsf{Q}|_{W_{N/2+j}}=j-\frac{1}{2}\,,\qquad \mathsf{Q}|_{\overline{W}_{N/2+j}}=j+\frac{1}{2}\,,
\end{equation*}
and so
\begin{equation}
\begin{aligned}
Z(q,\alpha,g) & :=\lim_{M\to\infty}\lim_{N\to\infty}\tr_{\Hb}^{<M} g^\mathsf{Q}q^{\frac{N}{\pi v_F}(H_b-Ne_b-E_s)}\\
& \,=\left(g^\frac{1}{2}+g^{-\frac{1}{2}}\right)\frac{q^{-\frac{c}{24}}}{P(q)}\sum_{j\in\ZZ} g^j q^{h_{\alpha,\alpha+2j}}\\
&\, =\left(g^\frac{1}{2}+g^{-\frac{1}{2}}\right)Z_{(\eta,\xi)}\left(q,\tau=\frac{1-\alpha}{2},g\right)\,.
\end{aligned}
\end{equation}
The additional factor of $g^\frac{1}{2}+g^{-\frac{1}{2}}$ comes from the $\Uqu$ symmetry which creates two copies of the same sector, with the $\mathsf{U}(1)$ charge shifted by $\pm\frac{1}{2}$. One can get rid of it by restricting to $\bigoplus_j W_{N/2+j}\subset\Hb$, and redefining the charge as $\mathsf{Q}\to \mathsf{Q}+\frac{1}{2}$ (or to $\bigoplus_j \overline{W}_{N/2+j}$, with $\mathsf{Q}\to \mathsf{Q}-\frac{1}{2}$).

Therefore, we can identify the scaling limit of our twisted XX spin chain with the $(\eta,\xi)$ ghost CFT on the upper half-plane with boundary conditions given by \eqref{bc}, with the identification $\tau=\frac{1-\alpha}{2}$.

\section{Spectrum of $H_{2b}$ for $\qa=i$ and its scaling limit}
\label{spech2b}

Let us now turn to the two-boundary Hamiltonian, still at $\qa=i$. Recall that it is defined on the Hilbert space
\begin{equation*}
\Hbb=\VV_{\alpha_l}\otimes\left(\CC^2\right)^{\otimes N}\otimes\VV_{\alpha_r}
\end{equation*}
by the Hamiltonian
\begin{equation*}
H_{2b}=-\mu_lb_l-H_{\mathrm{XX}}-\mu_r b_r \,,
\end{equation*}
where $b_l:=b$ is still given by \eqref{leftblob},   with $\alpha$ replaced by $\alpha_l$, whereas 
\begin{equation*}
b_r=\frac{1}{2}(1+\sigma_{N+1}^z)-\cot{\frac{\pi\alpha_r}{2}}\left(\sigma_N^-\sigma_{N+1}^++\frac{e^{\frac{i\pi\alpha_r}{2}}}{\cos{\frac{\pi\alpha_r}{2}}}\sigma_N^+\sigma_{N+1}^-+\frac{i}{2}(\sigma_{N+1}^z-\sigma_N^z)\right).
\end{equation*}
is given by \eqref{blobrqroot}. In Jordan-Wigner form, recall the fermions $c_j$ and $c^{\dagger}_j$ introduced in~\eqref{eq:def-c},
\begin{equation*}
b_r=c^\dagger_{N+1}c_{N+1}+\cot{\frac{\pi\alpha}{2}}\left(c_N c_{N+1}^\dagger+\frac{e^{\frac{i\pi\alpha_r}{2}}}{\cos{\frac{\pi\alpha_r}{2}}}c_{N+1}c_{N}^\dagger+i(c^\dagger_{N}c_{N}-c^\dagger_{N+1}c_{N+1})\right)\,.
\end{equation*}

The spectrum of this model can be computed using the same method as for $H_b$, by looking for operators $\theta^\dagger$ diagonalising the adjoint action of $H_{2b}$
\begin{equation}
\label{adjdag2b}
[H_{2b},\theta^\dagger]=\lambda \theta^\dagger
\end{equation}
for some $\lambda\in\mathbb{C}$. Taking $\theta^\dagger$ of the form
\begin{equation*}
\theta^\dagger=\sum_{j=0}^{N+1}A_jc^\dagger_j \,,
\end{equation*}
and eliminating $A_{0}$ and $A_{N+1}$ from the system of equations given by \eqref{adjdag2b}, we obtain
\begin{equation}
\label{speceqA2b}
\left\{
\begin{array}{lll}
A_2-(i-\chi_l(\lambda))A_1=\lambda A_1 \,, \\
A_{j-1}+A_{j+1}=\lambda A_j \qquad \forall~2\leq j\leq N-1 \,,\\
A_{N-1} +(i+\chi_r(\lambda))A_N=\lambda A_N \,,
\end{array}
\right.
\end{equation}
where
\begin{equation*}
\chi_l(\lambda)=\frac{i\lambda\cos{\frac{\pi\alpha_l}{2}}}{ie^{-\frac{i\pi\alpha_l}{2}}+\frac{\lambda}{\mu}\sin{\frac{\pi\alpha_l}{2}}}
\end{equation*}
and
\begin{equation*}
\chi_r(\lambda)=\frac{i\lambda\cos{\frac{\pi\alpha_r}{2}}}{ie^{-\frac{i\pi\alpha_r}{2}}-\frac{\lambda}{\mu}\sin{\frac{\pi\alpha_r}{2}}}\,.
\end{equation*}

Taking $A_j=a_+x^j+a_-x^{-j}$ for $x\in\mathbb{C}\backslash\{-1,1\}$, \eqref{speceqA2b} reduces to
\begin{equation*}
\left\{
\begin{array}{ll}
x(i+x^{-1}-\chi_l(\lambda))a_++x^{-1}(i+x-\chi_l(\lambda))a_-=0 \,, \\
x^N(i-x+\chi_r(\lambda))a_++x^{-N}(i-x^{-1}+\chi_r(\lambda))a_-=0 \,,
\end{array}
\right.
\end{equation*}
with $\lambda=x+x^{-1}$. The spectral equation is obtained by imposing that the determinant of this system vanishes. After a tedious computation, and dividing by $\lambda(x-x^{-1})$ as before, we obtain the polynomial equation
\begin{equation}
\label{2bwaveqU}
U_{N+1}+(\mu_l+\mu_r)U_{N}+(\mu_l\mu_r+t_l+t_r)U_{N-1}+(\mu_l t_r+\mu_r t_l)U_{N-2}+t_lt_rU_{N-3}=0
\end{equation}
in the variable $\lambda/2$, where
\begin{equation*}
t_l:=1-\mu_l y_l=1-\mu_l \cot\frac{\pi\alpha_l}{2}\,,\quad t_r:=1-\mu_r y_r=1-\mu_r\cot\frac{\pi\alpha_r}{2}\,.
\end{equation*}
It is the two-boundary analogue of \eqref{waveqU}. Note that it has real coefficients for real $\mu_{l/r}, y_{l/r}$, that it is invariant under the left/right exchange $(\mu_{l/r}, y_{l/r})\leftrightarrow (\mu_{r/l}, y_{r/l})$, and that if we set $\mu_r$ or $\mu_l$ to $0$ we recover \eqref{waveqU} up to a factor $\lambda$ (which accounts for the fact that $H_{2b}$ then becomes equivalent to two copies of $H_b$).

Again, parametrising $\lambda=2\cos\xi$, the polynomial equation \eqref{2bwaveqU} can be rewritten as
\begin{equation}
\label{phieq}
\sin\left(N\xi+\varphi_{2b}(\xi)\right)=0\,,\qquad \xi\in]0,\pi[ \,,
\end{equation}
with
\begin{equation*}
\varphi_{2b}(\xi)=\arccot{\frac{\mu_l\mu_r+t_l+t_r+\cos(2\xi)(1+t_lt_r)+\cos(\xi)(\mu_l(1+t_r)+\mu_r(1+t_l))}{\sin(\xi)\left(2\cos(\xi)(1-t_lt_r)+\mu_l(1-t_r)+\mu_r(1-t_l)\right)}}\,.
\end{equation*}
Here we are working with the generalised $\arccot$ function taking values in $[0,2\pi]$. The function $\varphi_{2b}$ is then defined either by continuity, or by tracking the signs of the numerator and denominator to figure out in which quadrant of the $[0,2\pi]$ circle it must take its values. At fixed $y_l,y_r>0$ and for sufficiently small $\mu_l,\mu_r>0$, $\varphi_{2b}$ is a strictly increasing function of $\xi$ with $\varphi_{2b}(0)=0$ and $\varphi_{2b}(\pi)=2\pi$, meaning that \eqref{phieq} has exactly $N+1$ real solutions (with real associated momenta). This provides $N+1$ operators $\theta_k^\dagger$, with $1\leq k\leq N+1,$ whose action on $\ket{\downarrow}^{\otimes N+2}$, together with the zero-mode $\theta_0^\dagger$ coming from the $\Uqu$ symmetry, generates a complete basis of eigenstates of $H_{2b}$. Note that the $\Uqu$-decomposition of the Hilbert space is now
\begin{equation*}
\Hbb=\bigoplus_{k=0}^{N+1} {N+1\choose k} \VV_{\alpha_l+\alpha_r-N-1+2k} \,,
\end{equation*}
and that the new $k$-fermion subspaces $W_k$ and $\overline{W}_k$ are of dimension ${N+1\choose k}$.

To find the scaling limit of $H_{2b}$, we do not need to redo all the computations, because from the analysis of Section \ref{sclim} it is clear that the ground state energies in each sector $W_{N/2+j}$ only depend on $\phi(\frac{\pi}{2})=\pi-\varphi(\frac{\pi}{2})$, and that the combinatorics for the excitations is still valid for our new fermionic modes. There are two minor subtleties, however. First, to obtain the analogue of \eqref{neweq},
\begin{equation*}
N\xi+\phi_{2b}(\xi)=k\pi \,,
\end{equation*}
we now have to set $\phi_{2b}:=2\pi-\varphi_{2b}(\pi-\xi)$, because $\varphi_{2b}$ takes values in $[0,2\pi]$. This way $1\leq k\leq N+1$, and the mode $\theta^\dagger_k$ corresponding to a solution $\xi_k$ decreases the energy by $\lambda_k=2\cos{\xi_k}$. Second,
\begin{equation*}
\phi_{2b}\left(\frac{\pi}{2}\right)=2\pi-\arccot\frac{1-y_ly_r}{y_l+y_r}=\pi+\frac{\pi(\alpha_l+\alpha_r)}{2} \,,
\end{equation*}
so for $0<\alpha_l,\alpha_r< 1$ (that is, $0<y_l,y_r$), $\pi<\phi_{2b}\left(\frac{\pi}{2}\right)<2\pi$. This means that $\lambda_{N/2+1}>0$ and $\lambda_{N/2+2}<0$, so the vacuum is now given by $\ket{{\rm vac}_0}=\prod_{k=1}^{N/2+1}\theta_k^\dagger\ket{\downarrow}^{\otimes N+2}$, that is, it belongs to the ($N/2+1$)-fermion sector $W_{N/2+1}$, and not to $W_{N/2}$ as before (again we assume that $N$ is even).

Taking all this into account we obtain\footnote{Note the $j\to j+1$ shift with respect to the definition \eqref{wpartfun}.}
\begin{equation}
\begin{aligned}
Z_{2b,j}(q,\alpha_l,\alpha_r) & :=  \lim_{M\to\infty}\lim_{N\to\infty}\tr_{W_{N/2+j+1}}^{<M} q^{\frac{N}{\pi v_F}(H_{2b}-Ne_b-E_s)}\\
&\, =\frac{2q^{-\frac{1-3(\alpha_{lr}-1)^2}{24}}}{P(q)}q^{\frac{j(j -(\alpha_{lr}-1))}{2}}\\
&\, =Z_j(q,\alpha_l+\alpha_r-1) \,,
\end{aligned}
\end{equation}
where we introduced $\alpha_{lr}:=\alpha_l+\alpha_r$. The full partition function is then given by
\begin{equation}
\label{twobqi}
Z_{2b}(q,\alpha_l,\alpha_r)=\frac{2q^{-\frac{1-3(\alpha_{lr}-1)^2}{24}}}{P(q)}\sum_{j\in\ZZ}q^{\frac{j(j -(\alpha_{lr}-1))}{2}}=Z(q,\alpha_l+\alpha_r-1)\,.
\end{equation}
Note that it only depends on the sum $\alpha_l+\alpha_r$. Using the fact that $Z(q,\alpha+2)=Z(q,\alpha)$ we can also rewrite this partition function in a more suggestive manner,
\begin{equation}
\label{sugg}
\begin{aligned}
Z_{2b}(q,\alpha_l,\alpha_r) & =\frac{q^{-\frac{c}{24}}}{P(q)}\sum_{j\in\ZZ} \left(q^{h_{\alpha_{lr}+1,\alpha_{lr}+1+2j}}+q^{h_{\alpha_{lr}-1,\alpha_{lr}-1+2j}}\right) \\
& =\frac{1}{2}\left(Z(q,\alpha_{lr}+1)+Z(q,\alpha_{lr}-1)\right) \,,
\end{aligned}
\end{equation}
making apparent the analogy with the $\Uqu$-fusion rule 
\begin{equation*}
\VV_{\alpha_l}\otimes\VV_{\alpha_r}\cong\VV_{\alpha_{lr}+1}\oplus\VV_{\alpha_{lr}-1}
\end{equation*}
valid for $\alpha_l, \alpha_r, \alpha_{lr}\notin\ZZ$ \cite{geermodtr0}.

These results are extended by analytic continuation to all values of $y_l,y_r$, that is, $0<\alpha_l,\alpha_r<2$. As before, even if $\mu_l$ and $\mu_r$ do not appear in the final expressions we had to assume that they were positive in our derivations: changing the sign of $\mu_{l/r}$ amounts to replacing $\alpha_{l/r}$ by $2-\alpha_{l/r}$. The special values $\alpha_{l/r}=0,1$ play the same role as for the one-boundary system (see the discussion at the end of Section~\ref{sclim}). Namely, taking $\alpha_{l/r}\to 0^+$, while rescaling $\mu_{l/r}=\frac{2}{\pi}\tilde{\mu}_{l/r}\alpha_{l/r}>0$, amounts to taking a boundary coupling of TL type (that is, proportional to the bulk coupling) on the left, right or both sides of the spin chain whereas setting $\alpha_{l/r}=1$ is equivalent to imposing the usual XX boundary conditions on the left, right or both boundaries. This last property follows from the same kind of arguments as in the one-boundary case. For example, taking $\alpha_r=1$ we obtain the block structure (compare with~\eqref{Hb-block})
\begin{equation*}
H_{2b}=
\begin{pmatrix}
H_b & *\\
0 & H_b-\mu_r
\end{pmatrix}
\end{equation*}
with respect to the decomposition $\Hbb=\Hb\otimes\ket{\downarrow}\oplus\Hb\otimes\ket{\uparrow}$ and similarly for $\alpha_l=1$. This is consistent with the fact that $Z_{2b}(q,\alpha,1)=Z_{2b}(q,1,\alpha)=Z(q,\alpha)$ and $Z_{2b}(q,\alpha,0)=Z_{2b}(q,0,\alpha)=Z^\mathrm{odd}(q,\alpha)$ (recall the one-boundary partition functions in~\eqref{partfun} and~\eqref{zodd}). Moreover, since the scaling limit of the XX spin chain on $N$ and $N+2$ sites must be the same, we naturally have $Z_{2b}(q,1,1)=Z_{2b}(q,0,0)$. Finally, to obtain the scaling limit for odd $N$, we can, as before, just formally replace the summation over integer $j$ in \eqref{twobqi} by a summation over half-integers.

These results are also consistent with the $(\eta,\xi)$ field-theory computation, where boundary conditions \eqref{bc} are now replaced by
\begin{equation}
\label{bc2b}
\begin{gathered}
\eta(x)=e^{-i\pi(1-\alpha_l)}\bar{\eta}(x)\,,\quad \xi(x)=e^{i\pi(1-\alpha_l)}\bar{\xi}(x)\quad\text{for}~x\in\RR_{<0}\,,\\
\eta(x)=e^{i\pi(1-\alpha_r)}\bar{\eta}(x)\,,\quad \xi(x)=e^{-i\pi(1-\alpha_r)}\bar{\xi}(x)\quad \text{for}~x\in\RR_{>0}\,.
\end{gathered}
\end{equation}
Note that these boundary conditions are not manifestly left/right symmetric for $\alpha_l=\alpha_r$. There is no contradiction here, as permuting the left and right boundaries also reverses the orientation of the $\mathsf{U}(1)$ current in the field theory and the sign of the surface term at the level of the spin chain. Therefore, when doing so, one should also conjugate all the phases in \eqref{bc2b}. Taking into account this subtlety, these boundary conditions are thus indeed left/right symmetric when $\alpha_l=\alpha_r$. More abstractly, one can see this symmetry as the field-theoretic realisation of the $\Uqu$ automorphism $(\KK,\KK^{-1},\EE,\FF,\HH)\to (\KK^{-1},\KK,\FF,\EE,-\HH)$. One also checks that taking $\alpha_{r,l}\to 1$ in \eqref{bc2b} we recover the usual boundary conditions of the $(\eta,\xi)$ field theory on the left, right or both boundaries, in agreement with the above analysis for the special values of $\alpha_{l/r}$.

\section{Lattice algebras underlying the one-boundary spin chains}
\label{symprop}

We now go back to the general $\qa$ case, and study in detail, for all $\qa$, the lattice algebra underlying the XXZ spin chain and the $\Uq$-invariant boundary conditions we have constructed in Section~\ref{construction} for arbitrary complex parameters. We will show that the Hilbert space $\Hb$ admits an action of the blob algebra which, for $\qa^\alpha\notin\pm\qa^\ZZ$, is faithful and Schur-Weyl dual to $\Uq$ (or $\Uqu$ if $\qa$ is a root of unity), and also establish the corresponding bimodule decomposition.

\subsection{The XXZ spin chain}
\label{XXZalg}

Recall the expression of the XXZ Hamiltonian $H_{\mathrm{XXZ}}$ in~\eqref{XXZ}. Let us introduce the Hamiltonian densities, or the nearest-neighbour coupling operators,
\begin{equation}\label{eq:ei-Pauli}
\begin{aligned}
e_i & =-\frac{1}{2}\left(\sigma^x_{i}\sigma^x_{i+1}+\sigma^y_{i}\sigma^y_{i+1}+\frac{\qa+\qa^{-1}}{2}(\sigma^z_{i}\sigma^z_{i+1}-1)\right)-\frac{\qa-\qa^{-1}}{4}(\sigma_{i+1}^z-\sigma_{i}^z)\\
& =\Id_{(\CC^2)^{\otimes (i-1)}}\otimes
\begin{pmatrix}
0 & 0 & 0 & 0\\
0 & \qa & -1 & 0\\
0 & -1 & \qam & 0\\
0 & 0 & 0 & 0
\end{pmatrix}
\otimes \Id_{(\CC^2)^{\otimes (N- i-1)}}
\end{aligned}
\end{equation}
for all $1\leq i\leq N-1$ to rewrite
\begin{equation*}
H_{\mathrm{XXZ}}=\frac{\qa+\qa^{-1}}{2}(N-1)-\sum_{i=1}^{N-1}e_i\,.
\end{equation*}
By direct computation, the $e_i$ satisfy the relations
\begin{equation*}
e_i^2=(\qa+\qam)e_i\,,\qquad e_ie_{i\pm1}e_i=e_i\,,\qquad [e_i,e_j]=0\quad\forall~|i-j|\geq 2 \,,
\end{equation*}
making them a representation of the Temperley-Lieb (TL) algebra $\TL$ with loop weight $\delta=\qa+\qam$  \cite{TLalg, PASQUIER1990523}. If we set
\begin{equation*}
\begin{tikzpicture}[scale=0.8]
\draw (0,1) node {$e_i=$};
\draw (1,0) [dnup=0];
\draw (3,0) [dnup=0];
\draw (2,1) node {...};
\draw (4,2) [upup=1];
\draw (4,0) [dndn=1];
\draw (7,1) node {...};
\draw (6,0) [dnup=0];
\draw (8,0) [dnup=0];
\draw (4,-0.5) node {$i$};
\draw (5,-0.5) node {$i+1$};
\end{tikzpicture}
\end{equation*}
these relations are neatly expressed by the graphical rules
\begin{equation*}
\begin{tikzpicture}[scale=0.5]
\draw (-1.5,2) node {$e_i^2=$};
\draw (1,0) [dnup=0];
\draw (0,1) node {...};
\draw (2,2) [upup=1];
\draw (2,0) [dndn=1];
\draw (5,1) node {...};
\draw (4,0) [dnup=0];
\draw (1,2) [dnup=0];
\draw (0,3) node {...};
\draw (2,4) [upup=1];
\draw (2,2) [dndn=1];
\draw (5,3) node {...};
\draw (4,2) [dnup=0];
\draw (6,2) node {=};
\draw (7,2) node {$\delta$};
\draw (9,1) [dnup=0];
\draw (8,2) node {...};
\draw (10,3) [upup=1];
\draw (10,1) [dndn=1];
\draw (13,2) node {...};
\draw (12,1) [dnup=0];
\end{tikzpicture}
\end{equation*}

\begin{equation*}
\begin{tikzpicture}[scale=0.5]
\draw (-2.5,3) node {$e_ie_{i+1}e_i=$};
\draw (0,1) node {...};
\draw (1,0) [dnup=0];
\draw (2,2) [upup=1];
\draw (2,0) [dndn=1];
\draw (4,0) [dnup=0];
\draw (5,0) [dnup=0];
\draw (6,1) node {...};
\draw (0,3) node {...};
\draw (1,2) [dnup=0];
\draw (2,2) [dnup=0];
\draw (3,4) [upup=1];
\draw (3,2) [dndn=1];
\draw (5,2) [dnup=0];
\draw (6,3) node {...};
\draw (0,5) node {...};
\draw (1,4) [dnup=0];
\draw (2,6) [upup=1];
\draw (2,4) [dndn=1];
\draw (4,4) [dnup=0];
\draw (5,4) [dnup=0];
\draw (6,5) node {...};

\draw (7,3) node {=};

\draw (8,3) node {...};
\draw (9,2) [dnup=0];
\draw (10,4) [upup=1];
\draw (10,2) [dndn=1];
\draw (12,2) [dnup=0];
\draw (13,2) [dnup=0];
\draw (14,3) node {...};
\end{tikzpicture}
\,.
\end{equation*}
In other words, the $e_i$ expression~\eqref{eq:ei-Pauli} in terms of Pauli matrices provides a representation of  the  Temperley-Lieb algebra $\TL$ (which is known to be faithful).
This way, $H_{\mathrm{XXZ}}$ can be interpreted as an abstract element of $\TL$ represented on the spin chain $\Hs$.

It is known that the actions of $\TL$ and $\Uq$ on $\Hs$ are mutual maximal centralisers \cite{Jimbo1986AQO, goodman}, even at roots of unity \cite{martin1992commutants}. If moreover $\qa$ is generic (not a root of unity), $\TL$ and $\Uq$ are both semi-simple and we can use Schur-Weyl duality to decompose the Hilbert space as a $(\TL,\Uq)$-bimodule
\begin{equation*}
\Hs=\bigoplus_{j=0}^{N/2}\TT_j\otimes\CC^{2j+1} \,,
\end{equation*}
where $\CC^{2j+1}$ are spin-$j$ representations of $\Uq$, whereas $\TT_j$ are the standard $\TL$-modules with $2j$ through lines (if $N$ is odd, $j$ is a half-integer). The standard modules $\TT_j$ are irreducible and of dimension $\binom{N}{N/2+j} - \binom{N}{N/2+j+1}$ (see more details on their definition in~\cite{alma9912017293902959}). As $H_{\mathrm{XXZ}}$ commutes with $\Uq$, seen as an element of $\TL$, its action can be restricted to $\TT_j$ while the $\CC^{2j+1}$ components are just multiplicity spaces. Our goal is to generalise this picture to one-boundary and two-boundary Hamiltonians $H_b$ and $H_{2b}$ introduced in~\eqref{eq:Hb} and~\eqref{eq:H2b}, respectively.

\subsection{The one-boundary spin chain and the blob algebra}
\label{blobsec}

During the construction of the one-boundary Hamiltonian $H_b$ in Section~\ref{sec:Hb-gen-q}, we defined an additional $\Uq$-invariant generator $b:=b_+$ in~\eqref{blobqgen}, and so the Temperley-Lieb algebra formalism is insufficient to deal with it.

Let us recall the blob algebra introduced in~\cite{Martin:1993jka} and denoted by $\Blob$. The algebra depends on  two complex parameters, $\delta$ and $y$, denoting respectively the loop weight and the blob weight. It is defined by adding a generator $b$ -- called the blob -- to $\TL$ with the additional relations
\begin{equation}
\label{eq:blob-rel}
b^2=b\,,\qquad e_1 b e_1 = y e_1\,, \qquad [b,e_i]=0 \quad \text{for} \quad 2\leq i\leq N-1.
\end{equation}
This is a finite-dimensional algebra, and its dimension  does not depend on $\qa$ or~$y$.
Graphically, $b$ is represented by
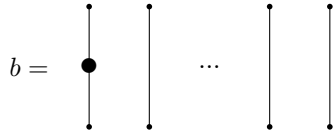
\begin{equation*}
\begin{tikzpicture}[scale=0.8]
\draw (0,1) node {$b=$};
\draw (1,0) [dnup=0];
\draw (1,1) node {$\bbullet$};
\draw (2,0) [dnup=0];
\draw (3,1) node {...};
\draw (4,0) [dnup=0];
\draw (5,0) [dnup=0];
\end{tikzpicture}
\end{equation*}
and the rules~\eqref{eq:blob-rel} mean that
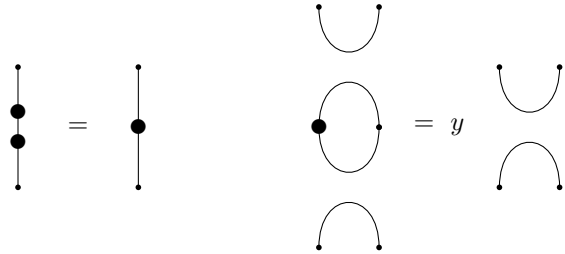
\begin{equation*}
\begin{tikzpicture}[scale=0.8]
\draw (0,0) [dnup=0];
\draw (0,0.75) node {$\bbullet$};
\draw (0,1.25) node {$\bbullet$};
\draw (1,1) node {$=$};
\draw (2,0) [dnup=0];
\draw (2,1) node {$\bbullet$};

\draw (5,-1) [dndn=1];
\draw (5,1) [upup=1];
\draw (5,1) [dndn=1];
\draw (5,3) [upup=1];
\draw (5,1) node {$\bbullet$};
\draw (7,1) node {$=~y$};
\draw (8,2) [upup=1];
\draw (8,0) [dndn=1];

\end{tikzpicture}
\end{equation*}
One also introduces the anti-blob $b_-=1-b$, represented by
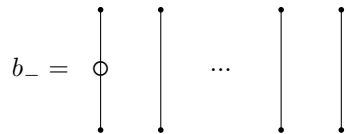
\begin{equation*}
\begin{tikzpicture}[scale=0.8]
\draw (0,1) node {$b_-=$};
\draw (1,0) [dnup=0];
\draw (1,1) node {$\ccirc$};
\draw (2,0) [dnup=0];
\draw (3,1) node {...};
\draw (4,0) [dnup=0];
\draw (5,0) [dnup=0];
\end{tikzpicture}
\end{equation*}
which satisfies relations \eqref{eq:blob-rel} but  with the blob weight $y$ replaced by $\delta-y$. Whenever it is convenient we will denote by $b_\pm$ simultaneously the blob and the anti-blob, and by $y_\pm$ their respective weights.

One can verify by direct computation that the projectors $b_\pm$ that we defined in \eqref{blobqgen} and \eqref{blobqroot} satisfy the relations~\eqref{eq:blob-rel} with 
\begin{equation}
\label{eq:y-alpha}
y_\pm:=\frac{[\alpha\pm 1]_\qa}{[\alpha]_\qa}=\frac{\qa^{\alpha\pm 1}-\qa^{-\alpha\mp 1}}{\qa^{\alpha}-\qa^{-\alpha}}\,,
\end{equation}
so $H_b= -\mu b + H_{\mathrm{XXZ}}$ can be interpreted as an abstract element of $\Blob$ represented on the spin chain $\Hb=\VV_\alpha\otimes (\CC^2)^{\otimes N}$. We will call $y:=y_+$ generic if it is as in~\eqref{eq:y-alpha} for $\qa^{\alpha}\notin\pm\qa^\ZZ$.

A complete classification of the irreducible $\Blob$-modules for generic $y$ was performed in \cite{Martin:1993jka} (see also \cite{dubail:tel-00555624}). The result is that for all $\qa$, including roots of unity cases, and generic $y$, $\Blob$ is semi-simple and has $N+1$ non-equivalent irreducible representations known as standard modules. They are constructed using a special diagrammatical basis encoded graphically by through lines and nested arcs -- similarly to the TL case -- but now also decorated by blobs and anti-blobs. Modules whose basis diagrams have $1\leq 2j\leq N$ through lines and such that the leftmost through line carries a blob (resp. anti-blob) are called standard blob (resp. anti-blob) modules and will be denoted $\WW_j$ (resp. $\WW_{-j}$). The so-called vacuum standard module with no through lines does not depend on the blob/anti-blob configuration and will be denoted $\WW_0$. The $N+1$ standard modules $\WW_j$ are thus labelled by $-N/2\leq j\leq N/2$, with $j$ integer if $N$ is even and $j$ half-integer if $N$ is odd (see more details in \cite{Martin:1993jka}). Below, we will also use the notation  $\WW^N_j$ to emphasise the number $N$ of sites used. 

Let us first assume that $\qa$ is generic (not a root of unity). Our main algebraic claim here is the following.

\begin{prop}
\label{blobrepqgen}
Let $\qa\in\CC\backslash e^{i\pi\QQ}$, $\qa^\alpha\in\CC\backslash\{\pm1\}$, and $N\in\mathbb{N}^*$. Then $\Hb:=\VV_{\alpha}\otimes (\CC^2)^{\otimes N}$, where $\VV_\alpha$ is the Verma module defined in~\eqref{alpharepgen}, carries a representation of $\Blob$ with
\begin{equation}
\label{blobweight}
\delta=[2]_\qa\,,\qquad y=\frac{[\alpha+1]_\qa}{[\alpha]_\qa} \,,
\end{equation}
and
\begin{equation}
\label{blobrepex}
\begin{gathered}
b=\frac{1}{\{\alpha\}}
\begin{pmatrix}
 -\qam \KK^{-1}+\qa^{\alpha} & \{1\}\FF\\
 \qa\{1\}\KK^{-1}\EE & \qa\KK^{-1}-\qa^{-\alpha}
\end{pmatrix}\,, \\
e_i =-\frac{1}{2}\left(\sigma^x_{i}\sigma^x_{i+1}+\sigma^y_{i}\sigma^y_{i+1}+\frac{\qa+\qa^{-1}}{2}(\sigma^z_{i}\sigma^z_{i+1}-1)\right)-\frac{\qa-\qa^{-1}}{4}(\sigma_{i+1}^z-\sigma_{i}^z)\,.
\end{gathered}
\end{equation}
Moreover, if $\qa^\alpha\notin\pm\qa^\ZZ$, this representation is faithful, $\Uq$ and $\Blob$ are mutual maximal centralisers and we have the decomposition of $(\Blob,\Uq)$-bimodules
\begin{equation}
\label{bidecompb}
\Hb=\bigoplus_{j=-N/2}^{N/2} \WW_j\otimes\VV_{\alpha+2j}\,.
\end{equation}
\end{prop}

\begin{proof}
As was already mentioned above, using the explicit expressions of $b$ and the $e_i$ one can check by direct computation  that they indeed satisfy the defining relations of the blob algebra \eqref{eq:blob-rel} with parameters as in \eqref{blobweight}. By construction, they also commute with the action of $\Uq$. It remains to show that if $\qa^\alpha\notin\pm\qa^\ZZ$ this representation is faithful, that it is indeed the full centraliser of $\Uq$, and that we have the decomposition \eqref{bidecompb}.

Let us start with \eqref{bidecompb}. In \cite{Martin:1993jka}, it is explained how to inductively construct the irreducible $\Blob$-modules $\WW^N_j$ from the $\Blobb$-modules $\WW^{N-1}_{j-\frac{1}{2}}$ and $\WW^{N-1}_{j+\frac{1}{2}}$. An immediate consequence is that $\dim\WW_j={N\choose N/2+j}$ and 
\begin{equation*}
\Ress_{\Blobb} \WW^N_j=\WW^{N-1}_{j-\frac{1}{2}}\oplus\WW^{N-1}_{j+\frac{1}{2}}
\end{equation*}
with respect to the natural embedding $\Blobb\hookrightarrow\Blob$. Reciprocally, any irreducible $\Blob$-module $\WW$ such that 
\begin{equation}
\Ress_{\Blobb} \WW=\WW^{N-1}_{j-\frac{1}{2}}\oplus\WW^{N-1}_{j+\frac{1}{2}}\,,\qquad -N/2\leq j\leq N/2\,,
\end{equation}
is isomorphic to $\WW_j^N$ (by convention $\WW^{N-1}_{\pm (N+1)/2}={0}$).

Let us now use this characterisation to show \eqref{bidecompb}. From the fusion rule $\VV_{\alpha}\otimes \CC^2\cong\VV_{\alpha+1}\oplus\VV_{\alpha-1}$ we already know that the $\Uq$-decomposition of $\Hb$ is given by
\begin{equation*}
\Hb=\bigoplus_{j=-N/2}^{N/2} {N\choose N/2+j} \VV_{\alpha+2j} \,,
\end{equation*}
where $j$ is integer if $N$ is even and half-integer if $N$ is odd. We will now show that, as a $(\Uq,\Blob)$-bimodule,
\begin{equation}
\label{blobdecomp}
{N\choose N/2+j} \VV_{\alpha+2j}=\VV_{\alpha+2j}\otimes\WW_j^N\,.
\end{equation}

Denote by $\tilde{\VV}^N_j\subset {N\choose N/2+j} \VV_{\alpha+2j}$ the ${N\choose N/2+j}$-dimensional space of $\Uq$ highest-weight vectors of weight $\alpha+2j-1$. Since $\Blob$ commutes with $\Uq$, $\tilde{\VV}^N_j$ is stable by the action of $\Blob$ and is thus a $\Blob$-module. Let us show by induction on $N$ that $\tilde{\VV}^N_j\cong\WW_j^N$.

For $N=1$ this is obvious, as $b_\pm$ are the projectors on $\VV_{\alpha\pm 1}$. Now assume that the proposition holds for $N-1$. For all $-N/2+1\leq j\leq N/2$, if $\ket{v}\in\tilde{\VV}^{N-1}_{j-\frac{1}{2}}$ then 
\begin{equation}
\label{phip}
\rho_j^+(v):=\ket{v}\otimes\ket{\uparrow}
\end{equation}
belongs to $\tilde{\VV}^{N}_{j}$ and for all $-N/2\leq j\leq N/2-1$ if $\ket{v}\in\tilde{\VV}^{N-1}_{j+\frac{1}{2}}$ then
\begin{equation}
\label{phim}
\rho_j^-(v):=[\alpha+2j]_\qa\ket{v}\otimes\ket{\downarrow}-\qam\left(\FF_{\VV_\alpha\otimes(\CC^2)^{\otimes(N-1)}}\ket{v}\right)\otimes\ket{\uparrow}
\end{equation}
is also an element of $\tilde{\VV}^{N}_{j}$. Indeed, using the coproduct \eqref{coproduct}, we have
\begin{equation*}
\EE_{\Hb}=\Id_{\VV_\alpha\otimes(\CC^2)^{\otimes(N-1)}}\otimes \sigma^++\EE_{\VV_\alpha\otimes(\CC^2)^{\otimes(N-1)}}\otimes\qa^{\sigma^z}
\end{equation*}
and so $\ket{v}\in\tilde{\VV}^{N-1}_{j-\frac{1}{2}}$ implies
\begin{equation*}
\EE_{\Hb}\rho_j^+(v)=\ket{v}\otimes\sigma^+\ket{\uparrow}+\EE_{\VV_\alpha\otimes(\CC^2)^{\otimes(N-1)}}\ket{v}\otimes\qa^{\sigma^z}\ket{\uparrow}=0
\end{equation*}
and $\ket{v}\in\tilde{\VV}^{N-1}_{j+\frac{1}{2}}$ implies
\begin{equation*}
\begin{aligned}
\EE_{\Hb}\rho_j^-(v) & =[\alpha+2j]_\qa\ket{v}\otimes\ket{\uparrow}-(\EE\FF)_{\VV_\alpha\otimes(\CC^2)^{\otimes(N-1)}}\ket{v}\otimes\ket{\uparrow}\\
& =[\alpha+2j]_\qa\ket{v}\otimes\ket{\uparrow}-\left(\FF\EE+\frac{\KK-\KKm}{\qa-\qam}\right)_{\VV_\alpha\otimes(\CC^2)^{\otimes(N-1)}}\ket{v}\otimes\ket{\uparrow}\\
& = [\alpha+2j]_\qa\ket{v}\otimes\ket{\uparrow}-[\alpha+2j]_\qa\ket{v}\otimes\ket{\uparrow}=0\,.
\end{aligned}
\end{equation*}

Thus, we have defined two injections $\rho_j^\pm :\tilde{\VV}^{N-1}_{j\mp\frac{1}{2}}\hookrightarrow\tilde{\VV}^N_j$ for $-N/2+1\leq j\leq N/2-1$, as well as two isomorphisms $\rho^\pm_{\pm N/2} : \tilde{\VV}^{N-1}_{\pm(N-1)/2}\xrightarrow{\sim}\tilde{\VV}^N_{\pm N/2}$. Moreover, by construction, these morphisms all commute with the action of the subalgebra $\Blobb\subset \Blob$. Indeed, $\Blobb$ commutes with the action of $\Uq$ on the first $N-1$ sites, so one can permute the elements of $\Blobb$ acting on vectors $v$ with $\FF$ in formula~\eqref{phim} (for morphisms $\rho_j^+$ the commutation property is obvious). We also notice that the two set of vectors in~\eqref{phip} and~\eqref{phim} are all linearly independent so by comparing dimensions their linear span is $\tilde{\VV}_j^N$. We thus have a decomposition of $\Blobb$-modules:
\begin{equation}
\label{resdecompv}
\tilde{\VV}^{N}_{j}=\Im(\rho_j^+) \oplus \Im(\rho_j^-)
\end{equation}
(by convention $\rho^\pm_{\mp N/2}=0$). By the induction hypothesis this means that 
\begin{equation}
\label{resdecomp}
\Ress_{\Blobb} \tilde{\VV}^{N}_{j}=\WW^{N-1}_{j-\frac{1}{2}}\oplus\WW^{N-1}_{j+\frac{1}{2}} \,,
\end{equation}
so to show that $\tilde{\VV}^{N}_{j}\cong \WW_j^N$ it only remains to prove that $\tilde{\VV}^{N}_{j}$ is irreducible.

For $j=\pm N/2$ this is obvious as $\tilde{\VV}^{N}_{\pm N/2}$ is one-dimensional. For $-N/2+1\leq j\leq N/2-1$ we have to show that $\tilde{\VV}^{N}_{j}$ has no non-trivial stable subspace. Because of \eqref{resdecompv} and \eqref{resdecomp}, if such a subspace exists, it can only be $\Im(\rho_j^+)$ or $\Im(\rho_j^-)$. Moreover, since $\Blob$ is semi-simple for $\qa^\alpha\notin\pm\qa^\ZZ$ it suffices to prove that the generator $e_{N-1}$ does not stabilise $\Im(\rho_j^+)$ (as it will automatically imply that $\Im(\rho_j^-)$ is not stable either).

To see this, note that all the vectors of $\Im(\rho_j^+)$ are of the form $\ket{v}\otimes\ket{\uparrow}$ and that by \eqref{phim} $\tilde{\VV}^{N-1}_{j-\frac{1}{2}}$ contains a vector $\ket{v}$ of the form
\begin{equation*}
\ket{v}=\ket{w}\otimes\ket{\downarrow}+\ket{w'}\otimes\ket{\uparrow}
\end{equation*}
as long as $-N/2+1\leq j\leq N/2-1$ and $\qa^\alpha\notin\pm\qa^\ZZ$. If $\ket{v}\otimes\ket{\uparrow}\in\Im(\rho_j^+)$ then
\begin{equation*}
e_{N-1}\ket{v}\otimes\ket{\uparrow}=\ket{w}\otimes e_{N-1}\ket{\downarrow\uparrow}+\ket{w'}\otimes e_{N-1}\ket{\uparrow\uparrow}=\qa\ket{w}\otimes\ket{\uparrow\downarrow}\notin\Im(\rho_j^+)\,.
\end{equation*}
Therefore, $\Im(\rho_j^+)$ is not stable by the action of $e_{N-1}$ and so $\tilde{\VV}^{N}_{j}$ is indeed irreducible for all $-N/2\leq j\leq N/2$. We have thus proven that $\tilde{\VV}^{N}_{j}\cong \WW_j^N$.

Finally, since all the vectors of $\tilde{\VV}^{N}_{j}$ are highest weight, $\FF^k\tilde{\VV}^{N}_{j}\cong\WW_j$ for all $k\geq 0$, so we have \eqref{blobdecomp}.
This enables us to write the $(\Blob,\Uq)$-bimodule decomposition of $\Hb$ \eqref{bidecompb}, which is equivalent to the statement that $\Blob$ and $\Uq$ are mutual maximal centralisers on $\Hb$, as it is multiplicity-free. Since all the irreducible $\Blob$-modules appear in this decomposition, this also implies that the representation of $\Blob$ on $\Hb$ is faithful.
\end{proof}

Note that the statement and the proof above remain valid also if $\qa^\alpha=\pm\qa^{n}$ for some $n\in\ZZ^*$ as long as $N\leq |n|$ because the algebra stays semisimple in this case \cite{Martin:1993jka, MARTIN2000957}.

It is worth mentioning that an analogue of Proposition \ref{blobrepqgen} for $\alpha\in\NN$ was recently proven in  \cite{iohara2019schurweyl} using abstract categorical and diagrammatic methods, while the case of $\alpha$ being a formal parameter was further treated in \cite{blobgen}. To our knowledge, Proposition \ref{blobrepqgen} is the first result of this kind that treats $\qa$ and $\alpha$ as actual complex numbers in some explicit allowed domains.

For $\qa=e^{\frac{i\pi}{p}}$ being a $2p$-th root of unity, a similar statement is true.

\begin{prop}
\label{blobreproot}
Let $\qa=e^{\frac{i\pi}{p}}$ with $p\in\NN\backslash\{0,1\}$, $\alpha\in\CC\backslash p\ZZ$ and $N\in\mathbb{N}^*$. Then $\Hb=\VV_{\alpha}\otimes (\CC^2)^{\otimes N}$, where $\VV_\alpha$ is the $p$-dimensional module defined in~\eqref{alpharep}, carries a representation of $\Blob$ with
\begin{equation}
\delta=[2]_\qa\,,\qquad y_\pm=\frac{[\alpha+1]_\qa}{[\alpha]_\qa} \,,
\end{equation}
and
\begin{equation}
\begin{gathered}
b=-\frac{1}{\{\alpha\}}
\begin{pmatrix}
-\qa^{-1}\KK^{-1}-\qa^{\alpha} & \{1\}\FF\\
\qa\{1\}\KK^{-1}\EE & \qa\KK^{-1}+\qa^{-\alpha}
\end{pmatrix}\,, \\
e_i =-\frac{1}{2}\left(\sigma^x_{i}\sigma^x_{i+1}+\sigma^y_{i}\sigma^y_{i+1}+\frac{\qa+\qa^{-1}}{2}(\sigma^z_{i}\sigma^z_{i+1}-1)\right)-\frac{\qa-\qa^{-1}}{4}(\sigma_{i+1}^z-\sigma_{i}^z)\,.
\end{gathered}
\end{equation}
Moreover, if $\alpha\in\CC\backslash\ZZ$, this representation is faithful, $\Uqu$ and $\Blob$ are mutual maximal centralisers and we have the decomposition of $(\Blob,\Uqu)$-bimodules
\begin{equation}
\label{bimodqroot}
\Hb=\bigoplus_{j=-N/2}^{N/2}\WW_j\otimes\VV_{\alpha+2j}\,.
\end{equation}
\end{prop}

\begin{proof}
The proof for $\qa$ generic will also work for the root-of-unity case with almost no modifications. We already know from direct computation that the operators $b$ and $e_i$ indeed define a representation of $\Blob$, and that by construction they commute with the action of $\Uqu$. As for the bimodule decomposition \eqref{bidecompb}, we only used the fact that the spaces of highest-weight vectors $\tilde{\VV}^{N}_{j}$ are of dimension ${N\choose N/2+j}$ coming from the fusion rule $\VV_{\alpha}\otimes \CC^2\cong\VV_{\alpha+1}\oplus\VV_{\alpha-1}$ which also holds for $\Uqu$; the injections $\rho_j^\pm$, whose construction relies on the coproduct formula \eqref{coproduct} and the commutation relation between $\EE$ and $\FF$, which is the same for $\Uq$ and $\Uqu$; and finally the semi-simplicity of $\Blob$ for $\qa^\alpha\notin\pm\qa^\ZZ$, which is preserved at roots of unity (the semi-simplicity condition then simply becomes $\alpha\notin\ZZ$). The only change is that we will now have only $p$ copies of $\WW_j$ in each sector given by $\FF^k\tilde{\VV}^{N}_{j}$ for $0\leq k\leq p-1$. Also, one should note that the weight conventions of the $\VV_\alpha$ representations in \eqref{alpharepgen} and \eqref{alpharep} are slightly different, so one has to apply the shift $\alpha\to\alpha+p$ to all expressions.
\end{proof}

\paragraph{Remark.}

At roots of unity, it is also possible to use yet another version of $\Uq$, where $\EE^p$ and $\FF^p$ are central elements, instead of just equal to $0$ as in $\Uqu$. This choice gives rise to $p$-dimensional so-called cyclic representations $\BBB(x_\pm,z,\ell)$, first considered in \cite{Skl83}, now depending on three continuous parameters -- the eigenvalues of the three central elements $\EE^p$, $\FF^p$ and $\KK^{p}$ -- and a discrete parameter $\ell\in\ZZ/p\ZZ$ (for more details see~\cite{Arnaudon:1992ig} and~\cite[Ch.\,9.2, 11.1]{qgroups}). All these representations can also be used to construct a representation of the blob algebra. However, it turns out that the corresponding blob weight $y$ will only depend on the value of the Casimir of $\BBB(x_\pm,z,\ell)$ which is still given by \eqref{casimir}. More precisely, if we define $\check{\alpha}\in\CC/2\ZZ$ as
\begin{equation*}
\qa^{p(\check{\alpha}+p)}+\qa^{p(\check{\alpha}+p)}:=\{1\}^{2p}x_+x_--(z+z^{-1}) \,,
\end{equation*}
then
\begin{equation*}
\CCC_{\BBB(x_\pm,z,\ell)}=-\qa^{\check{\alpha}+2\ell}-\qa^{-\check{\alpha}-2\ell} \,,
\end{equation*}
and the blob weight is
\begin{equation*}
y=\frac{[\check{\alpha}+2\ell+1]_\qa}{[\check{\alpha}+2\ell]_\qa}\,.
\end{equation*}
Specialising at $\VV_\alpha$, that is, taking $x_+=x_-=0$ and $\alpha=\check{\alpha}+2\ell$, we recover $y=\frac{[\alpha+1]_\qa}{[\alpha]_\qa}$, so without loss of generality we can work with the $\VV_\alpha$ representations of $\Uqu$.

\paragraph{Example.}
Recall the introduction of fermionic Fock spaces in~\eqref{eq:Fock-space-W}.
For $\qa=i$, $y_\pm=\pm\cot{\frac{\pi\alpha}{2}}$, we have isomorphisms of the blob algebra representations
\begin{equation*}
\overline{W}_{N/2+j}=\tilde{\VV}_j\cong \WW_j\,,\qquad W_{N/2+j}=\FF\tilde{\VV}_j\cong \WW_j
\end{equation*}
for all $-N/2\leq j\leq N/2$. Thus
\begin{equation}
\label{bidecompqi}
\WW_j\otimes\VV_{\alpha+2j}={N\choose N/2+j} \VV_{\alpha+2j}=2\WW_j=W_{N/2+j}\oplus\overline{W}_{N/2+j} \,,
\end{equation}
where we have written the same vector space as a $(\Blob,\Uqu)$-bimodule, a $\Uqu$-module, a $\Blob$-module, and finally as the $(N/2+j)$-fermion Fock subspace. Retrospectively, it is quite natural that the spectral equation \eqref{waveqU} only depends on $\mu$ and $y$ and has real coefficients even if $H_b$ is not self-adjoint, because, when restricted to a $\Blob$-module, its matrix entries are all real.

Let us show how this works for $N=2$ in terms of link states. The basis states are given by $\ket{\topinset{$\bullet$}{$|$}{3pt}{}~|}$, $\ket{\topinset{$\circ$}{$|$}{3pt}{}~|}$, $\ket{~\topinset{$\frown$}{$\bullet$}{2pt}{}~}$, $\ket{~\topinset{$\frown$}{$\circ$}{2pt}{}~}$ on which the one-boundary Hamiltonian $H_\ell=-\mu~\topinset{$\bullet$}{$|$}{3pt}{}~|-{\smile\atop\frown}$ acts by the graphical rules of the blob algebra. We have
\begin{equation*}
H_\ell=
\begin{pmatrix}
-\mu & 0 & 0 & 0\\
0 & 0 & 0 & 0\\
-1 & -1 & -\mu-y & y\\
-1 & -1 & -y & y
\end{pmatrix}
\,.
\end{equation*}
Therefore, denoting $\lambda_1$ and $\lambda_2$ the two roots of the polynomial
\begin{equation*}
P(\lambda)=\lambda^2+\mu \lambda-\mu y\,,
\end{equation*}
we have
\begin{equation*}
\mathrm{Spec}(H_\ell)=\{0,-\mu,\lambda_1, \lambda_2\}\,.
\end{equation*}
On the other hand, by \eqref{waveqU},
\begin{equation*}
U_2(\lambda/2)+\mu U_1(\lambda/2)+(1-\mu y)U_0(\lambda/2)=(\lambda^2-1)+\mu\lambda+(1-\mu y)=P(\lambda)\,.
\end{equation*}
Thus
\begin{equation*}
\mathrm{Spec}(H_b)=\{0,\lambda_1, \lambda_2, \lambda_1+\lambda_2=-\mu\}=\mathrm{Spec}(H_\ell)\,.
\end{equation*}
Moreover, introducing the blob modules $\WW_1=\CC\ket{\topinset{$\bullet$}{$|$}{3pt}{}~|}$, $\WW_0=\CC\ket{~\topinset{$\frown$}{$\bullet$}{2pt}{}~}\oplus\CC\ket{~\topinset{$\frown$}{$\circ$}{2pt}{}~}$, and $\WW_{-1}=\CC\ket{\topinset{$\circ$}{$|$}{3pt}{}~|}$ he have
\begin{equation*}
\begin{aligned}
& H_\ell|_{\WW_1} =H_b|_{W_2}=-\mu\,,\\
& H_\ell|_{\WW_{0}}=H_b|_{W_1}=
\begin{pmatrix}
-\mu-y & y\\
 -y & y
\end{pmatrix}\,,\\
& H_\ell|_{\WW_{-1}}=H_b|_{W_0}=0\,.
\end{aligned}
\end{equation*}

\section{Lattice algebras underlying the two-boundary spin chains}
\label{twoboundarysec}

We recall that the two-boundary Hamiltonian acts on the Hilbert space
\begin{equation*}
\Hbb:=\VV_{\alpha_l}\otimes\left(\CC^2\right)^{\otimes N}\otimes\VV_{\alpha_r}
\end{equation*}
as
\begin{equation}
H_{2b}=-\mu_l b_l-\mu_r b_r-\sum_{i=1}^{N-1}e_i\,.
\end{equation}
with $b_{l/r}$ introduced in Sections \ref{sec:Hb-gen-q} and \ref{sec:Hb-root-q}. From the above it is clear how to proceed to define the adapted lattice algebra: we have to extend $\Blob$, which already contains the left blob $b_l:=b$ of weight $y_l:=y$, by adding a right blob $b_r$ with diagram
\begin{equation*}
\begin{tikzpicture}[scale=0.8]
\draw (0,1) node {$b_r=$};
\draw (1,0) [dnup=0];
\draw (5,1) node {$\blacksquare$};
\draw (2,0) [dnup=0];
\draw (3,1) node {...};
\draw (4,0) [dnup=0];
\draw (5,0) [dnup=0];
\end{tikzpicture}
\end{equation*}
and satisfying
\begin{equation}
\label{blobruler}
b_r^2=b_r\,,\qquad e_{N-1} b_r e_{N-1} = y_r e_1\,,\qquad [b_l,b_r]=0\,, \qquad [b_r,e_i]=0 \quad \text{for} \quad 1\leq i\leq N-2\,,
\end{equation}
where $y_r$ is now the weight of the loop carrying the right blob $\blacksquare$. By direct computation we can check that the projector $b_r$, defined in \eqref{blobrqgen} and \eqref{blobrqroot}, satisfies relations \eqref{blobruler} with
\begin{equation*}
y_r=\frac{[\alpha_r+1]_\qa}{[\alpha_r]_\qa}\,.
\end{equation*}

For even $N$, this is however not enough from a diagrammatic point of view, because we also need to specify the weight of a loop carrying both the left and the right blob. Concretely, we need to introduce some $Y$ such that
\begin{equation}
\label{2bloop}
\left(\prod_{i=1}^{N/2}e_{2i-1}\right)b_l\left(\prod_{i=1}^{N/2-1}e_{2i}\right)b_r\left(\prod_{i=1}^{N/2}e_{2i-1}\right)=Y\prod_{i=1}^{N/2}e_{2i-1}\,.
\end{equation}
If $Y$ is a complex number, this relation, together with \eqref{blobruler} and the relations of the blob algebra, define the two-boundary Temperley-Lieb algebra \cite{de_Gier_2009}, denoted $\twoBlob$. However, in our spin chains, it will not be possible to fix $Y$ and instead we have sectors with different values of $Y$. In other words, it is convenient to think of $Y$ as an abstract central element or, said differently, to consider a central extension of the ``standard"  two-boundary Temperley-Lieb algebra. Then the same relations~\eqref{2bloop} define what we call the \textit{universal two-boundary Temperley-Lieb} algebra, denoted $\utwoBlob$.

Computing $Y$ is difficult because, contrary to the other loop weights $\delta$, $y_l$ and $y_r$ which can be found by performing some elementary computations on two or three sites, $Y$ is an intrinsically \emph{non-local} quantity as one has to go through the whole system to form a loop touching both boundaries. It is thus not surprising that $Y$ turns out to be some non-trivial non-local operator and not a fixed complex number. In what follows, we will express $Y$ in terms of the Casimir element of $\Uq$ (which commutes with all the $e_i$, $b_l$ and $b_r$ as well as $\Uq$) and then decompose $\Hbb$ into $Y$-eigenspaces so as to obtain representations of the standard two-boundary TL algebra with some well-defined scalar value of $Y$ in each eigenspace. Finally, we will formulate some conjectures on isomorphisms of these representations to standard modules or quotients thereof.

\subsection{$\Hbb$ as a representation of the universal two-boundary TL algebra}
\label{Ycomp}

Let us state the final result, first for generic $\qa$, and then prove it in detail.
\begin{prop}
\label{repthtwo}
Let $\qa\in\CC\backslash e^{i\pi\QQ}$, $\qa^{\alpha_l},\qa^{\alpha_r}\in\CC\backslash \{\pm 1\}$ and $N\in 2\mathbb{N}^*$. Then $\Hbb:=\VV_{\alpha_l}\otimes\left(\CC^2\right)^{\otimes N}\otimes\VV_{\alpha_r}$, where $\VV_{\alpha_l}$ and $\VV_{\alpha_r}$ are two Verma modules defined in~\eqref{alpharepgen}, carries a representation of  the universal two-boundary Temperley-Lieb algebra $\utwoBlob$ commuting with the $\Uq$ action, with parameters
\begin{equation}
\delta=[2]_\qa\,,\qquad y_{l}=\frac{[\alpha_{l}+1]_\qa}{[\alpha_l]_\qa}\,,\qquad y_{r}=\frac{[\alpha_{r}+1]_\qa}{[\alpha_r]_\qa}\,,
\end{equation}
and generators
\begin{equation}
\begin{gathered}
b_l=\frac{1}{\{\alpha_l\}}
\begin{pmatrix}
 -\qam \KK^{-1}+\qa^{\alpha_l} &\{1\}\FF\\
 \qa\{1\}\KK^{-1}\EE & \qa\KK^{-1}-\qa^{-\alpha_l}
\end{pmatrix}\,, \qquad
b_r=\frac{1}{\{\alpha_r\}}
\begin{pmatrix}
 \qa\KK - \qa^{-\alpha_r} & \qa\{1\}\KK\FF\\
 \{1\}\EE & -\qam\KK+\qa^{\alpha_r}
\end{pmatrix}\,,\\
e_i =-\frac{1}{2}\left(\sigma^x_{i}\sigma^x_{i+1}+\sigma^y_{i}\sigma^y_{i+1}+\frac{\qa+\qa^{-1}}{2}(\sigma^z_{i}\sigma^z_{i+1}-1)\right)-\frac{\qa-\qa^{-1}}{4}(\sigma_{i+1}^z-\sigma_{i}^z)\,,\\
Y=\frac{\qa^{\alpha_l+\alpha_r+1}+\qa^{-\alpha_l-\alpha_r-1}-\CCC_{\Hbb}}{\{\alpha_l\} \{\alpha_r\}}\,,
\end{gathered}
\end{equation}
where $\CCC_{\Hbb}$ is the Casimir of $\Uq$ on $\Hbb$ given by \eqref{casimir} or, explicitly,
\begin{equation}
\label{tbcasimir1}
\CCC_{\Hbb}=\{1\}^2\FF_{\Hbb}\EE_{\Hbb}+\qa\KK_{\Hbb}+\qam\KK^{-1}_{\Hbb}
\end{equation}
with
\begin{equation}
\label{tbcasimir2}
\begin{gathered}
\KK_{\Hbb}^{\pm 1}  =\KK_{\VV_{\alpha_l}}^{\pm 1}\otimes\qa^{\pm\sum_{i=1}^{N}\sigma_i^z}\otimes\KK_{\VV_{\alpha_r}}^{\pm 1}\,,\\
\EE_{\Hbb}  =\EE_{\VV_{\alpha_l}}\otimes\qa^{\sum_{i=1}^{N}\sigma_i^z}\otimes\KK_{\VV_{\alpha_r}}+\left(\sum_{k=1}^N\sigma_k^+\otimes\qa^{\sum_{i=k+1}^{N}\sigma_i^z}\otimes\KK_{\VV_{\alpha_r}}\right)+\EE_{\VV_{\alpha_r}}\,,\\
\FF_{\Hbb}  =\FF_{\VV_{\alpha_l}}+\left(\sum_{k=1}^N\KK_{\VV_{\alpha_l}}^{-1}\otimes\qa^{-\sum_{i=1}^{k-1}\sigma_i^z}\otimes\sigma_k^-\right)+\KK_{\VV_{\alpha_l}}^{-1}\otimes\qa^{-\sum_{i=1}^{N}\sigma_i^z}\otimes\FF_{\VV_{\alpha_r}}\,.
\end{gathered}
\end{equation}
\end{prop}

\begin{proof}
We already know that the $e_i$ satisfy the relations of the TL algebra \eqref{TLrel}, and that the two operators $b_r$ and $b_l$ commute with each other and satisfy the left/right blob relations
\begin{equation}
\label{blobrules}
\begin{aligned}
b_l^2 & =b_l \,,\qquad & e_1 b_l e_1 & =y_l e_1\,,\qquad & [b_l,e_i] & =0\quad & \forall\,2\leq i\leq N-1\,,\\
b_r^2 & =b_r\,,\qquad & e_{N-1} b_r e_{N-1} & = y_r e_{N-1}\,,\qquad & [b_r,e_i] & =0\quad & \forall\,1\leq i\leq N-2\,.
\end{aligned}
\end{equation}
By construction, all these generators do commute with $\Uq$. Therefore it only remains to prove \eqref{2bloop}. For this, let us first introduce a powerful diagrammatic formalism.

\paragraph{Braiding.}

It is known that $\Uq$ admits a universal $\RRR$-matrix given by \cite{Drinfeld:1986in} (see also \cite[Ch.\,6.4]{qgroups})
\begin{equation}
\label{Rmatrix}
\RRR=\qa^{\frac{\HH\otimes\HH}{2}}\sum_{k\geq 0}\frac{\{1\}^{2k}}{\{k\}!}\qa^{k(k-1)/2}\EE^k\otimes \FF^k \,,
\end{equation}
where
\begin{equation*}
\{n\}!:=\prod_{k=1}^n\{k\}\,.
\end{equation*}
Although strictly speaking $\RRR\notin\Uq\otimes\Uq$, it can be evaluated on tensor product of any pair  $(\XX, \YY)$ of representations of $\Uq$ as long as \emph{at least one} of them is finite-dimensional.\footnote{It is enough to consider the action of $\RRR$ on eigenvectors of $\HH\otimes \HH$, so the action of the first factor $\qa^{\frac{\HH\otimes\HH}{2}}$ is well defined. If, say $\YY$, is finite-dimensional, then for every $v\in \YY$ only finitely many $F^k v$ are non-zero, and therefore the sum in~\eqref{Rmatrix} is finite on every vector $w\otimes v\in \XX\otimes \YY$.} We denote this evaluation by $\RRR_{\XX,\YY}$. One of the essential properties of $\RRR$ is that for any two such representations $\XX$ and $\YY$, the two operators 
\begin{equation}
\label{Rprop}
P_{\XX,\YY}\circ\RRR_{\XX,\YY} \qquad \text{and} \qquad \RRR_{\YY,\XX}^{-1}\circ P_{\XX,\YY} \,,
\end{equation}
where
\begin{equation*}
\begin{aligned}
P_{\XX,\YY}\, : ~\XX\otimes \YY & \to \YY\otimes \XX\\
 x\otimes y & \mapsto y\otimes x
\end{aligned}
\end{equation*}
is the operator permuting the two tensor factors, commute with the action of $\Uq$, that is, they belong to $\Hom_{\Uq}(\XX\otimes\YY, \YY\otimes\XX)$. In other words, $\RRR$ generates two (a priori different) intertwiners between $\XX\otimes \YY$ and $\YY\otimes \XX$. Graphically, one often represents
\begin{equation*}
\begin{tikzpicture}
\pic at (1.5,0.8) {braid={s_1}};
\node at (0,0) {$P_{\XX,\YY}\circ\RRR_{\XX,\YY}=$};
\node at (1.5,-1) {$\XX$};
\node at (2.5,-1) {$\YY$};
\node at (1.5,1) {$\YY$};
\node at (2.5,1) {$\XX$};
\pic at (6.5,0.8) {braid={s_1^{-1}}};
\node at (5,0) {$\RRR_{\YY,\XX}^{-1}\circ P_{\XX,\YY}=$};
\node at (6.5,-1) {$\XX$};
\node at (7.5,-1) {$\YY$};
\node at (6.5,1) {$\YY$};
\node at (7.5,1) {$\XX$};
\end{tikzpicture}
\,.
\end{equation*}
Now define the double braidings (also known as monodromies)
\begin{equation*}
\begin{aligned}
\BB_{\XX,\YY} & :=P_{\YY,\XX}\circ\RRR_{\YY,\XX}\circ P_{\XX,\YY}\circ\RRR_{\XX,\YY}\,,\\
\overline{\BB}_{\XX,\YY} & :=\RRR_{\XX,\YY}^{-1} \circ P_{\YY,\XX}\circ\RRR_{\YY,\XX}^{-1}\circ P_{\XX,\YY}=\BB_{\XX,\YY}^{-1}\,.
\end{aligned}
\end{equation*}
From \eqref{Rprop} it is clear that $\BB_{\XX,\YY}, \overline{\BB}_{\XX,\YY}\in\End_{\Uq}(\XX\otimes\YY)$. Graphically
\begin{equation*}
\begin{tikzpicture}
\pic at (1.5,0.8) {braid={s_1 s_1}};
\node at (0.5,-0.5) {$\BB_{\XX,\YY}=$};
\node at (1.5,-2) {$\XX$};
\node at (2.5,-2) {$\YY$};
\node at (1.5,1) {$\XX$};
\node at (2.5,1) {$\YY$};
\pic at (6.5,0.8) {braid={s_1^{-1} s_1^{-1}}};
\node at (5.5,-0.5) {$\overline{\BB}_{\XX,\YY}=$};
\node at (6.5,-2) {$\XX$};
\node at (7.5,-2) {$\YY$};
\node at (6.5,1) {$\XX$};
\node at (7.5,1) {$\YY$};
\end{tikzpicture}
\,.
\end{equation*}

\paragraph{Quantum traces and the Casimir element.}

Additionally, for any finite-dimensional $\Uq$-module $\XX$, one has natural linear maps $\End_{\Uq}(\XX)\to\CC$ defined by
\begin{equation*}
\qtr_\XX(f)  :=\tr_\XX(\KK_\XX f)\,,\qquad
\overline{\qtr}_\XX(f)  :=\tr_\XX(f\KK_\XX^{-1})\ , \qquad f\in\End_{\Uq}(\XX)\ ,
\end{equation*}
called, respectively, the right and left quantum traces. These traces can be thought as a result of consecutive  application of the following three $\Uq$-intertwining operators: for the right  quantum trace, first, the standard coevaluation map 
\begin{equation*}
\mathrm{coev}_\XX : \CC\to\XX\otimes \XX^*, \qquad \mathsf{1}\mapsto \sum_{i} v_i\otimes v^i\,,
\end{equation*}
where $v_i$ is a basis in $\XX$ and $v^i\in\XX^*$ is the dual basis, i.e.\ $v^i(v_j)=\delta_{ij}$, then followed by $f\otimes \Id_{\XX^*}$, and finally by the evaluation map 
\begin{equation*}
\mathrm{ev}_\XX : \XX\otimes \XX^*\to \CC, \qquad v\otimes f\mapsto f(\KK_\XX v)
\end{equation*}
that uses the pivotal structure of $\Uq$ given by action of $\KK$\footnote{The term `pivotal' means here that $S^2(x) = \KK x \KK^{-1}$ for all $x\in\Uq$ where $S:\Uq\to\Uq$ is the antipode of $\Uq$ (see \cite[Ch.\,4.2]{qgroups} and \cite[Ch.\,VII.1]{kassel}). This property assures that $\mathrm{ev}_\XX$ commutes with $\Uq$ action. Indeed recall that, for all $f\in\XX^*$ the action on the dual space $\XX^*$ is given via the antipode $S$ by $(x\cdot f)(-) := f(S(x)-)$, the action on a tensor product via the coproduct $\Delta$ and on $\mathbb{C}$ via the counit $\epsilon:\Uq\to\CC$. Then, for all $x\in\Uq$, $v\in\XX$, $f\in\XX^*$ and denoting $\Delta(x)=\sum x_{(1)}\otimes x_{(2)}$, we have $\mathrm{ev}_\XX(x\cdot(v\otimes f))=f\left(\sum S(x_{(2)})\KK x_{(1)} v\right)=f\left(\sum S(x_{(2)})S^2(x_{(1)}) \KK v\right)=f\left(S\left(\sum S(x_{(1)})x_{(2)}\right)\KK v\right)=f(S(\iota\circ\epsilon(x))\KK v)=\epsilon(x)\mathrm{ev}_\XX(v\otimes f)$ where we use the anti-automorphism property of $S$ as well as the axiom $\sum S(x_{(1)})x_{(2)}=\iota\circ\epsilon(x)$ with $\iota :\CC\to\Uq$ the unit map.} ; whereas for the left quantum trace, first, the pivotal coevaluation map 
\begin{equation*}
\overline{\mathrm{coev}}_\XX : \CC\to\XX^*\otimes \XX, \qquad \mathsf{1}\mapsto \sum_{i} v^i\otimes \KK^{-1}_{\XX}v_i\,,
\end{equation*}
then $\Id_{\XX^*}\otimes f$, and lastly the standard evaluation map 
\begin{equation*}
\overline{\mathrm{ev}}_\XX : \XX^*\otimes \XX\to \CC, \qquad v\otimes f\mapsto f(v)\,.
\end{equation*}
Graphically, representing the coevaluation map by a cup and the evaluation map by a cap, the right and left quantum traces of a map $f\in\End_{\Uq}(\XX)$ are respectively drawn as
\begin{equation}
\begin{tikzpicture}
\draw(0.5,0) [dndn=1];
\draw (0.5,-1) [upup=1];
\draw (0,0) rectangle (1,-1);
\draw (0.5, -0.5) node{$f$};
\draw (1.5,0) -- (1.5,-1);
\draw (-1,-0.5) node{$\qtr_\XX(f)=$};
\draw (0.25,0.5) node{$\XX$};
\draw (1.75,0.5) node{$\XX^*$};

\draw(5.5,0) [dndn=1];
\draw (5.5,-1) [upup=1];
\draw (6,0) rectangle (7,-1);
\draw (6.5, -0.5) node{$f$};
\draw (5.5,0) -- (5.5,-1);
\draw (4.5,-0.5) node{$\overline{\qtr}_\XX(f)=$};
\draw (5.25,0.5) node{$\XX^*$};
\draw (6.75,0.5) node{$\XX$};
\end{tikzpicture}
\,.
\end{equation}
Note that the auxiliary dual space $\XX^*$ is on the left (resp.\ right) for the left (resp.\ right) quantum trace, which justifies their names.

We also note that the fundamental representation $\mathbb{C}^2$ is self-dual and the TL generators can be written in terms of the (co)evaluation maps\footnote{Strictly speaking, one needs to compute explicitly the isomorphism $\CC^2\cong(\CC^{2})^*$ and compose the (co)evaluation maps with them.}
 $$
 e_i = \mathrm{coev}_{\mathbb{C}^2} \circ \mathrm{ev}_{\mathbb{C}^2}  \colon \mathbb{C}^2\otimes \mathbb{C}^2 \to \mathbb{C}^2\otimes \mathbb{C}^2
 $$
  acting on the $i$-th and $i+1$-th $\CC^2$ sites which is consistent with their diagrammatic representation.

The quantum traces over finite-dimensional $\XX$ are important because they allow to construct new intertwining operators via taking partial quantum traces. Indeed, for any (not necessarily finite-dimensional) $\Uq$ module $\YY$ and any $\Uq$-intertwining operators $f\in\End_{\Uq}(\YY\otimes\XX)$ and $\overline{f}\in\End_{\Uq},(\XX\otimes\YY)$ we can define
\begin{equation*}
\begin{aligned}
\qtr_\XX(f) & :=\tr_\XX(\KK_\XX f) \in\End_{\Uq}(\YY)\,,\\
\overline{\qtr}_\XX(\overline{f}) & :=\tr_\XX(\overline{f}\KK_\XX^{-1}) \in\End_{\Uq}(\YY)\,,
\end{aligned}
\end{equation*}
where we abuse our notation and write $\KK_\XX$ instead of more lengthy $\Id_{\YY}\otimes \KK_\XX$, etc., and $\tr_\XX$ stays for the (usual) partial trace over the $\XX$ component of $\XX\otimes\YY$.
Graphically, one represents these partial traces as
\begin{equation}
\begin{tikzpicture}
\draw(0.5,0) [dndn=1];
\draw (0.5,-1) [upup=1];
\draw (-0.25,0) rectangle (0.75,-1);
\draw (0.25, -0.5) node{$f$};
\draw (1.5,0) -- (1.5,-1);
\draw (0,0) -- (0, 0.5);
\draw (0,-1) -- (0,-1.5);
\draw (0, 0.5) node[scale=2]{.};
\draw (0,-1.5) node[scale=2]{.};
\draw (0, 0.75) node{$\YY$};
\draw (0,-1.75) node{$\YY$};
\draw (-1.25,-0.5) node{$\qtr_\XX(f)=$};

\draw(5.5,0) [dndn=1];
\draw (5.5,-1) [upup=1];
\draw (6.25,0) rectangle (7.25,-1);
\draw (6.75, -0.5) node{$\overline{f}$};
\draw (5.5,0) -- (5.5,-1);
\draw (7,0) -- (7, 0.5);
\draw (7,-1) -- (7,-1.5);
\draw (7, 0.5) node[scale=2]{.};
\draw (7,-1.5) node[scale=2]{.};
\draw (7, 0.75) node{$\YY$};
\draw (7,-1.75) node{$\YY$};
\draw (4.5,-0.5) node{$\overline{\qtr}_\XX(\overline{f})=$};
\end{tikzpicture}
\end{equation}
From this diagrammatic presentation it is clear why the partial traces are $\Uq$-intertwining operators on~$\YY$, and not just linear endomorphisms of $\YY$: they are again compositions of three intertwining operators, $\left(\Id_\YY\otimes\mathrm{ev}_\XX\right)\circ\left(f\otimes \Id_{\XX^*}\right)\circ\left(\Id_\YY\otimes\mathrm{coev}_\XX\right)$ for the right partial quantum trace and $\left(\overline{\mathrm{ev}}_\XX\otimes\Id_\YY\right)\circ\left(\Id_{\XX^*}\otimes f\right)\circ\left(\overline{\mathrm{coev}}_\XX\otimes\Id_\YY\right)$ for the left one.

An essential property of this pictorial formalism is that isotopic deformation of strings in a diagram does not affect the $\Uq$-intertwiner it represents \cite[Ch.\,5.3]{qgroups} (see \cite[Part III]{kassel} for more details). This will be particularly important for us later on.

We are now ready to state the main lemma we will use for the proof of Proposition \ref{repthtwo}.

\begin{lemma}
For any $\Uq$-module $\VV$,
\begin{equation}
\begin{aligned}
\qtr_{\CC^2}\left(\BB_{\VV,\CC^2}\right) & =\tr_{\CC^2}\left(\qa^{\sigma^z}\BB_{\VV,\CC^2}\right)=\CCC_\VV\,,\\
\overline{\qtr}_{\CC^2}\left(\overline{\BB}_{\CC^2,\VV}\right) & =\tr_{\CC^2}\left(\overline{\BB}_{\CC^2,\VV}\qa^{-\sigma^z}\right)=\CCC_\VV\,,
\end{aligned}
\end{equation}
where $\CCC_\VV$ is the Casimir element action on $\VV$. Graphically,
\begin{equation}
\label{graphrule}
\begin{tikzpicture}[scale=0.5]
\pic[scale=0.8] at (-12.1,2) {braid={s_1 s_1}};
\draw(-10.5,2) [dndn=1];
\draw(-10.5,-2) [upup=1];
\draw (-9.5,-2) -- (-9.5,2);
\draw (-12.1,-2.5) node{$\VV$};
\draw (-12.1, 2.5) node{$\VV$};
\draw (-12.5, 0.8) node{$\CC^2$};

\draw (-7.75, 0) node{$=$};

\pic[scale=0.8] at (-5,2) {braid={s_1^{-1} s_1^{-1}}};
\draw(-6,2) [dndn=1];
\draw(-6,-2) [upup=1];
\draw (-6,-2) -- (-6,2);
\draw (-3.4,-2.5) node{$\VV$};
\draw (-3.4, 2.5) node{$\VV$};
\draw (-3, 0.8) node{$\CC^2$};

\draw (-1.5, 0) node{$=$};

\draw (2,0) ellipse (2 and 1);
\fill[white] (2,1) circle (0.1);
\draw (2,-0.9) -- (2,2);
\draw (2,-2) -- (2,-1.1);
\draw (2,-2.5) node{$\VV$};
\draw (2, 2.5) node{$\VV$};
\draw (0, 0.8) node{$\CC^2$};
\draw (6, 0) node{$=~\CCC_\VV$};
\end{tikzpicture}
\,.
\end{equation}
\end{lemma}
\begin{proof}
By direct computation,
\begin{equation*}
\begin{aligned}
\BB_{\VV,\CC^2} & =\qa^{\frac{\HH\otimes\sigma^z}{2}}(1+\{1\}\FF\otimes\sigma^+)\qa^{\frac{\HH\otimes\sigma^z}{2}}(1+\{1\}\EE\otimes\sigma^-)\\
& =
\begin{pmatrix}
\KK+\qam\{1\}^2\FF\EE & \qam\{1\}\FF\\
\{1\}\KK^{-1}\EE & \KK^{-1}
\end{pmatrix}
\end{aligned}
\end{equation*}
and
\begin{equation*}
\begin{aligned}
\overline{\BB}_{\CC^2,\VV} & =(1-\{1\}\sigma^+\otimes\FF)\qa^{-\frac{\sigma^z\otimes\HH}{2}}(1-\{1\}\sigma^-\otimes\EE)\qa^{-\frac{\sigma^z\otimes\HH}{2}}\\
& =
\begin{pmatrix}
\KKm +\qa\{1\}^2\FF\EE & -\qa^2\{1\}\KK\FF\\
-\qa\{1\}\EE & \KK
\end{pmatrix} \,,
\end{aligned}
\end{equation*}
so
\begin{equation*}
\qtr_{\CC^2}\left(\BB_{\VV,\CC^2}\right):=\tr_{\CC^2}\left(\qa^{\sigma^z}\BB_{\VV,\CC^2}\right)=\{1\}^2\FF\EE+\qa\KK+\qam\KK^{-1} :=\CCC_\VV
\end{equation*}
and
\begin{equation*}
\overline{\qtr}_{\CC^2}\left(\overline{\BB}_{\CC^2,\VV}\right) :=\tr_{\CC^2}\left(\overline{\BB}_{\CC^2,\VV}\qa^{-\sigma^z}\right)=\{1\}^2\FF\EE+\qa\KK+\qam\KK^{-1} :=\CCC_\VV \,.
\end{equation*}
\end{proof}

\paragraph{End of the proof.}

We now have all the ingredients to prove \eqref{2bloop} with the expression for $Y$ in Proposition~\ref{repthtwo}. First, observe that
\begin{equation*}
\begin{aligned}
b_l & =\frac{1}{\{\alpha_l\}}\left(\qa\BB_{\VV_{\alpha_l},\CC^2_{(1)}}-\qa^{-\alpha_l}\right)\,,\\
b_r & =\frac{1}{\{\alpha_r\}}\left(\qa^{\alpha_r}-\qam\overline{\BB}_{\CC^2_{(N)},\VV_{\alpha_r}}\right)\,,
\end{aligned}
\end{equation*}
where $\CC^2_{(k)}$ denotes the $k$-th $\CC^2$ site of the spin chain $\Hbb$. Then write
\begin{equation}
\label{aux1}
b_lb_r=-\frac{\BB_{\VV_{\alpha_l},\CC^2_{(1)}}\overline{\BB}_{\CC^2_{(N)},\VV_{\alpha_r}}}{\{\alpha_l\}\{\alpha_r\}}+\frac{\qa^{\alpha_r}}{\{\alpha_r\}}b_l-\frac{\qa^{-\alpha_l}}{\{\alpha_l\}}b_r+\frac{\qa^{\alpha_r-\alpha_l}}{\{\alpha_l\}\{\alpha_r\}}\,.
\end{equation}
From the usual Temperley-Lieb \eqref{TLrel} and blob \eqref{blobrules} relations we have
\begin{equation}
\label{aux2}
\begin{split}
& \left(\prod_{i=1}^{N/2}e_{2i-1}\right)\left(\frac{\qa^{\alpha_r}}{\{\alpha_r\}}b_l-\frac{\qa^{-\alpha_l}}{\{\alpha_l\}}b_r+\frac{\qa^{\alpha_r-\alpha_l}}{\{\alpha_l\}\{\alpha_r\}}\right)\left(\prod_{i=1}^{N/2-1}e_{2i}\right)\left(\prod_{i=1}^{N/2}e_{2i-1}\right)=\\
& \qquad =\dfrac{\qa^{\alpha_r}\{\alpha_l+1\}-\qa^{-\alpha_l}\{\alpha_r+1\}+\delta\qa^{\alpha_r-\alpha_l}}{\{\alpha_l\}\{\alpha_r\}} \left(\prod_{i=1}^{N/2}e_{2i-1}\right)\\
&\qquad=\dfrac{\qa^{\alpha_l+\alpha_r+1}+\qa^{-\alpha_l-\alpha_r-1}}{\{\alpha_l\}\{\alpha_r\}} \left(\prod_{i=1}^{N/2}e_{2i-1}\right)\,.
\end{split}
\end{equation}
On the other hand, the graphical expression for 
$$
\left(\prod_{i=1}^{N/2}e_{2i-1}\right)\left(\BB_{\VV_{\alpha_l},\CC^2_{(1)}}\overline{\BB}_{\CC^2_{(N)},\VV_{\alpha_r}}\right)\left(\prod_{i=1}^{N/2-1}e_{2i}\right)\left(\prod_{i=1}^{N/2}e_{2i-1}\right)
$$
is given by the diagram
\begin{equation}
\label{2bdiag}
\begin{tikzpicture}
\pic[braid/.cd, strand 1/.style={red, thick}] at (0,0) {braid={s_1 s_1}};
\pic[braid/.cd, strand 2/.style={blue, thick}] at (5,0) {braid={s_1^{-1} s_1^{-1}}};
\draw (1,0) [dndn=4];
\draw (1,-2.5) [upup=4];
\draw (1,-4) [dndn=1];
\draw (4,-4) [dndn=1];
\node at (3,-3.75) {$\ldots$};
\draw (1,1.5) [upup=1];
\draw (4,1.5) [upup=1];
\node at (3,1.25) {$\ldots$};
\draw[red, thick] (0,0) -- (0,1.5);
\draw[blue, thick] (6,0) -- (6,1.5);
\draw[red, thick] (0,-4) -- (0,-2.5);
\draw[blue, thick] (6,-4) -- (6,-2.5);
\node at (0,1.75) {$\VV_{\alpha_l}$};
\node at (6,1.75) {$\VV_{\alpha_r}$};
\node at (0,-4.25) {$\VV_{\alpha_l}$};
\node at (6,-4.25) {$\VV_{\alpha_r}$};
\node at (5,-4.25) {$\CC^2$};
\node at (4,-4.25) {$\CC^2$};
\node at (1,-4.25) {$\CC^2$};
\node at (2,-4.25) {$\CC^2$};
\node at (5,1.75) {$\CC^2$};
\node at (4,1.75) {$\CC^2$};
\node at (1,1.75) {$\CC^2$};
\node at (2,1.75) {$\CC^2$};
\end{tikzpicture}
\,.
\end{equation}
where we used TL relations and isotopy to straighten the strings. Now we use an additional property of the diagrammatic representation: we can locally pass a cap or cup through a string without changing the overall intertwiner. This is explained by the general property of naturality of the braiding (see more details in~\cite[Ch.\,5.3]{qgroups}). In other words, we have, graphically,
\begin{equation*}
\begin{tikzpicture}
\draw (0,0) [dndn=1];
\draw (-0.5,1) -- (1.5,1);

\node at (2.5,0.5) {$=$};

\draw (4,0.25) [dndn=1];
\fill[white] (4.1,0.75) circle (0.1);
\fill[white] (4.9,0.75) circle (0.1);
\draw (3.5,0.75) -- (5.5,0.75);

\node at (6.5,0.5) {$=$};

\draw (7.5,0.75) -- (9.5,0.75);
\fill[white] (8.1,0.75) circle (0.1);
\fill[white] (8.9,0.75) circle (0.1);
\draw (8,0.25) [dndn=1];
\end{tikzpicture}
\end{equation*}
and similarly for the cups. Passing all the caps in \eqref{2bdiag} through the middle loop and bringing them together with the cups we see that this diagram can be expressed as the product of a loop going around the whole system $\Hbb$ and of the intertwiner $\prod_{i=1}^{N/2}e_{2i-1}$ (recall that $e_i = \mathrm{coev}_{\mathbb{C}^2_{(i)}} \circ \mathrm{ev}_{\mathbb{C}^2_{(i)}}$). But by \eqref{graphrule}, a diagram isotopic to a loop going around a $\Uq$ module $\VV$ represents $\CCC_\VV$. Applying this result to the whole spin-chain representation $\Hbb$, we thus obtain
\begin{equation}
\label{aux3}
\left(\prod_{i=1}^{N/2}e_{2i-1}\right)\left(\BB_{\VV_{\alpha_l},\CC^2_{(1)}}\overline{\BB}_{\CC^2_{(N)},\VV_{\alpha_r}}\right)\left(\prod_{i=1}^{N/2-1}e_{2i}\right)\left(\prod_{i=1}^{N/2}e_{2i-1}\right)=\CCC_{\Hbb}\left(\prod_{i=1}^{N/2}e_{2i-1}\right)\,.
\end{equation}
Combining \eqref{aux1}, \eqref{aux2} and \eqref{aux3} we finally have
\begin{equation}
\left(\prod_{i=1}^{N/2}e_{2i-1}\right)b_l\left(\prod_{i=1}^{N/2-1}e_{2i}\right)b_r\left(\prod_{i=1}^{N/2}e_{2i-1}\right)=Y\prod_{i=1}^{N/2}e_{2i-1}
\end{equation}
with
\begin{equation}
Y=\frac{\qa^{\alpha_l+\alpha_r+1}+\qa^{-\alpha_l-\alpha_r-1}-\CCC_{\Hbb}}{\{\alpha_l\} \{\alpha_r\}}\,.
\end{equation}
\end{proof}

For $\qa$ a $2p$-th root of unity we have an analogous result.

\begin{prop}
\label{repthtworoot}
Let $\qa=e^{\frac{i\pi}{p}}$ with $p\in\NN\backslash\{0,1\}$, $\alpha_l,\alpha_r\in\CC\backslash p\ZZ$ and $N\in 2\mathbb{N}^*$. Then $\Hbb=\VV_{\alpha_l}\otimes\left(\CC^2\right)^{\otimes N}\otimes\VV_{\alpha_r}$, where $\VV_{\alpha_l}$ and $\VV_{\alpha_r}$ are two $p$-dimensional modules defined in~\eqref{alpharep}, carries a representation of  the universal two-boundary Temperley-Lieb algebra $\utwoBlob$ commuting with the $\Uqu$ action, with parameters
\begin{equation}
\label{weights2broot}
\delta=[2]_\qa\,,\qquad y_{l}=\frac{[\alpha_{l}+1]_\qa}{[\alpha_l]_\qa}\,,\qquad y_{r}=\frac{[\alpha_{r}+1]_\qa}{[\alpha_r]_\qa}\,,
\end{equation}
and generators
\begin{equation}
\begin{gathered}
b_l=\frac{1}{\{\alpha_l\}}
\begin{pmatrix}
 \qam \KK^{-1}+\qa^{\alpha_l} & -\{1\}\FF\\
 -\qa\{1\}\KK^{-1}\EE & -\qa\KK^{-1}-\qa^{-\alpha_l}
\end{pmatrix}\,, \qquad
b_r=\frac{1}{\{\alpha_r\}}
\begin{pmatrix}
 -\qa\KK - \qa^{-\alpha_r} & -\qa\{1\}\KK\FF\\
 -\{1\}\EE & \qam\KK+\qa^{\alpha_r}
\end{pmatrix}\,,\\
e_i =-\frac{1}{2}\left(\sigma^x_{i}\sigma^x_{i+1}+\sigma^y_{i}\sigma^y_{i+1}+\frac{\qa+\qa^{-1}}{2}(\sigma^z_{i}\sigma^z_{i+1}-1)\right)-\frac{\qa-\qa^{-1}}{4}(\sigma_{i+1}^z-\sigma_{i}^z)\,,\\
Y=\frac{\qa^{\alpha_l+\alpha_r+1}+\qa^{-\alpha_l-\alpha_r-1}-\CCC_{\Hbb}}{\{\alpha_l\} \{\alpha_r\}}\ ,
\end{gathered}
\end{equation}
where $\CCC_{\Hbb}$ is the Casimir element of $\Uqu$ acting on $\Hbb$, given by \eqref{tbcasimir1}-\eqref{tbcasimir2}.
\end{prop}

\begin{proof}
All the formalism and results that we have introduced for $\Uq$ also apply to $\Uqu$, the only modification being that the $\RRR$-matrix of $\Uqu$ \cite{geermodtr0}
\begin{equation}
\label{Rmatroot}
\RRR=\qa^{\frac{\HH\otimes\HH}{2}}\sum_{k=0}^{p-1}\frac{\{1\}^{2k}}{\{k\}!}\qa^{k(k-1)/2}\EE^k\otimes \FF^k \,,
\end{equation}
is truncated at order $p$ because of the relations $\EE^p=\FF^p=0$. Apart from this, the proof is exactly the same, up to the usual shift $\alpha_{l/r}\to\alpha_{l/r}+p$ in the definition of $\VV_{\alpha_{l/r}}$ at roots of unity.
\end{proof}

Propositions \ref{repthtwo} and \ref{repthtworoot} are weaker than their one-boundary counterparts, Propositions \ref{blobrepqgen} and \ref{blobreproot}, because we were unable to prove that $\Uq$ (or $\Uqu$) and $\utwoBlob$ are mutual maximal centralisers, nor to compute the decomposition of $\Hbb$ into irreducible $\utwoBlob$-modules. 

\subsection{Decomposition of $\Hbb$ into $\twoBlob$-modules with different values of $Y$}
\label{twoblobdec}

Let us now relate the representation of the universal two-boundary TL algebra $\utwoBlob$ that we have constructed to some representations of $\twoBlob$ for different numerical values of $Y$. 

To do so, we simply have to decompose $\Hbb$ into $\CCC$-eigenspaces. In each such eigenspace $Y$ will act as a complex number, and we will thus obtain a representation of the ordinary two-boundary TL algebra $\twoBlob$.

Let us start with generic $\qa$. We first compute the $\Uq$-decomposition of $\Hbb$. For this, we bring the two Verma modules $\VV_{\alpha_l}$ and $\VV_{\alpha_r}$ together to the left using the braiding \eqref{Rprop} and then apply the fusion rule for Verma modules of generic weights $\qa^\alpha,\qa^\beta,\qa^{\alpha+\beta}\notin\pm\qa^\ZZ $
\begin{equation}
\label{fusverma}
\VV_\alpha\otimes\VV_\beta\cong\bigoplus_{n\geq 0}\VV_{\alpha+\beta-1-2n}
\end{equation}
which can be obtained, for example, by recursively constructing highest-weight vectors from each $\VV_{\alpha+\beta-1-2n}$ and matching the dimensions of the weight spaces (see \cite{JACKSON20111689} for an explicit construction). Then, assuming $\qa^{\alpha_l}, \qa^{\alpha_r}, \qa^{\alpha_l+\alpha_r}$ are generic, we have
\begin{equation}
\label{2buqdec}
\begin{aligned}
\Hbb & =\left(\bigoplus_{n\geq 0}\VV_{\alpha_l+\alpha_r-1-2n}\right)\otimes (\CC^2)^{\otimes N}\\
& =\bigoplus_{n\geq 0} \bigoplus_{k=0}^N {N\choose k} \VV_{\alpha_l+\alpha_r+N-1-2(n+k)}\\
& =\bigoplus_{m\geq -N/2} d_m \VV_{\alpha_l+\alpha_r-1-2m}=:\bigoplus_{m\geq -N/2}\Hs_m \,,
\end{aligned}
\end{equation}
where the dimensions of multiplicity spaces are
\begin{equation}
\label{dm}
d_m:=\sum_{k=0}^{\min(N,m+N/2)} {N\choose k}\,.
\end{equation}
In particular, $d_m=2^N$ for $m\geq N/2$. 

Now, using the fact that
\begin{equation*}
\CCC_{\VV_{\alpha}}=\qa^{\alpha}+\qa^{-\alpha} \,,
\end{equation*} 
we have
\begin{equation*}
\CCC_{\Hs_m}=\qa^{\alpha-(2m+1)}+\qa^{-\alpha+(2m+1)}\,,
\end{equation*}
so
\begin{equation*}
\begin{aligned}
Y_m:=Y_{\Hs_m} & =\frac{\qa^{\alpha_l+\alpha_r+1}+\qa^{-\alpha_l-\alpha_r-1}-\qa^{\alpha-(2m+1)}-\qa^{-\alpha+(2m+1)}}{\{\alpha_l\} \{\alpha_r\}}\\
& = \frac{[m+1]_\qa[\alpha_l+\alpha_r-m]_\qa}{[\alpha_l]_\qa [\alpha_r]_\qa}\,.
\end{aligned}
\end{equation*}

Finally, notice that since $\Uq$ commutes with $\utwoBlob$, the subspace $\tilde{\Hs}_m\subset\Hs_m$ of highest-weight vectors is stable by the action of $\twoBlobm$. Moreover, for all $k\geq 0$, $\FF^k\tilde{\Hs}_m\cong\tilde{\Hs}_m$ as $\twoBlobm$-modules, so even though we do not know if $\tilde{\Hs}_m$ is irreducible we can still write 
\begin{equation*}
\Hs_m=\tilde{\Hs}_m\otimes\VV_{\alpha_l+\alpha_r-1-2m}\,.
\end{equation*}

To sum up:
\begin{prop}
\label{2bpropqgen}
For $\qa\in\CC\backslash\qa^{i\pi\QQ}$, $\qa^{\alpha_l},\qa^{\alpha_r}\in\CC\backslash\{\pm\qa^\ZZ\}$ such that $\qa^{\alpha_l+\alpha_r}\notin\pm\qa^\ZZ$ and $N\in 2\mathbb{N}^*$, $\Hbb$ decomposes as a $(\utwoBlob,\Uq)$-bimodule
\begin{equation}
\Hbb=\bigoplus_{m\geq -N/2}\tilde{\Hs}_m\otimes\VV_{\alpha_l+\alpha_r-1-2m}\,.
\end{equation}
Moreover, for all $m\geq -N/2$, $\tilde{\Hs}_m$ is a $d_m$-dimensional representation of the two-boundary Temperley-Lieb algebra $\twoBlobm$ with
\begin{equation}
Y_m =\frac{[m+1]_\qa[\alpha_l+\alpha_r-m]_\qa}{[\alpha_l]_\qa [\alpha_r]_\qa}\,.
\end{equation}
\end{prop}

If $\qa=e^{\frac{i\pi}{p}}$ is a $2p$-th root of unity some adjustments are needed. We can still use the braidings \eqref{Rprop}  defined by the $\RRR$-matrix \eqref{Rmatroot} to bring $\VV_{\alpha_l}$ and $\VV_{\alpha_r}$ together to the left but the fusion rule \eqref{fusverma} now becomes 
\begin{equation}
\VV_\alpha\otimes\VV_\beta=\bigoplus_{n=0}^{p-1}\VV_{\alpha+\beta+p-1-2n}
\end{equation}
for $\alpha,\beta,\alpha+\beta\in\CC\backslash\ZZ$ \cite{geermodtr0}. Second, because $\qa^{2p}=1$, we have
\begin{equation}
\CCC_{\VV_{\alpha+2p}}=-\qa^\alpha-\qa^{-\alpha}=\CCC_{\VV_{\alpha}} \,,
\end{equation}
so modules $\VV_{\alpha}$ with $\alpha$'s differing by a multiple of $2p$ will have the same value of~$Y$.

Now write $N=qp+r$, $q\in\NN$, $0\leq r\leq p-1$. Assuming $\alpha_l,\alpha_r,\alpha_l+\alpha_r\in\CC\backslash\ZZ$, we have
\begin{equation}
\label{Ydecomproot}
\begin{aligned}
\Hbb & =\left(\bigoplus_{n=0}^{p-1}\VV_{\alpha_l+\alpha_r+p-1-2n}\right)\otimes (\CC^2)^{\otimes N}\\
& =\bigoplus_{n= 0}^{p-1} \bigoplus_{k=0}^N {N\choose k} \VV_{\alpha_l+\alpha_r+N+p-1-2(n+k)}=:\bigoplus_{m=0}^{p-1} \underline{\Hs}_m \,,
\end{aligned}
\end{equation}
where
\begin{equation}
\label{2bd1}
\begin{aligned}
\underline{\Hs}_m:= & \left(\sum_{k=0}^{m}{N\choose k}\right)\VV_{\alpha_l+\alpha_r+N+p-1-2m}\\
&\oplus\, \bigoplus_{s=1}^{q-1}\left(\sum_{k=0}^{p-1}{N\choose ps-k+m}\right)\VV_{\alpha_l+\alpha_r+N+p-1-2m-2ps}\\
& \oplus\left(\sum_{k=0}^{r-m+p-1}{N\choose N-k}\right)\VV_{\alpha_l+\alpha_r+N+p-1-2m-2pq}=: \bigoplus_{s=0}^{q}\underline{\Hs}_m^{(s)}
\end{aligned}
\end{equation}
for $m\geq r$, and
\begin{equation}
\label{2bd2}
\begin{aligned}
\underline{\Hs}_m:= & \left(\sum_{k=0}^{m}{N\choose k}\right)\VV_{\alpha_l+\alpha_r+N+p-1-2m}\\
&\oplus\,  \bigoplus_{s=1}^{q}\left(\sum_{k=0}^{p-1}{N\choose ps-k+m}\right)\VV_{\alpha_l+\alpha_r+N+p-1-2m-2ps}\\
& \oplus \left(\sum_{k=0}^{r-m-1}{N\choose N-k}\right)\VV_{\alpha_l+\alpha_r+N+p-1-2m-2p(q+1)}=: \bigoplus_{s=0}^{q+1}\underline{\Hs}_m^{(s)}
\end{aligned}
\end{equation}
for $m<r$. From this we immediately see that $\dim \underline{\Hs}_m=p\cdot 2^N$. Moreover, assuming $N$ even and denoting $\bar{m}:=m-N/2$ mod $p$, we have
\begin{equation*}
\begin{aligned}
Y_m:=Y_{\underline{\Hs}_m} & =\frac{\qa^{\alpha_l+\alpha_r+1}+\qa^{-\alpha_l-\alpha_r-1}-\qa^{\alpha_l+\alpha_r+N-1-2m}-\qa^{-\alpha_l-\alpha_r-N+1+2m}}{\{\alpha_l\} \{\alpha_r\}}\\
& = \frac{[\bar{m}+1]_\qa[\alpha_r+\alpha_l-\bar{m}]_\qa}{[\alpha_r]_\qa [\alpha_l]_\qa}\,.
\end{aligned}
\end{equation*}
Note that $Y$ now takes only $p$ distinct values.

As before, introducing the subspace $\tilde{\underline{\Hs}}_m\subset\underline{\Hs}_m$ of highest-weight vectors, we have that $\FF^k\tilde{\underline{\Hs}}_m\cong\tilde{\underline{\Hs}}_m$ as $\twoBlobm$-modules for all $0\leq k\leq p-1$. Note, however, that by decompositions \eqref{2bd1}-\eqref{2bd2} and since $\Uqu$ commutes with $\utwoBlob$, $\tilde{\underline{\Hs}}_m$ is reducible and decomposes  into a direct sum of smaller spaces $\tilde{\underline{\Hs}}_m^{(s)}$. As we do not know whether $\Uqu$ is the full centraliser or not, we cannot say whether the $q+1$ (resp.\ $q+2$) summands $\tilde{\underline{\Hs}}_m^{(s)}$ within $\tilde{\underline{\Hs}}_m$ appearing in \eqref{2bd1} (resp.\ \eqref{2bd2}) are indeed the irreducible $\twoBlobm$-summands of $\tilde{\underline{\Hs}}_m$ or if they are further decomposed into irreducible submodules. In any case, we can still write
\begin{equation*}
\underline{\Hs}_m^{(s)}=\tilde{\underline{\Hs}}_m^{(s)}\otimes\VV_{\alpha_l+\alpha_r+N+p-1-2m-2ps}\,.
\end{equation*}

To sum up, we have the following statement.

\begin{prop}
\label{2bpropqroot}
For $\qa=e^{\frac{i\pi}{p}}$ with $p\in\NN\backslash\{0,1\}$, $\alpha_l,\alpha_r\in\CC\backslash\ZZ$ such that $\alpha_l+\alpha_r\notin\ZZ$ and $N=qp+r\in 2\mathbb{N}^*$ with $q\in\NN$, $0\leq r\leq p-1$, $\Hbb$ decomposes as a $(\utwoBlob,\Uqu)$-bimodule
\begin{equation}
\Hbb=\bigoplus_{m=0}^{p-1}\bigoplus_{s=0}^{q+\epsilon_m}\tilde{\underline{\Hs}}_m^{(s)}\otimes\VV_{\alpha_l+\alpha_r+N+p-1-2m-2ps} \,,
\end{equation}
where $\epsilon_m=0$ if $m\geq r$, and $\epsilon_m=1$ if $m<r$. Moreover, for all $0\leq m\leq p-1$ and $0\leq s\leq q+\epsilon_m$, $\tilde{\underline{\Hs}}_m^{(s)}$ is a representation of the two-boundary Temperley-Lieb algebra $\twoBlobm$ with
\begin{equation}
Y_m =\frac{[\bar{m}+1]_\qa[\alpha_l+\alpha_r-\bar{m}]_\qa}{[\alpha_l]_\qa [\alpha_r]_\qa}\,,\qquad \bar{m}:=m-N/2 ~\text{mod}~p\,,
\end{equation}
and whose dimension is given by the multiplicities in \eqref{2bd1}-\eqref{2bd2}. In particular,
\begin{equation}
\label{hmdec}
\tilde{\underline{\Hs}}_m:=\bigoplus_{s=0}^{q+\epsilon_m}\tilde{\underline{\Hs}}_m^{(s)}
\end{equation}
is a reducible $2^N$-dimensional representation of $\twoBlobm$.
\qed
\end{prop}

\medskip

Unfortunately, the representation theory of $\twoBlob$ is not completely understood even for generic $\qa$ but non-generic values of the weights $\delta$, $y_l$, $y_r$ and $Y$, as is the case here, where the $Y_m$ (which are functions of $\delta$, $y_l$ and $y_r$) take exactly the ``bad'' values at which standard $\twoBlob$-modules become reducible (see \cite[Corollary 5.18]{de_Gier_2009} and \cite[A.4.4]{dubail:tel-00555624}), so proving stronger statements than Propositions \ref{2bpropqgen} and~\ref{2bpropqroot} requires a special study which we leave for another work.

\subsection{Some conjectures on $\tilde{\Hs}_m$ and $\tilde{\underline{\Hs}}_m^{(s)}$}
\label{conjsec}

Let us finish by making some reasonable conjectures. It was shown in \cite[Corollary 5.18]{de_Gier_2009} that for generic $\qa$ the $2^N$-dimensional vacuum module $\WW_0$ of $\twoBlobm$ is non-irreducible for $-N/2\leq m\leq N/2-1$, and irreducible for $m\geq N/2$. But we know that $\tilde{\Hs}_m$ from Proposition~\ref{2bpropqgen} is of dimension $2^N$ exactly for $m\geq N/2$. Therefore it is tempting to identify it with the vacuum module of $\twoBlobm$. As for the other values of $m$, it was shown in \cite{de_Gier_2009} that the (non-irreducible) vacuum module $\WW_0$ of $\twoBlobm$ contains a unique non-trivial stable subspace $\overline{\WW}_m$, of dimension equal to $d_m$ from~\eqref{dm}, for $-N/2\leq m\leq -1$, and equal to $2^N-d_m$, for $0\leq m\leq N/2-1$, and that the quotient $\WW_0/\overline{\WW}_m$ is irreducible. This is a strong indication that $\tilde{\Hs}_m$ is isomorphic to $\overline{\WW}_m$ for $-N/2\leq m\leq -1$ and to $\WW_0/\overline{\WW}_m$ for $0\leq m\leq N/2-1$.

The submodules $\overline{\WW}_m$ can actually be described more explicitly. Similarly to the blob algebra, one can construct irreducible standard two-boundary modules with $2j$ through lines, $1\leq j\leq N/2$, but for which there are now four possible choices depending on whether the rightmost and leftmost through lines carry a blob or an anti-blob, denoted $\WW_{j}^{\bullet\sqq}$,  $\WW_{j}^{\circ\sqq}$, $\WW_{j}^{\bullet\usqq}$ and $\WW_{j}^{\circ\usqq}$. As the presence of through lines prohibits the formation of a loop touching both boundaries, these modules are independent of $Y$. Nevertheless, for non-generic values of $Y$, they appear as stable subspaces of the vacuum module $\WW_0$ (which does depend on $Y$). More precisely, it was conjectured in \cite{de_Gier_2009} and proven in \cite[A.4.4]{dubail:tel-00555624}, that $\overline{\WW}_m\cong\WW_{-m}^{\bullet\sqq}$ and $\overline{\WW}_m\cong\WW_{m+1}^{\circ\usqq}$ as $\twoBlobm$-modules for $-N/2\leq m\leq -1$ and $0\leq m\leq N/2-1$ respectively.\footnote{There are similar isomorphisms for $\WW_{j}^{\circ\sqq}$ and $\WW_{j}^{\bullet\usqq}$ but we will not need them here.} We thus arrive at the following conjecture.

\begin{conj}
\label{conj1}
For $\qa\in\CC\backslash\qa^{i\pi\QQ}$, $\qa^{\alpha_l},\qa^{\alpha_r}\in\CC\backslash\{\pm\qa^\ZZ\}$ such that $\qa^{\alpha_l+\alpha_r}\notin\pm\qa^\ZZ$ and $N\in 2\mathbb{N}^*$, the $\twoBlobm$-modules $\tilde{\Hs}_m$ from  Proposition~\ref{2bpropqgen}  are given by 
\begin{equation}
\begin{aligned}
\tilde{\Hs}_m & \cong\WW_{-m}^{\bullet\sqq} \qquad\qquad\;\, \text{for}~-N/2\leq m\leq -1\,,\\
\tilde{\Hs}_m & \cong\WW_0/\WW_{m+1}^{\circ\usqq} \qquad \text{for}~0\leq m\leq N/2-1\,,\\
\tilde{\Hs}_m & \cong\WW_0 \qquad\qquad\quad \text{for}~ N/2\leq m\,.
\end{aligned}
\end{equation}
In particular, they are irreducible, $\twoBlobm$ and $\Uq$ are mutual maximal centralisers on each $\Hs_m$  for all $-N/2\leq m$ and $\utwoBlob$ and $\Uq$ are mutual maximal centralisers on $\Hbb$.
\end{conj}

This conjecture is definitely true for $\tilde{\Hs}_{-N/2}=\CC\ket{0}\otimes\ket{\uparrow}^{\otimes N}\otimes\ket{0}$ and we believe it should hold for $m\geq 1-N/2$ as well. Note also that independently of its validity, we can say for sure that the representations $\tilde{\Hs}_m$ are not faithful, because their dimension is too small to contain all possible irreducible $\twoBlobm$-modules.

For $\qa$ a root of unity even less is known about the representation theory of $\twoBlobm$. Nevertheless, the spaces $\tilde{\underline{\Hs}}_m$ are $2^N$-dimensional, which makes it plausible that they are again related to the vacuum module $\WW_0$. It is however clear that $\tilde{\underline{\Hs}}_m$ cannot be isomorphic to $\WW_0$ because according to \eqref{hmdec} it decomposes into a direct sum of $\twoBlobm$-modules, whereas $\WW_0$ is expected to have a more complicated indecomposable structure. Still, we can consider a semi-simplified version of $\WW_0$ obtained by treating all subquotients as independent direct summands. We then arrive at the following conjecture.

\begin{conj}
\label{conj2}
For $\qa=e^{\frac{i\pi}{p}}$ with $p\in\NN\backslash\{0,1\}$, $\alpha_l,\alpha_r\in\CC\backslash\ZZ$ such that $\alpha_l+\alpha_r\notin\ZZ$, the module $\tilde{\underline{\Hs}}_m$ from Proposition~\ref{2bpropqroot} is isomorphic to the $2^N$-dimensional semi-simplified vacuum module of $\twoBlobm$. In particular, the $\twoBlobm$-modules $\tilde{\underline{\Hs}}_m^{(s)}$ are irreducible for all $0\leq m\leq p-1$, $0\leq s\leq q+\epsilon_m$, and $\twoBlobm$ and $\Uqu$ are mutual maximal centralisers on each $\underline{\Hs}_m$. Consequently,  $\utwoBlob$ and $\Uqu$ are mutual maximal centralisers on $\Hbb$.
\end{conj}

The representations $\tilde{\Hs}_m$ are not expected to be faithful, but we cannot prove it rigorously at present, as the classification of irreducible $\twoBlobm$-modules is not sufficiently known.

We believe that the appearance of non-generic $\twoBlob$-modules in our spin chain is not accidental, and that this phenomenon actually plays an important role in the two-boundary model. Note also that contrary to the one-boundary model, the algebraic decompositions for generic $\qa$ and at roots of unity differ sensibly.

\paragraph{Example.}

At $\qa=i$, the weight of a left (resp.\ right) blobbed loop is given by $y_r=\cot\frac{\pi\alpha_l}{2}$ (resp.\ $y_r=\cot\frac{\pi\alpha_l}{2}$). By \eqref{weights2broot}, central element $Y$ corresponding to the weight of a loop with both blobs  can be expressed as
\begin{equation*}
Y=\frac{2\sin(\frac{\pi(\alpha_l+\alpha_r)}{2})+\CCC_{\Hbb}}{4\sin(\frac{\pi\alpha_l}{2})\sin(\frac{\pi\alpha_r}{2})}\,.
\end{equation*}
By the fusion rules of $\Uqu$, assuming $\alpha_l,\alpha_r,\alpha_l+\alpha_r\notin\ZZ$,
\begin{equation*}
\Hbb=\bigoplus_{k=0}^{N+1} {N+1\choose k} \VV_{\alpha_l+\alpha_r+N+1-2k} \,,
\end{equation*}
and we have
\begin{equation*}
\CCC_{\VV_{\alpha_l+\alpha_r+N+1-2k}}=2(-1)^{k+N/2}\sin(\frac{\pi(\alpha_l+\alpha_r)}{2}) \,,
\end{equation*}
so the only eigenvalues of $Y$ are $0$ and $y_{lr}:=y_l+y_r$. The corresponding eigenspace decomposition is given by
\begin{equation*}
\Hbb=\Ker(Y-y_{lr})\oplus\Ker(Y)=\left\{
    \begin{array}{ll}
        \underline{\Hs}_0\oplus\underline{\Hs}_1 & \text{if $N/2$ is even}\\
        \underline{\Hs}_1\oplus\underline{\Hs}_0 & \text{if $N/2$ is odd}
    \end{array}
\right.
\end{equation*}
with
\begin{equation}
\label{qidecomp}
\begin{aligned}
\underline{\Hs}_0 & :=\bigoplus_{k=0}^{N/2} {N+1\choose 2k} \VV_{\alpha_l+\alpha_r+N+1-4k}\,,\\
\underline{\Hs}_1 & :=\bigoplus_{k=0}^{N/2} {N+1\choose 2k+1} \VV_{\alpha_l+\alpha_r+N-1-4k}\,.
\end{aligned}
\end{equation}
Note that $\dim \underline{\Hs}_0=\dim \underline{\Hs}_1=2\cdot 2^{N}$, and that $\underline{\Hs}_0$ (resp. $\underline{\Hs}_1$) is the odd (resp. even) fermionic subspace of $\Hbb$.

Therefore, $\Hbb$ decomposes into two halves, one carrying a representation of the two-boundary TL algebra $2\BB_{0,y_l,y_r,0,N}$ and the other of $2\BB_{0,y_l,y_r,y_{lr},N}$, or, equivalently, $\Hbb$ carries a representation of the universal two-boundary TL algebra $2\BB_{0,y_l,y_r,N}^{\mathrm{uni}}$. Its image commutes with the $\Uqu$ action.

Now, taking the highest-weight subspaces $\tilde{\underline{\Hs}}_0$ and $\tilde{\underline{\Hs}}_1$, we have
\begin{equation}
\label{qimod}
\underline{\Hs}_0=\tilde{\underline{\Hs}}_0\oplus\FF\tilde{\underline{\Hs}}_0\,,\qquad \Hs_1=\tilde{\underline{\Hs}}_1\oplus \FF\tilde{\underline{\Hs}}_1
\end{equation}
as $2\BB_{0,y_l,y_r,y_{lr},N}$ and $2\BB_{0,y_l,y_r,0,N}$ modules respectively. Since the fermionic number operator (and not only its parity) also commutes with $2\BB_{0,y_l,y_r,N}^{\mathrm{uni}}$, we can further decompose $\tilde{\underline{\Hs}}_0$ and $\tilde{\underline{\Hs}}_1$ according to \eqref{qidecomp} where every multiplicity space is conjecturally an irreducible module over $2\BB_{0,y_l,y_r,N}^{\mathrm{uni}}$.

Finally, we note that  the centraliser of the $\Uqu$-action on $\Hbb$ can be described in terms of the blob algebra. Indeed, if $\alpha_l,\alpha_r,\alpha_l+\alpha_r\notin\ZZ$ we have
\begin{equation*}
\Hbb\cong\VV_{\alpha_l+\alpha_r}\otimes\left(\CC^2\right)^{\otimes (N+1)}
\end{equation*}
as representations of $\Uqu$. Therefore, by Proposition \ref{blobreproot}, $\Hbb$ carries a faithful action of $\BB_{0,y,N+1}$ with $y=\cot{\frac{\pi(\alpha_l+\alpha_r)}{2}}$, centralising the $\Uqu$-action. From these considerations also follows the rather curious fact that all $\BB_{0,y,N+1}$-modules can be realised as either some $2\BB_{0,y_l,y_r,0,N}$-module or some $2\BB_{0,y_l,y_r,y_{lr},N}$-module.

\section{Summary and open questions}
\label{conclusion}

In this paper, we have constructed new $\Uq$-invariant boundary conditions for the open XXZ spin chain using infinite-dimensional Verma modules, or their truncated finite-dimensional analogues at roots of unity. Using free fermions, we computed the spectra of our new one-boundary and two-boundary Hamiltonians in the simplest case $\qa=i$. We were then able to investigate the scaling limit, and to connect our model with the $(\eta,\xi)$ ghost CFT on the upper-half plane with some specific boundary conditions on the real axis. 

In the remainder of the paper, we studied in full generality the symmetry properties of our modified XXZ spin chains and, specifically, the representations of the various lattice algebras that they give rise to. We showed that the Hilbert space of our one-boundary system carries  a representation of the blob algebra $\Blob$, and that the actions of $\Uq$ (or $\Uqu$) and $\Blob$ are mutual centralisers. We then identified the sectors of our spin chain with standard (irreducible) blob modules, thereby showing that this spin-chain representation is faithful and obtaining the $(\Blob,\Uq)$-bimodule decomposition of the Hilbert space (Propositions \ref{blobrepqgen}-\ref{blobreproot}). 

As for the two-boundary spin chain, we showed that it carries a representation of the universal two-boundary Temperley-Lieb algebra (Propositions \ref{repthtwo}-\ref{repthtworoot}). Expressing the central element $Y$ of the weight of a loop carrying both the left and the right blob in terms of the Casimir element, and using the $\Uq$ (or $\Uqu$) decomposition, we were able to further decompose the Hilbert space into representations of the (usual) two-boundary Temperley-Lieb algebra $\twoBlob$, with a constant value of $Y$ in each sector (Propositions \ref{2bpropqgen}-\ref{2bpropqroot}). As these values of $Y$ are non-generic, we could not prove complete Schur-Weyl duality and instead conjectured on the $\twoBlob$-modules appearing in the decomposition (Conjectures \ref{conj1}-\ref{conj2}).

These algebraic results are not only of mathematical interest, but rather it is expected that they are the key to understanding the continuum limit of our models. Indeed, it was argued in \cite{Jacobsen_2008, Gainutdinov_2013} that the blob modules $\WW_j$ are the lattice analogues of certain Virasoro Verma modules. More precisely, if we define -- now not only for $\qa=i$ but any $\qa=e^{\frac{i\pi}{p}}$, $p\in [2,+\infty[$ (not necessarily rational) -- the central charge
\begin{equation*}
c:=1-\frac{6}{p(p-1)}
\end{equation*}
and the conformal weights\footnote{Note that for $p=2$ (that is $\qa=i$) we recover $c=-2$ and $h_{r,s}=\frac{(2r-s)^2-1}{8}$ as in \eqref{confw}.}
\begin{equation*}
h_{r,s}:=\frac{(pr-(p-1)s)^2-1}{4p(p-1)} \,,
\end{equation*}
then, treating $H_b$ as an abstract element of the blob algebra, we should have
\begin{equation}
\label{gensclim}
\lim_{M\to\infty}\lim_{N\to\infty}\tr_{\WW_{j}}^{<M} q^{\frac{N}{\pi v_F}(H_{b}-Ne_{\rm b}-E_{\rm s})}=\frac{q^{-\frac{c}{24}}}{P(q)}q^{h_{\alpha,\alpha+2j}}
\end{equation}
with $v_{\rm F}:=p\sin{\frac{\pi}{p}}$ given by \eqref{fermiv}, and $e_{\rm b}$ and $E_{\rm s}$ being respectively the new bulk energy per site and the surface energy at $\qa$. In other words, the scaling limit of $H_b|_{\WW_j}$ can be identified with the $L_0$ generator of the Virasoro algebra represented on a Verma module of conformal weight $h_{\alpha,\alpha+2j}$. To formulate this conjecture, a different representation of the blob algebra -- the so-called cabling realisation -- was used in \cite{Jacobsen_2008}.\footnote{In \cite{Nichols2006} yet another representation was considered, but the link with our spin chain is less straightforward. We will not discuss it here.} It is constructed directly from the Temperley-Lieb algebra by adding $r-1$ $\CC^2$-sites at the leftmost boundary, and then applying a Jones-Wenzl projector on them to single out the spin-$(r-1)/2$ summand appearing in the $\Uq$ decomposition, or, said differently, by replacing our $\VV_\alpha$ with a spin-$(r-1)/2$ representation. From our previous computations it is easy to see that the blob weight is then given by
\begin{equation*}
y=\frac{[r+1]_\qa}{[r]_\qa} \,,
\end{equation*}
and so the $r$ in \cite{Jacobsen_2008} can be identified with a particular choice of our $\alpha$. The advantage of our representation of the blob algebra is that it makes it possible to reach all values of $y$ and not only the discrete set given by $r\in\NN^{*}$. Moreover it is known~\cite{Gainutdinov_2013} that the cabling representation is not faithful while our spin chains provide a faithful representation of the blob algebra.

In the present paper, we have proven \eqref{gensclim} rigorously for $\qa=i$, by computing the scaling limit of an explicit spin chain \eqref{wpartfun}-\eqref{wpartfunh}, and then showing that its sectors $W_{N/2+j}$ can by identified with standard blob modules $\WW_{j}$ according to \eqref{bimodqroot}-\eqref{bidecompqi}.

As for the two-boundary system, exact expressions for the partition functions in all standard representations of $\twoBlob$ were proposed in \cite{Dubail_2009}. Since we did not manage to identify the $\twoBlob$-modules appearing in the decomposition of $\Hbb$ we cannot compare our results with those of that paper. Let us note, however, that if we assume Conjecture \ref{conj2} to be true, then, at $\qa=i$ the prediction of \cite{Dubail_2009} matches \eqref{twobqi}. 

\medskip

From this, the most immediate task ahead is to generalise our result for the spectrum of the one-boundary and two-boundary XX model to any $\qa$ to prove the statements of \cite{Nichols2006, Jacobsen_2008, Dubail_2009} for all values of the parameters. This requires the introduction of the rather technical formalism of (boundary) Algebraic Bethe Ansatz and will be performed in the next paper. Furthermore, it seems necessary to better understand the algebraic properties of the two-boundary system -- in particular regarding Conjectures \ref{conj1}-\ref{conj2} -- and their link to the spectrum of the corresponding Hamiltonian. Indeed, these conjectures suggest that the two-boundary spin chain \eqref{intro:Hbb-space} contains a discrete series (infinite for generic $\qa$ and finite with $p$ sectors for a $2p$-th root of unity) of non-diagonal XXZ models, which are known to be vacuum standard modules over the two-boundary TL algebra \cite{de_Gier_2009}. More precisely, using the conventions from~\cite[Sec.\,3.4]{Dubail_2009}, the 6 parameters in the integrable non-diagonal XXZ boundary conditions are related to the Verma module weights $\alpha_{l/r}$, couplings $\mu_{l/r}$ and values $Y_m$ of $Y$ as follows : $r_1=\alpha_l$, $r_2=\alpha_r$, $\lambda_1=\mu_l$, $\lambda_2=\mu_r$, $r_{12}=s_2-s_1=\alpha_l+\alpha_r-2m-1$. As it turns out, computing the spectrum of $H_{2b}$ (for arbitrary $\qa$) via Algebraic Bethe Ansatz is somewhat simpler than that of the corresponding non-diagonal XXZ models due to its greater symmetry. In the next paper we will also see that, as for $\qa=i$, the scaling limit of $H_{2b}$ together with Conjectures \ref{conj1}-\ref{conj2} are consistent with the predictions of \cite{Dubail_2009}.

Finally, from a physical perspective, it can be argued that the two-boundary system ``is the fusion of two one-boundary systems'' -- but this idea is still not totally understood. A better grasp of this question would be relevant for a wide range of open problems, among which understanding the fusion of Virasoro Verma modules with generic conformal weights and the CFT properties of the periodic XXZ spin chain. Equation \eqref{sugg} suggests that the representation theory of $\Uq$ and $\Uqu$ plays an important role. Lastly, we notice that the $Y$ operator from Section~\ref{twoboundarysec} has diagrammatically the form of a hoop operator, up to a constant term, which resembles the topological defect $Y$ operator from \cite{topdef} in periodic models. We believe that these two operators are closely related. These topics will be investigated in detail in a future paper.

\section*{Acknowledgements}
We are very thankful to J. L. Jacobsen for very valuable discussions and the interest in this work. The authors also thank Institut Pascal, Universit\'e Paris-Saclay and the organisers of the workshop Bootstat 2021 where part of this work was done. This work was supported in part by the ERC Advanced Grant NuQFT. The work of A.M.G. was supported by the CNRS, and partially by the ANR grant JCJC ANR-18-CE40-0001 and the RSF Grant No.\ 20-61-46005. 
A.M.G. is also grateful to IPhT in Saclay for its kind hospitality in 2021 and 2022.

\appendix

\section{General properties of the spectrum of $H_b$ for $\qa=i$}
\label{generalpropspec}

In this appendix we provide a systematic study of the the spectral equation \eqref{waveqU} for all $\mu$, $y$ and $N$. Since it is a polynomial equation of degree $N$, it will have $N$ solutions $\lambda_k$. We want to know for which choices of $\mu$ and $y$ these $\lambda_k$, and the associated momenta $p_k$ defined by $\lambda_k=2\cos p_k$, will be real. It is sufficient to consider $\mu$ and $y$ real, since otherwise \eqref{waveqU} will obviously have complex solutions. As already discussed, we can assume without loss of generality that $\mu>0$.

Recall that for real momenta, \eqref{waveqU} is equivalent to \eqref{simpeq}
\begin{equation*}
N\xi+\varphi(\xi)=k\pi \,, \qquad k\in\mathbb{Z},~\xi\in]0,\pi[\,.
\end{equation*}
We are thus reduced to studying the function $\varphi$.

First note that $\varphi(0), \varphi(\pi)\in\{0,\pi\}$. Moreover $\frac{(2-\mu y)\cos{\xi}+\mu}{\mu y \sin{\xi}}$ can become infinite only at these points, so we can work with the usual $\arccot{}$ function and not its multivalued generalisation. Let us define the winding number 
\begin{equation*}
w:=N+\frac{\varphi(\pi)-\varphi(0)}{\pi}\in \mathbb{N} \,.
\end{equation*}
A simple analysis then gives the table of variations
\begin{equation*}
\begin{tikzpicture}
   \tkzTabInit[espcl =2]{$y$ / 1 , $w$ / 1}{$-\infty$, $\frac{2-\mu}{\mu}$, $\frac{2+\mu}{\mu}$, $+\infty$}
   \tkzTabLine{,N-1,d, N,d, N-1}
\end{tikzpicture}
\end{equation*}
for $2<\mu$, and
\begin{equation*}
\begin{tikzpicture}
   \tkzTabInit[espcl =2]{$y$ / 1 , $w$ / 1}{$-\infty$, $0$, $\frac{2-\mu}{\mu}$, $\frac{2+\mu}{\mu}$, $+\infty$}
   \tkzTabLine{,N-1,d, N+1,d,N,d, N-1}
\end{tikzpicture}
\end{equation*}
for $0<\mu<2$. Since for large $N$ the function $\xi\mapsto N\xi+\varphi(\xi)$ is strictly increasing, this already tells us that the number of real momenta is exactly $w-1$ as $N\to\infty$ (recall that the endpoints $0$ and $\pi$ are not solutions). As already discussed in Section \ref{compspec}, this remains true for all $N$ if $w=N+1$, in which case all the solutions and momenta are real. 

For the other cases we have to determine if the one or two remaining solutions are complex, or real but outside $]-2,2[$. Since \eqref{waveqU} is a polynomial equation with real coefficients, complex solutions must come in conjugated pairs. Therefore if $w=N$, the remaining solution must also be real. If $w=N-1$, however, a finer analysis is needed. By varying $\mu$ and $y$ outside the domain $w=N+1$ we can follow the real roots of $P$ and distinguish two cases: either two roots collide and become complex, or they leave $]-2,2[$ one after the other. If $0<\mu<2$ and $y<0$ the former is true, and if $y>\frac{\mu+2}{\mu}$ (and any $\mu>0$) -- the latter, so in these domains we will have respectively $N-2$ and $N$ real solutions. 

The last and most complicated case is when $\mu>2$ and $y<\frac{\mu-2}{\mu}$, because in this domain the two colliding roots are outside $]-2,2[$. The system $P_{\mu,y}(\lambda)=P_{\mu,y}'(\lambda)=0$ satisfied at the collision point then defines a non-trivial curve in $(\mu,y)$ space. Although there is no analytic expression for it, a good enough approximation can be obtained by considering the limit $\mu\to+\infty$, $y\sim \tilde{y}\mu$, $\lambda\sim \tilde{\lambda}\mu$, with $\tilde{y},\tilde{\lambda}\in\RR$. Plugging this ansatz into \eqref{waveqU} and keeping only the highest order in $\mu$, we obtain
\begin{equation}
\label{critcurve}
\tilde{\lambda}^2+\tilde{\lambda}-\tilde{y}=0\,.
\end{equation}
The double root appears when the discriminant of this polynomial (in $\tilde{\lambda}$) vanishes, that is, when $\tilde{y}=-\frac{1}{4}$, $\tilde{\lambda}=-1/2$. To obtain the next-order term we now plug $y\sim -\mu/4+\tilde{y}'$, $\lambda\sim -\mu/2+\tilde{\lambda}'$ into the spectral equation. We have
\begin{equation}
\label{consterm}
\left(\frac{N}{2}-(N-1)+\frac{N-2}{2}\right)\tilde{\lambda}'+\tilde{y}'=0\,\Rightarrow\,\tilde{y}'=0\,.
\end{equation}
Thus the critical curve has an asymptote of equation $y=-\mu/4$ as $\mu\to\infty$. This is confirmed numerically.

These results for the large-$N$ limit are summarised in Figure \ref{diagram}. All the domains of the graph are exact, except for the green curve close to the critical point $\bullet=(2,0)$ at which the black, red and green curves must meet, since $-1$ is then a double root of $P$. 
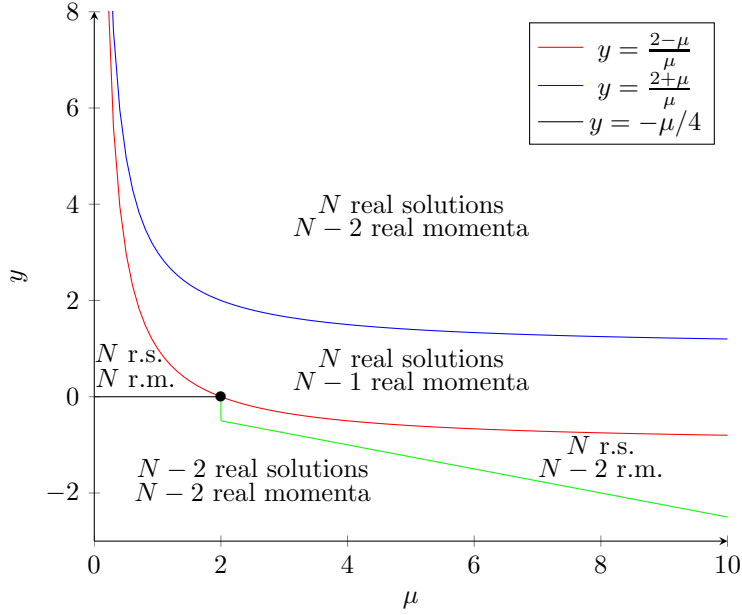
\begin{figure}[!t]
\centering
\begin{tikzpicture}
\begin{axis}[
    axis lines = left,
    xlabel = $\mu$,
    ylabel = {$y$},
    ymax=8,
    ymin=-3,
    xmax=10,
    xmin=0,
]
\addplot [
    domain=0:10, 
    samples=100, 
    color=red,
]
{(2-x)/x};
\addlegendentry{$y=\frac{2-\mu}{\mu}$}

 \addplot [
    domain=0:10, 
    samples=100, 
    color=blue,
]
{(2+x)/x};
\addlegendentry{$y=\frac{2+\mu}{\mu}$}

\addplot [
    domain=0:2, 
    samples=100, 
    color=black,
    ]
    {0};
    
 \addplot [
    domain=2:10, 
    samples=100, 
    color=green,
]
{-x/4};
\addlegendentry{$y=-\mu/4$}

\addplot [green] coordinates {
(2,-0.5) (2,0)
};
    
\node[] at (axis cs: 5,4) {$N$ real solutions};
\node[] at (axis cs: 5,3.5) {$N-2$ real momenta};
\node[] at (axis cs: 5,0.8) {$N$ real solutions};
\node[] at (axis cs: 5,0.3) {$N-1$ real momenta};
\node[] at (axis cs: 2.5,-1.5) {$N-2$ real solutions};
\node[] at (axis cs: 2.5,-2) {$N-2$ real momenta};
\node[] at (axis cs: 0.65,0.4) {$N$ r.m.};
\node[] at (axis cs: 0.55,0.9) {$N$ r.s.};
\node[] at (axis cs: 8,-1) {$N$ r.s.};
\node[] at (axis cs: 8,-1.5) {$N-2$ r.m.};
\node[] at (axis cs: 2,0) {$\bullet$};
\end{axis}
\end{tikzpicture}
\caption{Reality properties of the solutions and associated momenta of the spectral equation \eqref{waveqU} in terms of $\mu$ and $y$.}
\label{diagram}
\end{figure}

To perform the same analysis at finite $N$ we have to take into account the fact that the function $\xi\mapsto N\xi+\varphi(\xi)$ can have local extrema, and so it might cross some multiple of $\pi$ more than once.

The corrections to the critical curves $y=\frac{2\pm\mu}{\mu}$ are easy to compute, since they correspond to the values of $\mu$ and $y$ for which a solution of \eqref{waveqU} hits $\pm 2$. In other words, they come from a local minimum right after $\xi=0$, or to a local maximum right before $\xi=\pi$. Then, doing this computation, we obtain the new critical curves $y=\frac{2\pm\mu}{\mu}\frac{N}{N-1}$.

For the other boundaries of Figure \ref{diagram}, local extrema generate two additional real roots and so only affect the black and green lines. If $0<\mu<2$ we have already seen in Section \ref{compspec} that the black line will be replaced by some curve $\yyy<0$ given by the solution of the system $P_{\mu,y}(\lambda)=P'_{\mu,y}(\lambda)=0$. The green line will also be corrected, but since equations \eqref{critcurve}-\eqref{consterm} are independent of $N$, its asymptote will still be $y=-\mu/4$ for all $N$. Finally, the critical point $\bullet$ at which the black, green and red lines meet will change. One can compute its coordinates explicitly by solving $P_{\mu,y}(-1)=P'_{\mu,y}(-1)=0$. Since
\begin{equation*}
U'_N(-1)=(-1)^{N+1}\frac{N(N+1)(N+2)}{3} \,,
\end{equation*}
we have
\begin{equation*}
\bullet=\left(2\frac{2N+1}{2N-1},-\frac{4N}{(2N-1)(N-1)}\right)\,.
\end{equation*}
As expected, $\bullet\to (2,0)$ as $N\to\infty$.

\section{Proof of the anti-commutation relation \eqref{anticomsc}}
\label{anticomp}

In this appendix we prove equation \eqref{anticomsc},
\begin{equation*}
\lim_{M\to\infty}\lim_{N\to\infty}\{\xi(z)^{<M},\eta(w)^{<M}\} =\frac{1}{w}\sum_{k\in\ZZ}(z/w)^{k+\tau}:=\delta_\tau(z,w)\,.
\end{equation*}

First, from \eqref{anticomc}, \eqref{expA}, and~\eqref{expB}
\begin{equation*}
\begin{aligned}
\{\theta_k^\dagger,\theta_{k'}\} & =A_0B_0'+\sum_{j=1}^N A_jB_j'\\
& =-\frac{1}{4\sin(\xi_k)\sin(\xi_{k'})}\sum_{j=0}^{N-1}\left((e^{i\xi_k}-i)e^{ij\xi_k}-(e^{-i\xi_k}-i)e^{-ij\xi_k}\right)\\
&\qquad\qquad\qquad\qquad\qquad\qquad\qquad\left((e^{i\xi_{k'}}-i)e^{ij\xi_{k'}}-(e^{-i\xi_{k'}}-i)e^{-ij\xi_{k'}}\right)\\
& ~~~ +A_0B_0'
\end{aligned}
\end{equation*}
where
\begin{equation*}
A_0B_0'=-\frac{e^{-\frac{i\pi\alpha}{2}}\cos(\frac{\pi\alpha}{2})\left(i\sin(N\xi_k)+\sin((N-1)\xi_k)\right)\left(i\sin(N\xi_{k'})+\sin((N-1)\xi_{k'})\right)}{\left(ie^{-\frac{i\pi\alpha}{2}}+\frac{\lambda_k}{\mu}\sin{\frac{\pi\alpha}{2}}\right)\left(ie^{-\frac{i\pi\alpha}{2}}+\frac{\lambda_{k'}}{\mu}\sin{\frac{\pi\alpha}{2}}\right)\sin(\xi_k)\sin(\xi_{k'})}\,.
\end{equation*}
Summing the series we obtain, if $k\neq k'$,
\begin{equation*}
\begin{aligned}
-i\sin(\xi_k)\sin(\xi_{k'})\{\theta_k^\dagger,\theta_{k'}\}= & -\frac{\cos(\frac{\xi_k+\xi_{k'}}{2})}{2\sin(\frac{\xi_k-\xi_{k'}}{2})}\sin(N(\xi_k-\xi_{k'}))\\
& +\frac{\cos(\frac{\xi_k-\xi_{k'}}{2})}{2\sin(\frac{\xi_k+\xi_{k'}}{2})}\sin(N(\xi_k+\xi_{k'}))-i\sin(N\xi_k)\sin(N\xi_{k'})\\
& -i\sin(\xi_k)\sin(\xi_{k'})A_0B_0'\,.
\end{aligned}
\end{equation*}
If $k=k'$ the first term on the right-hand side is replaced by $-N\cos(\frac{\xi_k+\xi_{k'}}{2})$.

For $k$ close to $N/2$ we have
\begin{equation*}
\xi_k=\frac{\left(k+\frac{\alpha-1}{2}\right)\pi}{N}+o(1/N)\,.
\end{equation*}
Set $\ell:=k-N/2$. We have
\begin{equation*}
\begin{aligned}
-\frac{\cos(\frac{\xi_k+\xi_{k'}}{2})}{2\sin(\frac{\xi_k-\xi_{k'}}{2})}\sin(N(\xi_k-\xi_{k'})) & =\pi\left(\ell+\frac{\alpha-1}{2}\right)\delta_{\ell\ell'}+O(1/N)\,,\\
\frac{\cos(\frac{\xi_k-\xi_{k'}}{2})}{2\sin(\frac{\xi_k+\xi_{k'}}{2})}\sin(N(\xi_k+\xi_{k'}))) & =\frac{1}{2}\sin(\pi(\ell+\ell'+\alpha-1))+O(1/N) \\
& =-\frac{(-1)^{\ell+\ell'}}{2}\sin(\pi\alpha)+O(1/N)\,,\\
-i\sin(N\xi_k)\sin(N\xi_{k'}) & =-i(-1)^{\ell+\ell'}\cos(\pi\alpha/2)^2+O(1/N)\,,
\end{aligned}
\end{equation*}
and
\begin{equation*}
\begin{aligned}
-i\sin(\xi_k)\sin(\xi_{k'})A_0B_0' & =-ie^{\frac{i\pi\alpha}{2}}\cos(\frac{\pi\alpha}{2})e^{-iN(\xi_k+\xi_{k'})}+O(1/N)\\
& =-ie^{\frac{i\pi\alpha}{2}}\cos(\frac{\pi\alpha}{2})(-1)^{\ell+\ell'}e^{-i\pi(\alpha-1)}+O(1/N)\\
& = \frac{i}{2}(-1)^{\ell+\ell'}\left(1+e^{-i\pi\alpha}\right)+O(1/N)\,.
\end{aligned}
\end{equation*}
Since
\begin{equation*}
\frac{i}{2}\left(1+e^{-i\pi\alpha}\right)-\frac{1}{2}\sin(\pi\alpha)-i\cos(\pi\alpha/2)^2=0 \,,
\end{equation*}
we finally obtain
\begin{equation}\label{eq:anti-com-N}
-i\sin(\xi_k)\sin(\xi_{k'})\{\theta_\ell^\dagger,\theta_{\ell'}\}=\pi\left(\ell+\frac{\alpha-1}{2}\right)\delta_{\ell\ell'}+O(1/N)
\end{equation}
or, in other words,
\begin{equation*}
\{\chi_{k+\tau}^+,\chi^-_{k'-\tau}\}=(k-\tau)\delta_{k+k'}+O(1/N)
\end{equation*} 
with $\chi_{k\pm\tau}^\pm$ defined in \eqref{chimodes}, and 
\begin{equation*}
\tau=\frac{1-\alpha}{2}\,.
\end{equation*}

Using this expression, we have
\begin{equation*}
\begin{aligned}
\lim_{M\to\infty}\lim_{N\to\infty}\{\xi(z)^{<M},\eta(w)^{<M}\} & =\sum_{k,k'\in\ZZ}\frac{\{\chi_{k+\tau}^+,\chi^-_{k'-\tau}\}}{k-\tau}z^{-k+\tau}w^{-k'-1-\tau}\\
& =\sum_{k,k'\in\ZZ}\delta_{k+k'}z^{-k+\tau}w^{-k'-1-\tau}\\
& =\frac{1}{w}\sum_{k\in\ZZ}(z/w)^{k+\tau}:=\delta_\tau(z,w) \,,
\end{aligned}
\end{equation*}
where $\delta_\tau(z,w)$ is the $\tau$-twisted delta function. This terminology is justified by the following fact. Suppose we have a function $f$ with expansion
\begin{equation*}
f(w)=\sum_{k\in\ZZ}a_k w^{k+\tau}\,.
\end{equation*}
Then
\begin{equation*}
\frac{1}{2i\pi}\oint\delta_\tau(z,w)f(w)\ddd w=\sum_{k\in\ZZ}a_{k}z^{k+\tau}=f(z)\,.
\end{equation*}
Similarly, if 
\begin{equation*}
g(z)=\sum_{k\in\ZZ}b_k z^{k-\tau}
\end{equation*}
then
\begin{equation*}
\frac{1}{2i\pi}\oint\delta_\tau(z,w)g(z)\ddd z=\sum_{k\in\ZZ}b_{k}w^{k-\tau}=g(w)\,.
\end{equation*}
Notice that, unlike the usual delta function, $\delta_\tau(z,w)$ is not symmetric under the permutation of $z$ and $w$. Actually, one can easily see that
\begin{equation*}
\delta_\tau(z,w)=\delta_{-\tau}(w,z)\qquad\text{and}\qquad\delta_{\tau+1}(z,w)=\delta_\tau(z,w)\,.
\end{equation*}
For $\tau=0\in\RR/\ZZ$ we recover $\delta_0(z,w)=\delta(z-w)$.

\bibliographystyle{hunsrt}
\bibliography{biblio}

\begin{thebibliography}{10}

\bibitem{PASQUIER1990523}
V.~Pasquier and H.~Saleur.
\newblock Common structures between finite systems and conformal field theories
  through quantum groups.
\newblock {\em Nuclear Physics B}, 330(2):523 -- 556, 1990.

\bibitem{Baxter:1982zz}
R.~J. Baxter.
\newblock {\em {Exactly solved models in statistical mechanics}}.
\newblock {Academic Press}, 1982.

\bibitem{Alcaraz_1987}
F.~C. Alcaraz, M.~N. Barber, M.~T. Batchelor, R.~J. Baxter, and G.~R.~W.
  Quispel.
\newblock Surface exponents of the quantum {X}{X}{Z}, {A}shkin-{T}eller and
  {P}otts models.
\newblock {\em Journal of Physics A: Mathematical and General},
  20(18):6397--6409, Dec 1987.

\bibitem{Kausch:1995py}
H.~G. Kausch.
\newblock {Curiosities at $c = -2$}.
\newblock 1995, arXiv:9510149 [hep-th].

\bibitem{Kausch_2000}
H.~G. Kausch.
\newblock Symplectic fermions.
\newblock {\em Nuclear Physics B}, 583(3):513–541, Sep 2000.

\bibitem{Jimbo1986AQO}
M.~Jimbo.
\newblock A $q$-analogue of {$Ugl(N+1)$}, {H}ecke algebra, and the
  {Y}ang-{B}axter equation.
\newblock {\em Letters in Mathematical Physics}, 11:247--252, 1986.

\bibitem{goodman}
F.~M. Goodman and H.~Wenzl.
\newblock The {T}emperley--{L}ieb algebra at roots of unity.
\newblock {\em Pacific J. of Math.}, 161(2):307–334, 1993.

\bibitem{martin1992commutants}
P.~P. Martin and D.~S. McAnally.
\newblock On commutants, dual pairs and non-semisimple algebras from
  statistical mechanics.
\newblock {\em International Journal of Modern Physics A}, 7(Supp.
  1B):675--705, 1992.

\bibitem{Martin1992ONSD}
P.~P. Martin.
\newblock On {S}chur-{W}eyl duality, ${A}_n$ {H}ecke algebras and quantum
  $sl(n)$ on $\otimes^{n+1}{C}^{N}$.
\newblock {\em International Journal of Modern Physics A}, 07:645--673, 1992.

\bibitem{Martin:1993jka}
P.~P. Martin and H.~Saleur.
\newblock {The Blob algebra and the periodic Temperley-Lieb algebra}.
\newblock {\em Lett. Math. Phys.}, 30:189, 1994.

\bibitem{Oh}
T.~Ohtsuki.
\newblock {\em {Quantum invariants. A study of knots, 3-manifolds, and their
  sets}}.
\newblock Series on Knots and Everything: v. 29. World Scientific, Dec 2001.

\bibitem{GPT}
N.~Geer, B.~Patureau-Mirand, and V.~Turaev.
\newblock Modified quantum dimensions and re-normalized link invariants.
\newblock {\em Compositio Mathematica}, 145(1):196--212, Jan 2009.

\bibitem{Kulre}
P.~P. Kulish and N.~Yu. Reshetikhin.
\newblock {Quantum linear problem for the sine-Gordon equation and higher
  representations}.
\newblock {\em Zap. Nauchn. Sem. LOMI}, 101:101--110, 1981.

\bibitem{KSk}
P.~P. Kulish and E.~K. Sklyanin.
\newblock {The general $U_q sl(2)$ invariant {XXZ} integrable quantum spin
  chain}.
\newblock {\em Journal of Physics A: Mathematical and General},
  24(8):L435--L439, Apr 1991.

\bibitem{twob2004}
S.~Mitra, B.~Nienhuis, J.~de~Gier, and M.~T. Batchelor.
\newblock Exact expressions for correlations in the ground state of the dense
  ${O}(1)$ loop model.
\newblock {\em Journal of Statistical Mechanics: Theory and Experiment},
  2004(09):P09010, Oct 2004.

\bibitem{de_Gier_2009}
J.~de~Gier and A.~Nichols.
\newblock The two-boundary {T}emperley–{L}ieb algebra.
\newblock {\em Journal of Algebra}, 321(4):1132–1167, Feb 2009.

\bibitem{Nichols2006}
A.~Nichols.
\newblock The {T}emperley–{L}ieb algebra and its generalizations in the
  {P}otts and {XXZ} models.
\newblock {\em Journal of Statistical Mechanics: Theory and Experiment},
  2006(01), Jan 2006.

\bibitem{Jacobsen_2008}
J.~L. Jacobsen and H.~Saleur.
\newblock Conformal boundary loop models.
\newblock {\em Nuclear Physics B}, 788(3):137–166, Jan 2008.

\bibitem{Grimm:1990gg}
U.~Grimm and V.~Rittenberg.
\newblock {Null states of the irreducible representations of the {V}irasoro
  algebra and hidden symmetries of the finite {XXZ} {H}eisenberg chain. {A}
  {S}tory about moving and frozen energy levels}.
\newblock {\em Nucl. Phys. B}, 354:418--440, 1991.

\bibitem{Ritt1990}
U.~Grimm and V.~Rittenberg.
\newblock The modified {XXZ} {H}eisenberg chain, conformal invariance, surface
  exponents of $c<1$ systems, and hidden symmetries of the finite chains.
\newblock {\em International Journal of Modern Physics B}, 04(05):969–978,
  Apr 1990.

\bibitem{Dubail_2009}
J.~Dubail, J.~L. Jacobsen, and H.~Saleur.
\newblock Conformal two-boundary loop model on the annulus.
\newblock {\em Nuclear Physics B}, 813(3):430–459, Jun 2009.

\bibitem{magic}
J.~de~Gier, A.~Nichols, P.~Pyatov, and V.~Rittenberg.
\newblock Magic in the spectra of the {XXZ} quantum chain with boundaries at at
  {$\Delta=0$} and {$\Delta=-1/2$}.
\newblock {\em Nuclear Physics B}, 729(3):387–418, Nov 2005.

\bibitem{deG2005}
A.~Nichols, V.~Rittenberg, and J.~de~Gier.
\newblock One-boundary {T}emperley–{L}ieb algebras in the {XXZ} and loop
  models.
\newblock {\em Journal of Statistical Mechanics: Theory and Experiment},
  2005(03), Mar 2005.

\bibitem{Nichols22006}
A.~Nichols.
\newblock Structure of the two-boundary {XXZ} model with non-diagonal boundary
  terms.
\newblock {\em Journal of Statistical Mechanics: Theory and Experiment},
  2006(02), Feb 2006.

\bibitem{bk1993}
H.~J.~de Vega and A.~Gonzalez-Ruiz.
\newblock Boundary {K}-matrices for the six vertex and the $n(2n-1) {A}_{n-1}$
  vertex models.
\newblock {\em Journal of Physics A: Mathematical and General},
  26(12):519–524, Jun 1993.

\bibitem{bk1994}
H.~J.~de Vega and A.~Gonzalez-Ruiz.
\newblock Boundary ${K}$-matrices for the {XYZ}, {XXZ} and {XXX} spin chains.
\newblock {\em Journal of Physics A: Mathematical and General},
  27(18):6129–6137, Sep 1994.

\bibitem{Gier_2004}
J.~de~Gier and P.~Pyatov.
\newblock Bethe ansatz for the {T}emperley–{L}ieb loop model with open
  boundaries.
\newblock {\em Journal of Statistical Mechanics: Theory and Experiment},
  2004(03), Mar 2004.

\bibitem{cao2002exact}
J.~Cao, H.-Q. Lin, K.-J. Shi, and Y.~Wang.
\newblock Exact solution of {XXZ} spin chain with unparallel boundary fields.
\newblock {\em Nuclear Physics B}, 663(3):487--519, 2003.

\bibitem{nepomechie2002functional}
R.~I. Nepomechie.
\newblock Functional relations and {B}ethe {A}nsatz for the {XXZ} chain.
\newblock {\em Journal of Statistical Physics}, 111:1363--1376, 2002.

\bibitem{nep2003}
R.~I. Nepomechie.
\newblock Bethe ansatz solution of the open {XXZ} chain with nondiagonal
  boundary terms.
\newblock {\em Journal of Physics A: Mathematical and General},
  37(2):433–440, Dec 2003.

\bibitem{PBel}
S.~Belliard and R.~A. Pimenta.
\newblock {Modified algebraic Bethe ansatz for XXZ chain on the segment – II
  – general cases}.
\newblock {\em Nuclear Physics B}, 894:527--552, 2015.

\bibitem{Baseilhac_2007}
P.~Baseilhac and K.~Koizumi.
\newblock {Exact spectrum of the {XXZ} open spin chain from the $q$-Onsager
  algebra representation theory}.
\newblock {\em Journal of Statistical Mechanics: Theory and Experiment},
  2007(09):P09006--P09006, Sep 2007.

\bibitem{xx1999}
U.~Bilstein and B.~Wehefritz.
\newblock The {XX}-model with boundaries: {P}art {I}. {D}iagonalization of the
  finite chain.
\newblock {\em Journal of Physics A: Mathematical and General},
  32(2):191–233, Jan 1999.

\bibitem{xx2000}
U.~Bilstein.
\newblock The {XX} model with boundaries: {II}. finite-size scaling and
  partition functions.
\newblock {\em Journal of Physics A: Mathematical and General},
  33(24):4437–4449, Jun 2000.

\bibitem{blz1999}
V.~V. Bazhanov, S.~L. Lukyanov, and A.~B. Zamolodchikov.
\newblock On non-equilibrium states in {QFT} model with boundary interaction.
\newblock {\em Nuclear Physics B}, 549(3):529–545, Jun 1999.

\bibitem{BASEILHAC2003491}
P.~Baseilhac and K.~Koizumi.
\newblock {Sine-Gordon quantum field theory on the half-line with quantum
  boundary degrees of freedom}.
\newblock {\em Nuclear Physics B}, 649(3):491--510, 2003.

\bibitem{bas2003}
P.~Baseilhac and K.~Koizumi.
\newblock {$N=2$ boundary supersymmetry in integrable models and perturbed
  boundary conformal field theory}.
\newblock {\em Nuclear Physics B}, 669(3):417–434, Oct 2003.

\bibitem{saleur1998lectures}
H.~Saleur.
\newblock {Lectures on Non Perturbative Field Theory and Quantum Impurity
  Problems}, 1998, arXiv:9812110 [cond-mat].

\bibitem{eq2016}
D.~Bernard and B.~Doyon.
\newblock Conformal field theory out of equilibrium: a review.
\newblock {\em Journal of Statistical Mechanics: Theory and Experiment},
  2016(6):064005, Jun 2016.

\bibitem{Drinfeld:1985rx}
V.~G. Drinfeld.
\newblock {Hopf algebras and the quantum Yang-Baxter equation}.
\newblock {\em Sov. Math. Dokl.}, 32:254--258, 1985.

\bibitem{Jimbo:1985zk}
M.~Jimbo.
\newblock {A $q$-difference analog of $U(g)$ and the Yang-Baxter equation}.
\newblock {\em Lett. Math. Phys.}, 10:63--69, 1985.

\bibitem{qgroups}
V.~Chari and A.~Pressley.
\newblock {\em A Guide to Quantum Groups}.
\newblock Cambridge University Press, 1994.

\bibitem{kassel}
Ch. Kassel.
\newblock {\em Quantum Groups}.
\newblock Springer, New York, NY, 1995.

\bibitem{geermodtr0}
F.~Costantino, N.~Geer, and B.~Patureau-Mirand.
\newblock Some remarks on the unrolled quantum group of $sl(2)$.
\newblock {\em J. Pure Appl. Algebra}, Volume 219:pp. 3238--3262, 2015,
  arXiv:1406.0410 [math.QA].

\bibitem{humphreys}
J.~E. Humphreys.
\newblock {\em {Representations of semisimple Lie algebras in the BGG category
  O}}.
\newblock Graduate studies in mathematics: v. 94. American Mathematical
  Society, Providence, R.I, 2008.

\bibitem{Gainutdinov_2014}
A.~M. Gainutdinov, H.~Saleur, and I.~Yu. Tipunin.
\newblock Lattice {W}-algebras and logarithmic {CFT}s.
\newblock {\em Journal of Physics A: Mathematical and Theoretical},
  47(49):495401, Nov 2014.

\bibitem{Doikou_2003}
A.~Doikou and P.~P. Martin.
\newblock Hecke algebraic approach to the reflection equation for spin chains.
\newblock {\em Journal of Physics A: Mathematical and General},
  36(9):2203–2225, Feb 2003.

\bibitem{Gainutdinov_2013}
A.~M. Gainutdinov, J.~L. Jacobsen, H.~Saleur, and R.~Vasseur.
\newblock A physical approach to the classification of indecomposable
  {V}irasoro representations from the blob algebra.
\newblock {\em Nuclear Physics B}, 873(3):614–681, Aug 2013.

\bibitem{Friedan:1985ge}
D.~Friedan, E.~J. Martinec, and S.~H. Shenker.
\newblock {Conformal Invariance, Supersymmetry and String Theory}.
\newblock {\em Nucl. Phys. B}, 271:93--165, 1986.

\bibitem{LCFT2013}
T.~Creutzig and D.~Ridout.
\newblock Logarithmic conformal field theory: beyond an introduction.
\newblock {\em Journal of Physics A: Mathematical and Theoretical},
  46(49):494006, Nov 2013.

\bibitem{gl112006}
V.~Schomerus and H.~Saleur.
\newblock The {WZW}-model: From supergeometry to logarithmic {CFT}.
\newblock {\em Nuclear Physics B}, 734(3):221–245, Feb 2006.

\bibitem{gl112009}
T.~Creutzig and P.~B. Rønne.
\newblock The ${GL}(1|1)$-symplectic fermion correspondence.
\newblock {\em Nuclear Physics B}, 815(1-2):95–124, Jul 2009.

\bibitem{bc2006}
M.~R. Gaberdiel and I.~Runkel.
\newblock The logarithmic triplet theory with boundary.
\newblock {\em Journal of Physics A: Mathematical and General},
  39(47):14745–14779, Nov 2006.

\bibitem{Creutzig_2008}
T.~Creutzig, T.~Quella, and V.~Schomerus.
\newblock New boundary conditions for the $c=-2$ ghost system.
\newblock {\em Physical Review D}, 77(2), Jan 2008.

\bibitem{Cardy2004BoundaryCF}
J.~Cardy.
\newblock Boundary conformal field theory.
\newblock In Jean-Pierre Françoise, Gregory~L. Naber, and Tsou~Sheung Tsun,
  editors, {\em Encyclopedia of Mathematical Physics}, pages 333--340. Academic
  Press, Oxford, 2006.

\bibitem{TLalg}
H.~N.~V. Temperley and E.~H. Lieb.
\newblock {Relations between the 'Percolation' and 'Colouring' Problem and
  other Graph-Theoretical Problems Associated with Regular Planar Lattices:
  Some Exact Results for the 'Percolation' Problem}.
\newblock {\em Proceedings of the Royal Society of London. Series A,
  Mathematical and Physical Sciences}, 322(1549):251--280, 1971.

\bibitem{alma9912017293902959}
P.~P. Martin.
\newblock {\em Potts models and related problems in statistical mechanics},
  volume~5 of {\em Series on advances in statistical mechanics}.
\newblock World Scientific, 1991.

\bibitem{dubail:tel-00555624}
J.~Dubail.
\newblock {\em {Conditions aux bords dans des théories conformes non
  unitaires}}.
\newblock PhD thesis, {Universit{\'e} Paris Sud - Paris XI}, Sep 2010,
  https://tel.archives-ouvertes.fr/tel-00555624.

\bibitem{MARTIN2000957}
P.~P. Martin and D.~Woodcock.
\newblock On the structure of the blob algebra.
\newblock {\em Journal of Algebra}, 225(2):957--988, 2000.

\bibitem{iohara2019schurweyl}
K.~Iohara, G.~Lehrer, and R.~Zhang.
\newblock Schur-{W}eyl duality for certain infinite dimensional
  $\rm{U}_q(\mathfrak{sl}_2)$-modules, 2019, arXiv:1811.01325 [math.RT].

\bibitem{blobgen}
A.~Lacabanne, G.~Naisse, and P.~Vaz.
\newblock Tensor product categorifications, {V}erma modules and the blob
  2-category.
\newblock {\em Quantum Topology}, 12:705–812, 2021.

\bibitem{Skl83}
E.~K. Sklyanin.
\newblock {Some algebraic structures connected with the Yang-Baxter equation.
  Representations of quantum algebras}.
\newblock {\em Funktsional. Anal. i Prilozhen.}, 17(4):34--48, 1983.

\bibitem{Arnaudon:1992ig}
D.~Arnaudon.
\newblock {Composition of kinetic momenta: The $U_q sl(2)$ case}.
\newblock {\em Commun. Math. Phys.}, 159:175--194, 1994, arXiv:9212067
  [hep-th].

\bibitem{Drinfeld:1986in}
V.~G. Drinfeld.
\newblock {Quantum groups}.
\newblock {\em Zap. Nauchn. Semin.}, 155:18--49, 1986.

\bibitem{JACKSON20111689}
C.~Jackson and T.~Kerler.
\newblock The {L}awrence-{K}rammer-{B}igelow representations of the braid
  groups via ${U}_q(sl_2)$.
\newblock {\em Advances in Mathematics}, 228(3):1689--1717, 2011.

\bibitem{topdef}
J.~Belletête, A.~M. Gainutdinov, J.~L. Jacobsen, H.~Saleur, and T.~S. Tavares.
\newblock {Topological defects in lattice models and affine Temperley-Lieb
  algebra}, 2018, arXiv:1811.02551 [hep-th].

\end{thebibliography}

\end{document}